%% file: safe-sound.tex
\documentclass[acmsmall,screen]{acmart}

\usepackage{booktabs}   
\usepackage{subcaption} 
\usepackage{pdflscape}
\usepackage{enumitem}
\setlistdepth{9} 
\setlist[itemize,1]{label=$\bullet$}
\setlist[itemize,2]{label=$\bullet$}
\setlist[itemize,3]{label=$\bullet$}
\setlist[itemize,4]{label=$\bullet$}
\setlist[itemize,5]{label=$\bullet$}
\setlist[itemize,6]{label=$\bullet$}
\setlist[itemize,7]{label=$\bullet$}
\setlist[itemize,8]{label=$\bullet$}
\setlist[itemize,9]{label=$\bullet$}
\renewlist{itemize}{itemize}{9}
\usepackage{mathdesign}
\makeatletter
\def\arcr{\@arraycr}
\makeatother

\makeatletter\if@ACM@journal\makeatother
\acmJournal{PACMPL}
\acmVolume{1}
\acmNumber{POPL}
\acmArticle{1}
\acmYear{2021}
\acmMonth{1}
\acmDOI{10.1145/nnnnnnn.nnnnnnn}
\startPage{1}
\else\makeatother
\acmConference[PL'17]{ACM SIGPLAN Conference on Programming Languages}%
{January 01--03, 2017}%
{New York, NY, USA}
\acmYear{2017}
\acmISBN{978-x-xxxx-xxxx-x/YY/MM}
\acmDOI{}
\startPage{1}
\fi

\setcopyright{rightsretained} 

\setcitestyle{sort&compress}

\newcommand{\rparagraph}[1]{\paragraph{\hspace{-\parindent}\textbf{#1}}}

\input{definitions}

\newcommand{\x}{\mathsf{x}}
\newcommand{\y}{\mathsf{y}}
\newcommand{\f}{\mathsf{f}}
\newcommand{\q}{\mathsf{q}}
\newcommand{\bb}{\mathsf{hasM}}
\newcommand{\m}{\mathsf{m}}
\newcommand{\prodd}{\mathsf{sum}}

\newcommand{\CC}{\mathcal{C}}

\newcommand{\rx}{\mathbf{x}}
\newcommand{\rn}{\mathbf{n}}
\newcommand{\rb}{\mathbf{b}}
\newcommand{\rr}{\mathbf{e}}
\newcommand{\ru}{\mathbf{u}}
\newcommand{\rv}{\mathbf{v}}
\newcommand{\rl}{\mathbf{l}}
\newcommand{\wf}{\;\mathbf{wf}}
\renewcommand{\eve}{\ev\rr}  

\lstdefinelanguage{scala}{
  morekeywords={let,in,if,then,else,fun},
  otherkeywords={=>,!,:=,+},
  sensitive=true,
  morecomment=[l]{//},
  morecomment=[n]{/*}{*/},
  morestring=[b]",
  morestring=[b]',
  morestring=[b]""",
  escapeinside={(*}{*)},
  moredelim=**[is][{\btHL}]{`}{`},
  moredelim=**[is][\addeddenv]{@}{@}
}

\newcommand{\evidenceof}{\vdash} 
\usepackage{scalerel,stackengine}
\stackMath
\newcommand\reallywidehat[1]{%
  \savestack{\tmpbox}{\stretchto{%
      \scaleto{%
        \scalerel*[\widthof{\ensuremath{#1}}]{\kern-.6pt\bigwedge\kern-.6pt}%
        {\rule[-\textheight/2]{1ex}{\textheight}}
      }{\textheight}
    }{0.5ex}}%
  \stackon[1pt]{#1}{\tmpbox}%
}
\begin{document}

\title{Abstracting Gradual Typing Moving Forward: Precise and
  Space-Efficient (Technical Report)}

\author{Felipe Ba\~{n}ados~Schwerter}
\email{fbanados@cs.ubc.ca}
\affiliation{%
  \department[0]{Software Practices Lab} %
  \department[1]{Department of Computer Science} %
  \institution{University of British Columbia} %
  \streetaddress{201-2366 Main Mall} %
  \city{Vancouver} %
  \state{BC} %
  \postcode{V6T1Z4} %
  \country{Canada}}
\author{Alison M. Clark}
\email{alison.marie.clark@gmail.com}
\affiliation{%
}
\authornote{Work done while at UBC}
\author{Khurram A. Jafery}
\email{khurram.jafery@gmail.com}
\authornotemark[1]
\affiliation{%
  \institution{Amazon} %
}
\author{Ronald Garcia}
\email{rxg@cs.ubc.ca}
\affiliation{%
  \department[0]{Software Practices Lab} %
  \department[1]{Department of Computer Science} %
  \institution{University of British Columbia} %
  \streetaddress{201-2366 Main Mall} %
  \city{Vancouver} %
  \state{BC} %
  \postcode{V6T1Z4} %
  \country{Canada}}

\begin{abstract}
  Abstracting Gradual Typing (AGT) is a systematic approach to designing
  gradually-typed languages.  Languages developed using AGT automatically
  satisfy the formal \emph{semantic} criteria for gradual languages identified
  by~\citet{siek15criteria}.  Nonetheless, vanilla AGT semantics can still have
  important shortcomings.  First, a gradual language's runtime checks should
  preserve the \emph{space-efficiency} guarantees inherent to the underlying
  static and dynamic languages.  To the contrary, the default operational
  semantics of AGT break proper tail calls.
  Second, a gradual language's runtime checks should enforce basic modular
  type-based invariants expected from the static type discipline.  To the
  contrary, the default operational semantics of AGT may fail to enforce some
  invariants in surprising ways.  We demonstrate this in the
  $\GTFLcsub$ language of \citet{garciaAl:popl2016}.

  This paper addresses both problems at once by refining the theory underlying
  AGT's dynamic checks.  \citet{garciaAl:popl2016} observe that AGT involves
  \emph{two} abstractions of static types: one for the static semantics and one
  for the dynamic semantics.  We recast the latter as an abstract
  interpretation of \emph{subtyping} itself, while gradual types still abstract
  static types.  Then we show how
  \emph{forward-completeness}~\citep{giacobazzi01incomplete} is key to
  supporting both space-efficient execution and reliable\deleted{, predictable} runtime
  type enforcement.
\end{abstract}



\maketitle

\renewcommand{\shortauthors}{Ba\~{n}ados Schwerter et al.}
\listofchanges
\section{Introduction}
\label{sec:intro}

Gradual typing is an increasingly popular approach to designing programming
languages that seamlessly blend dynamic and static type checking.  Work in this
space has produced a variety of language models,
e.g.,~\cite{GradualTyping,STP,GradualTypingObjects,Wadler2009esop,%
  sergey12ownership,ina11generics}, evaluation criteria~\cite{siek15criteria},
and extensions of the concept to new contexts,
e.g.,~\cite{disney11flow,banados14effects}.

However, each excursion raises new questions about how language
designers can produce sensible gradually-typed counterparts to their chosen
static type discipline, and compare and evaluate alternative potential designs.
To address these problems, \citet{garciaAl:popl2016} proposed Abstracting
Gradual Typing (AGT), a methodology for developing gradually-typed extensions
of pre-existing static type disciplines.  Their approach systematizes the
construction of static and dynamic semantics for gradual languages that
\emph{by construction} admit straightforward proofs of broadly accepted formal
criteria for gradually typed languages~\cite{siek15criteria}.

Despite its metatheoretical gains, the AGT approach still requires extra work
and human creativity to achieve some operational and semantic goals for a
gradually-typed language.  First, AGT-induced semantics do not automatically
provide \emph{space-efficient} runtime checking.  Second, AGT's runtime
checking regime sometimes requires manual tuning to ensure that in addition to
Siek et al.'s criteria, desired and expected \emph{modular type-based semantic
  invariants} are properly enforced.  This paper addresses both of these
concerns, but first explains them in more detail.

\rparagraph{Space-Efficient Runtime Checking}
\citet{herman10space} showed that na\"ively implementing runtime checks using
wrappers can introduce space leaks and break proper tail calls.  For example,
consider the gradually-typed program:
\begin{small}
\begin{align*}
  \mathit{even} (n : \Int) : \? &= 
  \<if> (n = 0) \<then> \mathrm{True} \<else> \mathit{odd}(n - 1) \\
  \mathit{odd}  (n : \Int) : \Bool &= 
  \<if> (n = 0) \<then> \mathrm{False} \<else> \mathit{even}(n - 1)
\end{align*}
\end{small}
Some terms are ascribed the \emph{unknown type} $\?$, which marks absent type
information and enables runtime checking.  The program looks tail-recursive,
but is instrumented as follows\footnote{For clarity, these examples use a type
  cast notation with source and destination types.  Our formal semantics
  exhibit the same concepts, albeit using a necessarily less transparent
  representation.}:
\begin{small}
\begin{align*}
  \mathit{even} (n : \Int) : \? &= 
  \<if> (n = 0) \<then> \Gbox{\braket{\? \Leftarrow \Bool}} \mathrm{True}
  \<else> \Gbox{\braket{\? \Leftarrow \Bool}}\mathit{odd}(n - 1) \\
  \mathit{odd}  (n : \Int) : \Bool &= 
  \<if> (n = 0) \<then> \mathrm{False} \<else> \Gbox{\braket{\Bool \Leftarrow \?}}\mathit{even}(n - 1)
\end{align*}
\end{small}
In this program, the highlighted bracket expressions denote runtime type casts.
In \citet{GradualTyping} semantics, runtime casts accumulate at each recursive
call:
\begin{small}
  \begin{align*}
  \mathit{odd}(3) &--> 
  \Gbox{\braket{\Bool \Leftarrow \?}}\mathit{even}(2)  -->
  \Gbox{\braket{\Bool \Leftarrow \?}\braket{\? \Leftarrow \Bool}}\mathit{odd}(1)
  \\ &-->
  \Gbox{\braket{\Bool \Leftarrow \?}\braket{\? \Leftarrow \Bool}\braket{\Bool \Leftarrow \?}}\mathit{even}(0)
  \\ &-->
  \Gbox{\braket{\Bool \Leftarrow \?}\braket{\? \Leftarrow \Bool}\braket{\Bool \Leftarrow \?}\braket{\? \Leftarrow \Bool}}\mathrm{True}
  -->
  \Gbox{\braket{\Bool}\braket{\?}}\mathrm{True} -->
  \mathrm{True}
\end{align*}
\end{small}
This means that programs that appear to be tail-recursive, and thus consume
constant stack space, now consume linear space, simply because of how their
types are checked.
A variant of the same program written in continuation-passing-style
demonstrates how (higher-order) casts, which cannot be resolved immediately,
accumulate around values if one is not careful.
\begin{small}
  \begin{align*}
  \mathit{evenk}\;(n : \Int)\;(k : \? -> \?) : \Bool  &= 
  \<if> (n = 0) \<then> (k\;\mathrm{True}) \<else> \mathit{oddk}\;(n - 1)\;k \\
  \mathit{oddk}\;(n : \Int)\;(k : \Bool -> \Bool) : \Bool: \Bool &= 
  \<if> (n = 0) \<then> (k\;\mathrm{False}) \<else> \mathit{evenk}\;(n - 1)\;k
\end{align*}
\end{small}
This program elaborates to the following:
\begin{small}
  \begin{align*}
  \mathit{evenk}\;&(n : \Int)\;(k : \? -> \?) : \Bool  = \\
  &\<if> (n = 0) \<then> (k\;\Gbox{\braket{\? \Leftarrow \Bool}}\mathrm{True}) \<else> \mathit{oddk}\;(n - 1)\;(\Gbox{\braket{\Bool -> \Bool \Leftarrow \? -> \?}}k) \\
  \mathit{oddk}\;&(n : \Int)\;(k : \Bool -> \Bool) : \Bool = \\
  &\<if> (n = 0) \<then> (k\;\mathrm{False}) \<else> \mathit{evenk}\;(n - 1)\;(\Gbox{\braket{\? -> \? \Leftarrow \Bool -> \Bool}}k)
\end{align*}
\end{small}
and its evaluation also accumulates casts:
\begin{small}
\begin{align*}
  \mathit{odd}(3)\;k &--> 
  \mathit{even}(2)\;\Gbox{\braket{\?->\? \Leftarrow \Bool -> \Bool}}k  \\ &-->
  \mathit{odd}(1)\;\Gbox{\braket{\Bool->\Bool \Leftarrow \? -> \?}\braket{\?->\? \Leftarrow \Bool -> \Bool}}k  \\ &-->
  \mathit{even}(0)\;\Gbox{\braket{\?->\? \Leftarrow \Bool -> \Bool}\braket{\Bool->\Bool \Leftarrow \? -> \?}\braket{\?->\? \Leftarrow \Bool -> \Bool}}k
  --> \dots
\end{align*}
\end{small}
The solution to this problem is to aggressively compose casts, in tail position
and around values respectively, rather than allowing them to accumulate.
However, doing so requires some notion of bounded-space cast-composition whose
behavior is equivalent to the standard semantics.
Many gradual language semantics use coercions~\citep{henglein94dynamic} or
threesomes~\citep{siek10threesomes} to coalesce casts as they arise.  At first
glance, AGT appears to immediately support space efficiency, since it provides
an operator for composing checks.  However, this only suffices to prevent
accumulation on values.  To see why, consider how
space-efficient tail-recursion must proceed:
\footnote{For simplicity this example suppresses some details: see
  \citet{herman10space}.}
\begin{small}
\begin{align*}
  \mathit{odd}(3) --> 
  \Gbox{\braket{\Bool \Leftarrow \?}}\mathit{even}(2)  -->
  \Gbox{\braket{\Bool \Leftarrow \?}\braket{\? \Leftarrow \Bool}}\mathit{odd}(1)  -->
  \Gbox{\braket{\Bool \Leftarrow \Bool}}\mathit{odd}(1)    \\ -->
  \Gbox{\braket{\Bool \Leftarrow \Bool}\braket{\Bool \Leftarrow \?}}\mathit{even}(0)  -->
  \Gbox{\braket{\Bool \Leftarrow \?}}\mathit{even}(0)
  -->
  \Gbox{\braket{\Bool \Leftarrow \Bool}}\mathrm{True}
  -->
  \mathrm{True}
\end{align*}
\end{small}
To preserve proper tail calls the semantics must aggressively compose
tail-position checks as they arise, rather than waiting until a final tail call
returns a value before resolving checks up the call stack, as in the first
trace.
This only works if composing checks is equivalent in either direction: cast
composition must be associative.  Unfortunately this is not always the case: in
particular, the checks introduced by~\citet{garciaAl:popl2016} for $\GTFLcsub$,
a gradual language with record subtyping, do not compose associatively.  Some
cast sequences succeed when composed in one order, but fail when composed in
the other. \added{We show a detailed example in Sec.~\ref{ssec:shortcomings}.}

\rparagraph{Type-Based Reasoning}
Prior work on gradualizing sophisticated typing disciplines like
information-flow security~\citep{toro18secref} or
parametricity~\replaced{(e.g., \citet{10.1145/3110283,toro19parametricity,10.1145/3371114})}{\textbackslash
  citep\{toro19parametricity\}} has demonstrated that AGT applied
na\"ively does not automatically enforce the modular invariants that a
programmer may expect from the intended static type discipline.  Unfortunately
the same phenomenon arises in $\GTFLcsub$, even though its type
discipline---record subtyping---is comparatively simple and well-understood.
\footnote{Formally speaking, information-flow and parametricity are
  \emph{hyperproperties}\added{: relationships among
    \textbf{multiple runs} of a program}, whereas record subtyping is not.}

The following scenario demonstrates the problem.  Consider the following
program fragment:
\begin{displaymath}
  \<let> \q : \rec{\x:\Int} = \rec{\x = 5, \y = \ttt}
  \<in> \mathbf{\langle body\hspace{-1pt}\rangle}
\end{displaymath}
According to standard subtype-based reasoning, the body cannot access the $\y$
field of the record.  Such modular reasoning is a hallmark of static typing,
and programmers who mix static and dynamic typing want to reason about their
code using static types where available~\citep{tunnellwilson18study}.
Unfortunately, the following completed $\GTFLcsub$ program compiles and runs
successfully:
\begin{displaymath}
  \<let> \q : \rec{\x:\Int} = \rec{\x = 5, \y = \ttt}
  \<in> (\q :: \? :: \rec{\x:\Int,\y:\Bool}).\y 
\end{displaymath}
In essence, casting $\q$ to the unknown type and then back to a record type
exposes the extra field that should have been hidden by $\q$'s assumed type.
In other words, the language satisfies standard type safety, as AGT does by
construction, but its runtime treatment of type abstractions does not strictly
enforce the type abstraction properties one expects from subtyping.

\rparagraph{Forward-completeness is the key}  
Space efficiency and type-based enforcement problems have the same
source: an insufficiently-precise runtime type enforcement mechanism.
Previous AGT-based systems have had to tune their runtime
abstractions.  \citet{toro18secref} tuned their AGT abstraction to
achieve type enforcement.  Guided by intuition, \citet{TORO2020102496}
tuned their abstraction to achieve a space-efficient implementation.
This work develops a principled approach to tuning the runtime checks in an
AGT-based language, and shows that it suffices to address the two challenges
above.  

\pagebreak
Our \textbf{contributions} are as follows:
\begin{itemize}
\item We revisit the foundations of AGT's dynamic checks and refine their
  conception.  \citet{garciaAl:popl2016} conceive of them via post hoc
  manipulation of abstractions for tuples of types: we redefine them as a
  direct abstract interpretation of \emph{the subtyping relation itself}.  This
  reframing is subtle and equivalent to the original, but now techniques from
  the abstract interpretation literature become immediately applicable
  (Sec.~\ref{sec:agt}).
  
\item AGT, applied na\"ively, does not support space-efficiency because
  composing checks out of order may not preserve behaviour.  We explore this
  issue in the context of \citet{garciaAl:popl2016}'s $\GTFLcsub$ language, and
  introduce \emph{bounded records and rows}, a new abstraction for runtime
  evidence that retains precise information about gradual types and how they
  interact at runtime.  This precise representation admits associative cast
  composition.  Bounded rows uncover a subtle interplay between static record
  subtyping and \emph{gradual row types}.

\item As others observed, AGT-designed semantics do not automatically ensure
  all desired type-based reasoning principles.  This problem appeared in prior
  work on sophisticated typing disciplines that enforce properties
  fundamentally about multiple runs of a program (i.e., hyperproperties).
  Subtyping, in contrast, is a standard unary type property, yet still requires
  care to enforce when gradual types are sufficiently rich.  Ours is the first
  work to witness this phenomenon in a simple setting.
  
\item Since AGT is a framework for designing languages, it is desirable to
  frame these improvements in a general form.  To do so, we devise sufficient
  conditions for space-efficiency and \deleted{predictable} precise runtime monitoring.
  Both properties can be achieved, while satisfying the formal criteria
  of~\citet{siek15criteria}, by ensuring that the dynamic
  monitoring semantics is \emph{forward complete}, a concept from 
  abstract interpretation~\citep{giacobazzi01incomplete}.
  Forward completeness guarantees that an AGT semantics can be made
  space-efficient, and equates reasoning about the dynamic semantics to
  \replaced{reasoning up to precision $(\sqsubseteq)$}{imprecisely reasoning} about static subtyping.
  
\end{itemize}

\begin{figure}[]
  \begin{small}
    \flushleft\textbf{Syntax}
    \begin{displaymath}
      \begin{array}{rcll}
        \multicolumn{4}{c}{
          x \in \Var,\quad
          b \in \oblset{Bool},\quad
          n \in \Z,\quad
          l \in \Label,\quad
          \Gamma \in \Env = \Var \finto \Type
        } \\[2mm]
        T \in \Type & ::= & \Int | \Bool | T -> T | \rec{\overline{l:T}} &
        \text{(static types)} \\
        \cT \in \GType & ::= & \? | \Int | \Bool | \cT -> \cT 
        | \rec{\overline{l:\cT}} | \rec{\overline{l:\cT},\?} &
        \text{(gradual types)} \\
        t \in \Term & ::= & n | b | x | \lambda x:\cT.t |  t\;t | t + t
        | \<if> t \<then> t \<else> t | \rec{\overline{l=t}} | t.l
        | t :: \cT & \text{(terms)} \\
      \end{array}
    \end{displaymath}
    
    \flushleft\boxed{\Gamma |- t : \cT}\quad\textbf{Gradual Typing}
    \begin{mathpar}
      \inference[($\cT$x)]{x:\cT\in\Gamma}{\Gamma |- x : \cT}
      \and
      \inference[($\cT$n)]{}{\Gamma |- n : \Int}
      \and
      \inference[($\cT$b)]{}{\Gamma |- b : \Bool}
      \and
      \inference[($\cT$app)]{
        \overline{\Gamma |- t_i : \cT_i} & 
        \cT_2 \csub \cdom(\cT_1)
      }{
        \Gamma |- t_1\;t_2 : \ccod(\cT_1)
      }
      \and
      \inference[($\cT$+)]{
        \overline{\Gamma |- t_i : \cT_i} & 
        \overline{\cT_i \csub \Int}
      }{
        \Gamma |- t_1 + t_2 : \Int
      }
      \and
      \inference[($\cT$if)]{
        \overline{\Gamma |- t_i : \cT_i} & 
        \cT_1 \csub \Bool 
      }{
        \Gamma |- \<if> t_1 \<then> t_2 \<else> t_3 : \cT_2 \csubjoin \cT_3
      }  
      \and
      \inference[($\cT$proj)]{
        \Gamma |- t : \cT
      }{
        \Gamma |- t.l : \cproj(\cT,l) 
      }
      \and
      \inference[($\cT\lambda$)]{
        \Gamma, x:\cT_1 |- t : \cT_2
      }{
        \Gamma |- (\lambda x : \cT_1.t) : \cT_1 -> \cT_2
      }
      \and
      \inference[($\cT\!\!::$)]{
        \Gamma |- t : \cT &
        \cT \csub \cT_1
      }{
        \Gamma |- (t :: \cT_1) : \cT_1
      }  
      \and
      \inference[($\cT$rec)]{
        \overline{\Gamma |- t_i : \cT_i}
      }{
        \Gamma |- \rec{\overline{l_i = t_i}} : \rec{\overline{l_i:\cT_i}}
      }
    \end{mathpar}

    \flushleft\textbf{Helper Functions}
    \begin{displaymath}
      \begin{block}
        \cdom : \GType \rightharpoonup \GType \\
        \cdom(\cT_1 -> \cT_2) = \cT_1\\
        \cdom(\?) = \? \\
        \cdom(\cT)\text{ undefined} \text{ otherwise }
      \end{block}  
      \qquad
      \begin{block}
        \ccod : \GType \rightharpoonup \GType \\
        \ccod(\cT_1 -> \cT_2) = \cT_2\\
        \ccod(\?) = \? \\
        \ccod(\cT)\text{ undefined} \text{ otherwise }
      \end{block}  
      \qquad
      \begin{block}
        \cproj : \GType \times \Label \rightharpoonup \GType \\
        \cproj(\rec{l:\cT,\overline{l_i:\cT_i},*},l) = \cT \\
        \cproj(\rec{\overline{l_i:\cT_i},\?},l) = \? 
        \quad\text{if} \quad l \notin \set{\overline{l_i}} \\
        \cproj(\?,l) = \? \\[0.5em]
        \cproj(\cT,l)\text{ undefined otherwise} \\
      \end{block}
    \end{displaymath}

    \flushleft\boxed{\cT \csub \cT}\quad\textbf{Consistent Subtyping}
    \begin{mathpar}
      \inference{}{\? \csub \cT}
      \quad
      \inference{}{\cT \csub \?}
      \quad
      \inference{}{\Int \csub \Int}
      \quad
      \inference{}{\Bool \csub \Bool}
      \and
      \inference{
        \cT_{21} \csub  \cT_{11} &
        \cT_{12} \csub  \cT_{22} 
      }{
        \cT_{11} -> \cT_{12} \csub \cT_{21} -> \cT_{22}
      }
      \and
      \inference{
        \overline{\cT_{i1} \csub \cT_{i2}}
      }{
        \rec{\overline{l_i:\cT_{i1}},\overline{l_j:\cT_j}} \csub
        \rec{\overline{l_i:\cT_{i2}},*}
      }
      \and
      \inference{
        \overline{\cT_{i1} \csub \cT_{i2}}
      }{
        \rec{\overline{l_i:\cT_{i1}},\overline{l_j:\cT_j},\?} \csub
        \rec{\overline{l_i:\cT_{i2}},\overline{l_k:\cT_{k}},*}
      }
    \end{mathpar}
    
    \flushleft\boxed{\cT \gprec \cT}\quad\textbf{Precision}
    \begin{mathpar}
      \inference{\cT \in \set{\Int,\Bool}}{\cT \gprec \cT}
      \quad
      \inference{
        \cT_{11} \gprec \cT_{21} &
        \cT_{12}  \gprec \cT_{22}
      }{
        \cT_{11} -> \cT_{12} \gprec \cT_{21} -> \cT_{22}
      }
      \quad
      \inference{
        \overline{\cT_{1i} \gprec \cT_{2i}}
      }{
        \rec{\overline{l_i : \cT_{1i}}} \gprec \rec{\overline{l_i : \cT_{2i}}}
      }
      \quad
      \inference{
        \overline{\cT_{1i} \gprec \cT_{2i}}
      }{
        \rec{\overline{l_i : \cT_{1i}},\overline{l_j : \cT_{1j}},*} \gprec
        \rec{\overline{l_i : \cT_{2i}},\?}
      }
      \and
      \inference{}{\cT  \gprec \?}
    \end{mathpar}
  \end{small}%
  \caption{$\GTFLcsub$: Static Semantics}
  \label{fig:gtfl-statics}
\end{figure}

\section{$\GTFLcsub$: A Gradually Typed Language}
\label{sec:gtflcsub}
\mbox{} As a concrete starting point for our investigations, this section
presents the semantics of $\GTFLcsub$, a gradually-typed language with records
and subtyping that also supports migration between dynamic and static type
checking.  \citet{garciaAl:popl2016} developed this language using the AGT
methodology, and its semantics exhibits the shortcomings that this paper
addresses.

We first present $\GTFLcsub$ with little reference to the AGT machinery used to
construct and justify it.  Furthermore, we state correctness properties here
without proof, because AGT exploits calculational abstract interpretation
techniques~\cite{AbstractInterpretation} to intertwine the proof and definition
processes, making the design ``correct by construction.''
Sec.~\ref{sec:agt} connects $\GTFLcsub$'s design to AGT, and briefly connects
the correctness properties to proofs.  Ultimately our improvements are
presented in terms of AGT to ensure that they generalize across AGT-based
languages.

\subsection{Syntax and Typing}

Fig.~\ref{fig:gtfl-statics} presents the $\GTFLcsub$ syntax and type system.
Its terms are typical: numbers, booleans, functions, records,
and type ascriptions.  All of the novelty lies in its type structure, where
common static types---atomics, functions, and records---are augmented with two
gradual type constructs that denote imprecise type information.  The
now-standard \emph{unknown type} $\?$ denotes the complete omission of type
information~\cite{GradualTyping}.  The \emph{gradual row type}
$\rec{\overline{\lx : \cT},\?}$, on the other hand, represents a
record type with incomplete \emph{field} information.  It surely
constrains the list $\overline{l}$ of fields with corresponding (gradual) types
$\overline{\cT}$, but the \emph{gradual row designator} $\?$
denotes the possibility of additional statically unknown fields.%
\footnote{Throughout we use overlines to denote zero or more repetitions,
  $+$-annotated overlines to denote one or more repetitions, and
  $\rec{\overline{l:\cT},*}$ to simultaneously denote both traditional record
  types $\rec{\overline{l:\cT}}$ and gradual row types
  $\rec{\overline{l:\cT},\?}$.}
Gradual rows are somewhat analogous to polymorphic rows~\cite{remy89records},
except that their presence induces dynamic checks.  As such, a gradual row type
is only partially static with respect to record fields.

As is now standard for gradual typing, the concept of ``imprecise type
information'' is formalized using the \emph{precision} judgment
$S_1 \gprec S_2$, which says that $S_1$ is less imprecise than $S_2$.  Static
types are the \emph{least} elements of this partial order, while the greatest
element is the unknown type $\?$.

The rules for the typing judgment $\Gamma |- t : \cT$ are structured in the
style of~\citet{garcia15pts}: each typing judgment in the premise of a rule has
an arbitrary type metavariable $S_i$, but these premise types are constrained
by side conditions.  The result type in the conclusion of a rule is either a
particular type or a (partial) function of its premise types.  For example, the
two premises of the (Sapp) rule have types $\cT_1$ and $\cT_2$ respectively;
these premise types are constrained by the \emph{consistent subtyping} side
condition ${\cT_2 \csub \cdom(\cT_1)}$;
and the result type (Sapp) is the \emph{gradual codomain} $\ccod(\cT_1)$ of the
operator type.  This structure keeps the typing rules syntax-directed, while
subsuming the type system for the corresponding statically typed
language $\STFLsub$ \added{(introduced by~\citet{garciaAl:popl2016})}.  Consider, for instance, how $\ccod(\cT_1)$ extends the
idea of ``the codomain of a function type''.  For function types
$S_1 -> S_2$, its behaviour is as expected, but for the unknown type $\?$, the
codomain is completely unknown because the operator's type is completely
unknown: the operator may have a function type, but we do not know for sure.
We formalize this behaviour as a correctness criterion. 
\begin{definition}[Candidate Codomain]
  \label{def:ccod-correct}
  \mbox{}
  \begin{enumerate}
  \item A \emph{plausible function type} is a gradual type $\cT$ such that
    $T_1 -> T_2 \gprec \cT$ for some $T_1,T_2\in\Type$;
  \item A gradual type $\cT'$ is a \emph{candidate codomain} of a plausible
    function type $\cT$ if $T_1 -> T_2 \gprec \cT$ implies $T_2 \gprec \cT'$,
    for all $T_1,T_2\in\Type$.
    \end{enumerate}
  \end{definition}

\begin{proposition}[Codomain Correctness]
  \label{prop:ccod-correct}
  \mbox{} A gradual type $\cT$ has a candidate codomain if and only if it has a
  least (with respect to $\gprec$) candidate codomain, denoted $\ccod(\cT)$.
\end{proposition}
Prop.~\ref{prop:ccod-correct} implicitly defines $\ccod$ as a partial
function on $\GType$, which is defined for exactly the plausible function
types.
The notion of ``candidate codomain'' can be interpreted as a \emph{soundness}
property of $\ccod$: it broadly approximates the idea of ``codomain'' even in
the face of imprecision.  The proposition can then be interpreted as an
\emph{optimality} property: there is a ``best'' candidate codomain, if any, so
$\ccod$ never loses precision needlessly.  Analogous correctness criteria apply
to the other type operators, completely characterizing the gradual versions in
terms of their static counterparts.

The \emph{consistent subtyping} relation $\csub$
\added{\cite{GradualTypingObjects}} extends static subtyping $<:$
to optimistically account for imprecision in gradual types.  In essence,
$S_1 \csub S_2$ means that it is \emph{plausible} that $S_1$ is a subtype of
$S_2$, when the imprecision of gradual types is taken into account.
For instance, $\?$ is both a consistent supertype and consistent subtype of
each gradual type $S$, because it could represent any static type whatsoever,
including $S$ itself.
As above, we formally relate consistent subtyping to static subtyping
\footnote{\added{For $\STFLsub$, static subtyping $\sub$ is
    standard relation induced by width and depth record subtyping~\cite{tapl}}} .
\begin{proposition}
  \label{prop:csub-correct}
  $S_1 \csub S_2$ if and only if ${T_1 \sub T_2}$ for some $T_1 \gprec S_1$ and
  $T_2 \gprec S_2$.
\end{proposition}

One key benefit to defining gradual operators and relations using static
counterparts and gradual type precision is that the resulting language
naturally satisfies static criteria for gradual typing set forth
by~\citet{siek15criteria}.  First, the $\STFLsub$ type system can be recovered
from that of $\GTFLcsub$ by simply restricting source programs to only mention
static types $T$.  Doing so: restricts the $\cdom$, $\ccod$, and $\cproj$
partial functions to simple arrow type and record type destructors; restricts
the \emph{consistent subtyping} relation $\csub$ to a typical definition of
static subtyping $<:$; and restricts the $\csubjoin$ partial function to the
\emph{subtype join} partial function $\subjoin$, which yields the least upper
bound of two static types (if there is one) according to static subtyping $<:$.
Thus, by construction, $\GTFLcsub$ conservatively extends the static language.
Furthermore, the \emph{static gradual guarantee}, which ensures that increasing
the precision of a program's types cannot fix extant type errors, follows
\deleted{straightforwardly} from the correctness criteria for each operator and
relation\added{, each of which monotonically preserves this
  property, yielding a direct compositional proof.}

\added{Though the syntax of gradual rows is simple, its implications
  for the language semantics are nontrivial. Gradual rows expose a
  subtle interplay between gradual type precision and static width
  subtyping that does not arise in most gradual type systems developed
  to date.}  \added{Consider
  $\rec{x : \Int, y : \Bool} \csub{} \rec{x:\Int,\?}$: This judgement
  is justified by two static subtypings:
  ${\rec{x:\Int,y:\Bool} <: \rec{x:\Int}}$ and
  ${\rec{x:\Int,y:\Bool} <: \rec{x:\Int, y:\Bool}}$.  If this instance
  of consistent subtyping is viewed as a form of coercion, then it
  indicates \emph{two} different behaviours: in the first the $y$
  field is \emph{obscured} via static subtyping.  In the second, the y
  field is obscured via precision: not by static subtyping, but by
  gradual typing.  The static \emph{and} runtime semantics must reckon
  with these two \emph{different} explanations simultaneously, and all
  outcomes must be consistent with one, the other, or both.  This
  happens in no other language featuring consistent subtyping, and can
  be improved upon with a more precise abstraction. The abstraction of
  \citet{garciaAl:popl2016} sometimes loses information at runtime
  that was obscured due to loss of precision
  (See Sec.~\ref{ssec:example-gtfl}).  By contrast, the BRR abstraction
  introduced in Sec.~\ref{sec:brr} preserves all obscured information
  at runtime, but hides some of it from the programmer: the
  abstraction always reflects instances of static subtyping.
}

\subsubsection{An Example} 
\label{ssec:example-gtfl}

The following (somewhat contrived) example program demonstrates some of the
features and intended capabilities of $\GTFLcsub$, especially the semantics of
gradual rows.  For succinctness, we assume $\<let>\!\!$ binding, which can be
easily added to the language.
\begin{displaymath}
  \begin{block}
    \<let> \prodd =
    \begin{block}\lambda (\bb:\Bool).\,\lambda (\x:\rec{\f:\Int,\?}).\\
      \quad\<if> \bb \<then> \x.\f + \x.\m \<else> \x.\f + \x.\q
    \end{block}\\
    \<inn>
    (\prodd\;\ttt\;\rec{\f=6,\m=2}) +
    (\prodd\;\fff\;\rec{\f=6,\q=2})
  \end{block}
\end{displaymath}
The $\prodd$ function takes a record $\f$, and a Boolean value $\bb$, and uses
$\bb$ to determine which field to add to $\x.\f$.  The gradual row type
ascribed to $\x$ ensures statically that the record argument contains an $\f$
field of type $\Int$, but makes no commitment regarding what other fields may
be present.  Thus the body of the function type checks despite referring to
$\x.\m$ and $\x.\q$.

The program successfully calculates the $\prodd$ for two different records that
have different fields.  The $\bb$ argument indicates which branch of execution
is the right one.  This small program demonstrates the possibility of at least
partially statically checking the program, while deferring checks for extra
fields to runtime.  We can extrapolate from this program to a larger
application that fruitfully exploits field dynamism, while statically checking
stable record components.

This language design makes the possibility of dynamic checking evident in the
types.  The gradual type $\rec{\f:\Int,\?}$ indicates that the
record may have more fields, which might be used by the program.

Furthermore, the type structure of a program can be exploited to control the
amount and scope of necessary dynamic checking.  For instance, if the branches
of the conditional made repeated use of extra fields, then each access would
require a runtime check.  However, we can ensure that these are checked
statically (and centralize the type assumptions) using a type ascription.  For
example, we can replace the consequent branch
\begin{math}
  \x.\f + \x.\m
\end{math}
with
\begin{displaymath}
  \<let> \y = (\x :: \rec{\f:\Int,\m:\Int}) \<in> \y.\f + \y.\m
\end{displaymath}
This change would move the initial field check (for $\x.\m$) and field type
check (that its type is $\Int$) to the ascription: the body of the consequent
$\<let>\!\!$ would be fully statically checked.
This use of ascription is similar to downcasting in object-oriented languages,
but is justified by the presence of the imprecise gradual row in the
type of $\prodd$'s argument.

Now for some bad news.  Unfortunately, $\GTFLcsub$ as designed by
\citet{garciaAl:popl2016} does not enforce its type abstractions to
the extent that one might expect.  Sound gradual type regimes
typically have the property that precise type information is
persistent \comment{Examples? Citations?}.
For instance, if we replace the record $\rec{\f=6,\m=2}$ in the example with
$(\rec{\f=6,\m=2} :: \rec{\f:\Int})$, then one would expect the program to type
check, but fail at runtime.  In this way, the language would guarantee that the
$m$ field's very existence is encapsulated: no client would be able to violate
that guarantee.

To our surprise, however, $\GTFLcsub$ executes this program successfully, for
essentially the same reason that the errant example from Sec.~\ref{sec:intro}
succeeds.  We consider this a flaw in the $\GTFLcsub$ semantics: type
abstractions should be respected.  Bear in mind that a language designer may
desire a gradual language that treats all gradual record types as gradual rows,
thereby baking downcasting into the semantics.  This can easily be supported as
syntactic sugar over the $\GTFLcsub$ surface syntax, rather than as an accident
of AGT-induced semantics.  In short, the semantics of an AGT-induced language
should \deleted{be predictable and should} precisely enforce types.  In the following, we
diagnose this failure of type abstraction and show how to properly ensure this
with a refinement to the AGT framework.

\subsection{Runtime Language}
\label{ssec:runtime-language} 

\begin{figure}
  \begin{small}
    \flushleft\textbf{Syntax}
    \begin{displaymath}
      \begin{array}{rcll}
        \multicolumn{4}{c}{
        } \\[2mm]
        \ev \in \Ev{} & = & 
        \set{\braket{\cT_1,\cT_2} | |- \braket{\cT_1,\cT_2} \wf} &
        \text{(evidence objects)} \\
        \rr \in \RTerm & ::= & \rn | \rb | \rx | \lambda \rx.\rr | \eve\;\eve 
        | \eve + \eve | \<if> \eve \<then> \eve \<else> \eve
        | \rec{\overline{\rl=\rr}} | \eve.l | \eve &
        \text{(runtime terms)} \\
        \ru \in \RawValue & ::= & \rn | \rb | \rx | \lambda \rx.\rr 
        | \rec{\overline{\rl=\rv}} &
        \text{(raw values)} \\
        \rv \in \Value & ::= & \ru | \cast{\ev}{\ru} & 
        \text{(values)}
      \end{array}
    \end{displaymath}

    \flushleft\boxed{\rr \derivof \Gamma |- t : \cT}
    \quad\textbf{Runtime Typing}
    \begin{mathpar}
      \inference[($\cT$x)]{
        x:\Gbox{\cT'}\in\Gamma
      }{
        \rx \derivof \Gamma |- x : \Gbox{\cT'}
      }
      \and
      \inference[($\cT$n)]{}{\rn \derivof \Gamma |- n : \Int}
      \and
      \inference[($\cT$b)]{}{\rb \derivof \Gamma |- b : \Bool}
      \and
      \inference[($\cT$app)]{
        \overline{\rr_i \derivof \Gamma |- t_i : \cT_i} & 
        \ev_1 \evidenceof \cT_1 \csub \Gbox{\cT'_1 -> \cT'_2} & 
        \ev_2 \evidenceof \cT_2 \csub \Gbox{\cT'_1} 
      }{
        (\cast{\ev_1}{\rr_1}\;\cast{\ev_2}{\rr_2}) \derivof
        \Gamma |- t_1\;t_2 : \Gbox{\cT'_2} 
      }
      \and
      \inference[($\cT$+)]{
        \overline{\rr_i \derivof \Gamma |- t_i : \cT_i} & 
        \overline{\ev_i \evidenceof \cT_i \csub {\Int}} 
      }{
        (\cast{\ev_1}{\rr_1} + \cast{\ev_2}{\rr_2}) \derivof
        \Gamma |- t_1 + t_2 : {\Int}
      }
      \and
      \inference[($\cT$if)]{
        \overline{\rr_i \derivof \Gamma |- t_i : \cT_i} & 
        \ev_1 \evidenceof \cT_1 \csub \Bool &
        \ev_2 \evidenceof \cT_2 \csub {\cT_2 \csubjoin \cT_3} &
        \ev_3 \evidenceof \cT_3 \csub {\cT_2 \csubjoin \cT_3} &
      }{
        (\<if> \cast{\ev_1}{\rr_1} \<then> \cast{\ev_2}{\rr_2}
        \<else> \cast{\ev_3}{\rr_3})
        \derivof \Gamma |- \<if> t_1 \<then> t_2 \<else> t_3 : 
        {\cT_2 \csubjoin \cT_3}
      }
      \and
      \inference[($\cT$proj)]{
        \rr \derivof \Gamma |- t : \cT &
        \ev \evidenceof \cT \csub \rec{l:\Gbox{\cT'}}
      }{
        \cast{\ev}{\rr}.l \derivof \Gamma |- t.l : \Gbox{\cT'}
      }
      \and
      \inference[($\cT\lambda$)]{
        \rr \derivof \Gamma, x:\cT_1 |- t : \cT_2
      }{
        (\lambda x.\rr) \derivof \Gamma |- (\lambda x : \cT_1.t) : \cT_1 -> \cT_2
      }
      \and
      \inference[($\cT\!\!::$)]{
        \rr \derivof \Gamma |- t : \cT &
        \ev \evidenceof \cT \csub \Gbox{\cT'}
      }{
        \cast{\ev}{\rr} \derivof \Gamma |- (t :: \Gbox{\cT'}) : \Gbox{\cT'}
      }  
      \and
      \inference[($\cT$rec)]{
        \overline{\rr_i \derivof \Gamma |- t_i : \cT_i}
      }{
        \rec{\overline{\rl_i = \rr_i}} \derivof
        \Gamma |- \rec{\overline{l_i = t_i}} : \rec{\overline{l_i:\cT_i}}
      }
    \end{mathpar}
    \flushleft\boxed{\ev \evidenceof \cT_1 \csub \cT_2}\quad
    \textbf{Evidence for Consistent Subtyping}
    \begin{mathpar}
      \inference{
        |- \braket{\cT'_1,\cT'_2} \wf &
        \cT'_1 \gprec \cT_1 &
        \cT'_2 \gprec \cT_2 &
      }{
        \braket{\cT'_1,\cT'_2} \evidenceof \cT_1 \csub \cT_2
      }
    \end{mathpar}
    \flushleft\boxed{|- \ev \wf}
    \quad\textbf{Well-formed Evidence}
    \begin{mathpar}
      \inference{\cT \in \set{\Int,\Bool,\?}}{
        |- \braket{\cT,\cT} \wf
      }
      \and
      \inference{
        |- \braket{\cT_{21},\cT_{11}} \wf &
        |- \braket{\cT_{12},\cT_{22}} \wf
      }{
        |- \braket{\cT_{11} -> \cT_{12},\cT_{21} -> \cT_{22}} \wf
      }
      \and
      \inference{
        \overline{|- \braket{\cT_{i1},\cT_{i2}} \wf} &
        \braket{*_1,*_2} \neq \braket{\replaced{\ensuremath
            \cdot}{\ensuremath \emptyset}`,\?}
      }{
        |- \braket{\rec{\overline{l_i : \cT_{i1}},*_1},
          \rec{\overline{l_i : \cT_{i2}},*_2}} \wf
      }
      \and
      \inference{
        \overline{|- \braket{\cT_{i1},\cT_{i2}} \wf}
      }{
        |- \braket{\rec{\overline{l_i : \cT_{i1}},\overline{l_j:\cT_j}^{+},*_1},
          \rec{\overline{l_i : \cT_{i2}},*_2}} \wf
      }
    \end{mathpar}
  \end{small}%
  \caption{$\GTFLcsub$: Runtime Language Static Semantics}
  \label{fig:gtflr-statics}
\end{figure}

Since a gradual language defers some checks to runtime, its dynamic semantics
must account for those checks.  Using an approach inspired
by~\citet{toro18secref}, we capture $\GTFLcsub$ checks in an internal language
that decorates programs with runtime type information.

\subsubsection{Syntax and Static Semantics}
\label{ssec:gtflr-statics}

Fig.~\ref{fig:gtflr-statics} presents the syntax and type system of the
$\GTFLcsub$ Runtime Language (RL).
RL mirrors $\GTFLcsub$, but for a few differences.  First, runtime terms have
no source typing information: since RL programs are constructed by
type-directed translation from $\GTFLcsub$ (Sec.~\ref{ssec:gtflr-elaboration}),
we can presume the existence of a corresponding typing derivation, and we need
not reconstruct one by analyzing RL terms.  Omitting static types also lets us
formally distinguish type information that is used for static checking
from information that is used for runtime enforcement.

Second, many RL terms, in particular the subterms of elimination forms and
ascriptions, are decorated with \emph{evidence} objects $\ev$, which summarize
runtime consistent subtyping judgments.  Recall that consistent subtyping
denotes the \emph{plausibility} that a subtyping relationship holds between two
gradual types.  Evidence objects reify that plausibility, and evolve as part of
runtime type enforcement.  We describe this in more detail below.

The typing judgment for runtime terms is unusual in that it classifies a
runtime term $\rr$ with respect to a \emph{source typing judgment}
$\Gamma |- t : \cT$.  The judgment $\rr \derivof \Gamma |- t : \cT$
says that \added{the} runtime term $\rr$ \emph{represents} a source gradual typing
derivation $\Gamma |- t : \cT$, where $t$ is a source-language $\GTFLcsub$
term.
The motivation for this structure is that in a gradually-typed language, type
enforcement is not necessarily completed at type-checking time.  Some type
enforcement may be deferred to runtime, and this enforcement is construed
as an attempt to complete (or refute) the type safety argument at runtime.
Thus one can think of an RL term's execution as playing out type progress and
preservation, which may ultimately fail and justifiably signal a runtime type
error.  In short, RL terms represent the computationally relevant residual of a
$\GTFLcsub$ typing derivation.

Naturally, the structure of RL typing rules closely mirror those of
$\GTFLcsub$, but differ on a few details.
First, the (Sx) and (S::) rules show how missing source type information
appears only as part of the corresponding source derivation, and not as part of
the term.  Gradual types highlighted in grey are artifacts of source
type-checking that cannot be recovered by examining the runtime term itself.
We exploit this information only to establish the type safety of RL programs:
the runtime never needs to reconstruct a \replaced{typing}{type} derivation from a naked RL term.
Second, each instance of consistent subtyping from $\GTFLcsub$ is now decorated
with an evidence object that supports the judgment.  For instance, the RL term
$\ev\rr$ corresponds via typing directly to an ascription expression $t :: S_0$.
The \emph{evidence judgment}  $\ev \evidenceof S \csub S_0$ says that $\ev$
confirms that $S \csub S_0$ is plausible.  

The ($\cT$if) and ($\cT$app) rules consider evidence for extra consistent
subtyping judgments that were only implicitly required by the $\GTFLcsub$
rules.  For instance, ($\cT$app) demands that \linebreak
$\cT_1 \csub \cdom(\cT_1) -> \ccod(\cT_1)$.  This extra constraint is implied
by the $\GTFLcsub$ requirement that \added{$\cdom(\cT_1)$
  and} $\ccod(\cT_1)$, an \replaced{argument and a result}{
application's} type, \replaced{ respectively, are}{is}
well defined.  The elaboration rule, and the corresponding RL term, make this
implicit constraint explicit so as to enforce type structure at runtime.
Similarly, the ($\cT$if) type $\cT_2 \csubjoin \cT_3$ imposes implicit
consistent subtyping constraints on each branch of the conditional.

In RL, an evidence object is a pair of gradual types that are tightly related
to one another, and also related to the consistent subtyping judgment that they
support.  The well-formedness judgment $|- \ev\wf$ imposes several invariants on
evidence.
\begin{proposition}
  \label{prop:wfev}
  The judgment $|- \braket{\cT_1,\cT_2}\wf$ is equivalent to the
  \added{combination of the} following
  criteria:
  \begin{enumerate}
  \item $\cT_1 \csub \cT_2$, \ie $T_1 <: T_2$ for some
    $T_1 \sqsubseteq \cT_1$ and $T_2 \sqsubseteq \cT_2$.
  \item Suppose $\cT'_1, \cT'_2 \in \GType$ and whenever
    $T_1 <: T_2$, $T_1 \sqsubseteq \cT_1$, and $T_2 \sqsubseteq \cT_2$, we also
    have $T_1 \sqsubseteq \cT'_1$ and $T_2 \sqsubseteq \cT'_2$.  Then
    $\cT_1 \sqsubseteq \cT'_1$ and $\cT_2 \sqsubseteq \cT'_2$.
  \end{enumerate}
\end{proposition}
\noindent In short, an evidence object is the least (with respect to $\gprec$)
pair of gradual types that subsumes some set of static subtype instances
$T_1 <: T_2$.  In this sense they are minimal sound representatives of
consistent subtyping.
For instance, $\Int \csub \?$, but $\braket{\Int,\?}$ is
\emph{not} well-formed evidence, because $\Int <: \Int$ is the only pair of
static types that justifies that particular consistent subtyping instance.  The
well-formed evidence object $\braket{\Int,\Int}$ captures the same information.
From the definition of well-formed evidence, it follows immediately that if
$\ev \evidenceof S_1 \csub S_2$ for any $\ev$, then indeed $S_1 \csub S_2$
holds in general.  Evidence ensures that Prop.~\ref{prop:wfev} is maintained
during evaluation.

\subsubsection{Dynamic Semantics}
\label{ssec:gtflr-dynamics}

\begin{figure}
  \centering
  \begin{small}
    \flushleft\textbf{Syntax}
    \begin{displaymath}
      \begin{array}{rcll}
        E \in \ECtxt & ::= &
        [] |
        E[ \rec{\overline{l=\rv},l=[],\overline{l=\rr}} ] | 
        E[F[ \cast{\ev}{[]} ]] 
        & \text{(evaluation contexts)} \\
        F \in \EvFrame & ::= & [] | [] + \cast{\ev}{\rr} | \cast{\ev}{\ru} + [] 
        | []\;\cast{\ev}{\rr} | \cast{\ev}{\ru}\;[]
        | [].l \\
        & | & \<if> [] \<then> \cast{\ev}{\rr} \<else> \cast{\ev}{\rr} 
        &\text{(evidence frame)}
      \end{array}
    \end{displaymath}
    \flushleft\boxed{\rr \leadsto \rr,\; \rr \leadsto \error}
    \quad\textbf{Notions of Reduction}
    \setcounter{errorcontextlines}{999}
    \begin{align*}
      \evcast{\ev_1}{\rn_1} + \evcast{\ev_2}{\rn_2} &\leadsto \rn_3 &
      \rn_3 = \rn_1 + \rn_2 \\
      \cast{\ev_1}{(\lambda \rx.\rr)}\;\cast{\ev_2}{\ru}
      &\leadsto
      \begin{cases}
        \evcast{\bigl(\icod(\ev_1)\bigr)} {
          \bigl([(\cast{\ev_3}{\ru})/\rx] \rr\bigr)} \\
        \error \qquad
        \text{if $\ev_3$ not defined}
      \end{cases}
      & \ev_3 = (\ev_2 \trans{} \idom(\ev_1))
      \\
      \<if>   \evcast{\ev_1}{\ttt} 
      \<then> \evcast{\ev_2}{\rr_2}
      \<else> \evcast{\ev_3}{\rr_3}
      &\leadsto \evcast{\ev_2}{\rr_2} \\
      \<if>   \evcast{\ev_1}{\fff} 
      \<then> \evcast{\ev_2}{\rr_2}
      \<else> \evcast{\ev_3}{\rr_3}
      &\leadsto \evcast{\ev_3}{\rr_3} \\
      \cast{\ev}{\rec{\overline{l_i=\rv_i}}}.l_j
      &\leadsto \iproj(\ev,l_j) \rv_j  \\
    \end{align*}
    \flushleft\boxed{\rr --> \rr,\; \rr --> \error}
    \quad\textbf{Contextual Reduction}
    \begin{mathpar}
      \inference{\rr \leadsto \rr'}{E[\rr] --> E[\rr']}
      \and
      \inference{\rr \leadsto \error}{E[\rr] --> \error}
      \\
      \inference{
      }{
        E[F[\evcast{\ev_1}{\evcast{\ev_2}{\ru}}]] -->
        E[F[\evcast{(\ev_2 \trans{} \ev_1)}{\ru}]]
      } 
      \and
      \inference{
        \ev_2 \trans{} \ev_1 \;\text{not defined}
      }{
        E[F[\evcast{\ev_1}{\evcast{\ev_2}{\ru}}]] --> \error
      } 
    \end{mathpar}
    \flushleft\textbf{Helper Functions} 
    \begin{gather*}
      \textsc{FunEv} = \set{\ev \in \Ev{} | 
        \ev \evidenceof \cT_{11} -> \cT_{12} \csub \cT_{21} -> \cT_{22}} \\
      \textsc{IProjDomain} = \set{\braket{\ev,l} \in \Ev{} \times \Label | 
        \ev \evidenceof 
        \rec{l:\cT,\overline{l_i:\cT_i},*_1} \csub 
        \rec{l:\cT',\overline{l_j:\cT_j},*_2}} \\
      \begin{block}
        \idom : \textsc{FunEv} -> \Ev{} \\
        \idom(\cT_1,\cT_2) = \braket{\cdom(\cT_2),\cdom(\cT_1)}\\
      \end{block}  
      \qquad
      \begin{block}
        \icod : \textsc{FunEv} -> \Ev{} \\
        \icod(\cT_1,\cT_2) = \braket{\ccod(\cT_1),\ccod(\cT_2)}\\
      \end{block}  
      \\
      \begin{block}
        \iproj : \textsc{IProjDomain} -> \Ev{} \\
        \iproj(\braket{\cT_1,\cT_2},l) = 
        \braket{\cproj(\cT_1,l),\cproj(\cT_2,l)}\\
      \end{block}
    \end{gather*}
  \end{small}

  \caption{$\GTFLcsub$: Runtime Language Dynamic Semantics}
  \label{fig:gtflr-dynamics}
\end{figure}

\newcommand{\evpair}[2]{\evlangle #1, #2 \evrangle}

Fig.~\ref{fig:gtflr-dynamics} presents a reduction semantics for RL. The
notions of reduction augment standard reduction steps with operations that
manipulate evidence.  We first consider these reductions at a high-level, and
then delve into the role of evidence in these reductions.
Addition ignores its associated evidence and behaves as usual.
The intuition behind this is that the evidence is now superfluous:
$\rn_1$ and $\rn_2$ evidently have type $\Int$ because they are integers.
Similarly, conditionals ignore the evidence associated with the
predicate, because it is evidently a Boolean value.  The chosen branch's
evidence is propagated as-is to enforce its type invariants.
Record projection selects the relevant field of a record, but also applies the
$\iproj$ operator to the evidence that is associated to the projection
operator, in order to extract evidence that is relevant to the projected
value.
Function application is the most complex rule.  Using the $\idom$ and $\icod$
operators, it extracts evidence associated with the domain and codomain of the
function subterm, composes the domain evidence with the argument evidence, and
then associates the codomain evidence with the eventual result of the call.  If
composition fails, then the entire application fails.%
\footnote{Appendix~\ref{sec:ctrans} \added{in the
    accompanying technical report~\cite{techreport}} directly defines composition; we provide a
  succinct indirect definition in Sec.~\ref{sec:agt} using some additional
  concepts. }

To understand the above behaviours, it helps to think of the operators like
$\iproj$ as runtime representations of \emph{inversion principles} on
consistent subtyping.  For instance, given the inductive definition of
consistent subtyping in Fig.~\ref{fig:gtfl-statics}, we have the following
inversion lemma.
\begin{proposition}
  If $l_k \in \set{\overline{l_j}}$ and
  $\rec{\overline{l_i : \cT_i}} \csub \rec{\overline{l_j : \cT'_j}}$
  then $\cT_k \csub \cT'_k$.
\end{proposition}
The $\iproj$ operator reifies this inversion principle (and those for gradual
rows) at runtime: If
${\ev \evidenceof \rec{\overline{l_i : \cT_i}} \csub%
  \rec{\overline{l_j : \cT'_j}}}$
then $\iproj(\ev,l_k) \evidenceof \cT_k \csub \cT'_k$.  So just as inversion
principles are used to prove static type safety, evidence inversion operators
are used to enforce type invariants at runtime.  The three inversion operators
are \emph{total} functions over appropriate subdomains of evidence, which
properly reflect the kind of runtime evidence that they are used to manipulate
during reduction.  As such, inversion never fails if it is encountered.

While evidence inversion operations are used to simluate inversion principles,
evidence composition $\ev_1 \trans{} \ev_2$ is used to monitor
\emph{transitivity} of consistent subtyping.  For this reason, this composition
operation is called \emph{consistent transitivity}~\added{\cite{garciaAl:popl2016}}. 
The defining correctness criteria for consistent transitivity are as follows.
\begin{definition}
  \label{def:ctranscc}
  Suppose $\ev_1 = \braket{\cT'_1,\cT'_{21}} \evidenceof \cT_1 \csub \cT_2$
  and $\ev_2 = \braket{\cT'_{22},\cT'_3} \evidenceof \cT_2 \csub \cT_3$.    
  \begin{enumerate}
  \item Call ($\ev_1$,$\ev_2$) \emph{plausibly transisitive} if
    $T_1 \gprec \cT'_1$ and $\cT'_{21} \sqsupseteq T_2 \gprec \cT'_{22}$, and
    $T_3 \gprec \cT'_3$ hold for some triple $T_1 <: T_2 <: T_3$;
  \item call $\ev' = \braket{\cT''_1,\cT''_3}$ a \emph{candidate for
      transitivity of ($\ev_1$,$\ev_2$)} if ($\ev_1$,$\ev_2$) is plausibly
    transisitive and every triple $T_1 <: T_2 <: T_3$ such that
    $T_1 \gprec \cT'_1$ and $\cT'_{21} \sqsupseteq T_2 \gprec \cT'_{22}$, and
    $T_3 \gprec \cT'_3$
    implies $T_1 \gprec \cT''_1$
    and $T_3 \gprec\cT''_3$;
  \end{enumerate}
\end{definition}

\begin{proposition}
  \label{prop:ctranscc}
  Suppose $\ev_1 = \braket{\cT'_1,\cT'_{21}} \evidenceof \cT_1 \csub \cT_2$ and
  $\ev_2 = \braket{\cT'_{22},\cT'_3} \evidenceof \cT_2 \csub \cT_3$.  Then
  ($\ev_1$,$\ev_2$) has a candidate for transitivity if and only if it has a
  least candidate \added{over the precision ($\sqsubseteq$) order}, denoted $\ev_1 \trans{} \ev_2$.  Furthermore,
  $(\ev_1 \trans{} \ev_2) \evidenceof \cT_1 \csub \cT_3$.
\end{proposition}

As \citet{GradualTyping} first observed, consistent subtyping is not transitive
in general, and that property is fundamental to gradual type checking.
However, transitivity of subtyping is fundamental to proving type safety.
Thus, $\GTFLcsub$ distills runtime type errors down to a failure of
transitivity.  So unlike the evidence inversion functions, consistent
transitivity is a \emph{partial} function.

Contextual reduction formalizes three kinds of reduction.  First, it captures
how notions of reduction apply in evaluation position.  Second, it captures
how an error during a notion of reduction aborts the entire computation.
Third, it reflects how sequences of evidence objects are composed, producing
new evidence on success or signaling a runtime type error on failure.

\subsubsection{Type Safety}
\label{ssec:gtflr-safety}

RL satisfies standard safety~\added{\cite{WrightFelleisen1994}}: the semantics explicitly categorizes all
well-typed terms as reducible or as canonical forms.  Programs do not get
stuck, but they may signal runtime type errors.%
\footnote{$\GTFLcsub$ RL's safety follows from how its dynamic semantics were
  calculated, using AGT, from $\STFLsub$'s safety proof.}

\begin{proposition}[Progress~\added{\cite{garciaAl:popl2016}}]
  If $\rr \derivof \emptyset |- t : \cT$, then one of the following holds:
  \begin{enumerate}
  \item $\rr$ is a value $\rv$;
  \item $\rr -> \rr'$;
  \item $\rr -> \error$.
  \end{enumerate}
\end{proposition}
\begin{proposition}[Preservation~\added{\cite{garciaAl:popl2016}}]
  \label{prop:rr-preservation}
  If $\rr \derivof \emptyset |- t : \cT$
  and $\rr -> \rr'$ then 
  $\rr' \derivof \emptyset |- t' : \cT$
  for some source $t'$.
\end{proposition}

The statement of type preservation is unusual because the runtime typing
judgment expresses a crisp relationship between source terms and runtime terms.
The source terms $t$ in the runtime typing judgment evolve in lock-step with
the runtime terms $\rr$.  We say more about this in
Sec.~\ref{ssec:gtflr-elaboration}.

Unlike typical preservation theorems for languages with
subtyping~\added{\cite{tapl}}, the type
of the resulting term remains \emph{exactly} the same as the source term.
This is critical because consistent subtyping \emph{does not} denote safe
substitutibility.%
\footnote{This is most evident in the fact that consistent subtyping is not
  transitive.}
Safety requires that any use of consistent subtyping must be mediated by
runtime evidence.  Even if a subterm does evolve to a consistent subtype, it
will be wrapped with runtime evidence that explicitly coerces it back to the
original consistent supertype.

\subsection{Elaboration}
\label{ssec:gtflr-elaboration}

\begin{figure}
  \begin{small}
    \flushleft\boxed{\Gamma |- t \leadsto \rr : \cT}\quad\textbf{Elaboration}
    \begin{mathpar}
      \inference[($\leadsto$x)]{x:\cT\in\Gamma}{\Gamma |- x \leadsto \rx : \cT}
      \and
      \inference[($\leadsto$n)]{}{\Gamma |- n \leadsto \rn : \Int}
      \and
      \inference[($\leadsto$b)]{}{\Gamma |- b \leadsto \rb : \Bool}
      \and
      \inference[($\leadsto$app)]{
        \Gamma |- t_1 \leadsto \rr_1 : \cT_1 & 
        \ev_1 = \Isub|[ \cT_1 \csub \cdom(\cT_1) -> \ccod(\cT_1) |] \\
        \Gamma |- t_2 \leadsto \rr_2 : \cT_2 & 
        \ev_2 = \Isub|[ \cT_2 \csub \cdom(\cT_1) |]
      }{
        \Gamma |- t_1\;t_2 \leadsto \ev_1\rr_1\;\ev_2\rr_2 : \ccod(\cT_1)
      }
      \and
      \inference[($\leadsto$+)]{
        \overline{\Gamma |- t_i \leadsto \rr_i : \cT_i} & 
        \overline{\ev_i = \Isub|[ \cT_i \csub \Int |]}
      }{
        \Gamma |- t_1 + t_2 \leadsto \ev_1\rr_1 + \ev_2\rr_2 : \Int
      }
      \and
      \inference[($\leadsto$if)]{
        \overline{\Gamma |- t_i \leadsto \rr_i: \cT_i} &      
        \ev_1 = \Isub|[ \cT_1 \csub \Bool |] \\
        \overline{\ev_i = \Isub|[ \cT_i \csub \cT_2 \csubjoin \cT_3 |]}\;
        i \in \set{2,3}
      }{
        \Gamma |- 
        \begin{block} \<if> t_1 \<then> t_2 \<else> t_3 
          \leadsto \\\<if> \ev_1\rr_1 \<then> \ev_2\rr_2 \<else> \ev_3\rr_3
          : \cT_2 \csubjoin \cT_3
        \end{block}
      }  
      \and
      \inference[($\leadsto$proj)]{
        \Gamma |- t \leadsto \rr : \cT &
        \ev = \Isub|[ \cT \csub \rec{l:\cproj(\cT,l)} |]
      }{
        \Gamma |- t.l \leadsto \ev\rr.l : \cproj(\cT,l) 
      }
      \and
      \inference[($\leadsto\!\lambda$)]{
        \Gamma, x:\cT_1 |- t \leadsto \rr: \cT_2
      }{
        \Gamma |- (\lambda x : \cT_1.t) \leadsto (\lambda \rx.\rr): \cT_1 -> \cT_2
      }
      \and
      \inference[($\leadsto::$)]{
        \Gamma |- t \leadsto \rr : \cT &
        \ev = \Isub|[ \cT \csub \cT_1 |]
      }{
        \Gamma |- (t :: \cT_1) \leadsto \ev\rr : \cT_1
      }  
      \and
      \inference[($\leadsto$rec)]{
        \overline{\Gamma |- t_i \leadsto \rr_i : \cT_i}
      }{
        \Gamma |- \rec{\overline{l_i = t_i}}
        \leadsto \rec{\overline{l_i = \rr_i}} : \rec{\overline{l_i:\cT_i}}
      } 
      \\
    \end{mathpar}
  \end{small}%
  \caption{$\GTFLcsub$: Elaboration}
  \label{fig:gtfl-elaboration}
\end{figure}

$\GTFLcsub$ source programs are elaborated to RL programs by type-directed
translation (Fig.~\ref{fig:gtfl-elaboration}).  The elaboration process is
quite uniform, exhibiting the tight connection between $\GTFLcsub$ and RL.  In
essence RL terms represent $\GTFLcsub$ derivation trees, adding essential
runtime type information (\ie evidence) and erasing superfluous source type
information.

Each elaboration rule corresponds directly to a source $\GTFLcsub$ typing rule.
In particular, each premise typing judgment becomes a corresponding elaboration
judgment, and each consistent subtyping judgment introduces an evidence object,
using the \emph{initial evidence} operator.%
\footnote{An inductive definition of
  initial evidence is presented in the
  \replaced{accompanying technical report~\cite{techreport}}{appendix}, see
  Fig.~\ref{fig:gtfl-initial-evidence}}
The initial evidence operator $\Isub|[ S_1 \csub S_2 |]$ computes the
\emph{largest} evidence object \replaced{$\ev{} \equiv \braket{S'_1,S'_2}$}{$\braket{S'_1,S'_2}$}, such that
${\braket{S'_1,S'_2} \evidenceof S_1 \csub S_2}$.  From the definition of
well-founded evidence in Sec.~\ref{ssec:gtflr-statics}, we can deduce that
$\braket{S'_1,S'_2}$ is the \emph{smallest} pair of gradual types that subsumes
all of the same static subtype instances $T_1 <: T_2$ as $S_1 \csub S_2$ does
(i.e. if $T_1 <: T_2$ then $\braket{T_1,T_2} \gprec \braket{S_1,S_2}$ iff
$\braket{T_1,T_2} \gprec \braket{S'_1,S'_2}$).
Naturally, initial evidence is undefined if $S_1 \not\csub S_2$, but that
circumstance does not arise if type checking succeeds. 

The ($\leadsto$if) and ($\leadsto$app) elimination rules produce evidence for
extra consistent subtyping judgments that were not evident in the
corresponting ($S$if) and ($S$app) $\GTFLcsub$
rules.  For instance, ($\leadsto$app) demands that 
${\cT_1 \csub \cdom(\cT_1) -> \ccod(\cT_1)}$.  This extra constraint was implied
by the $\GTFLcsub$ requirement that $\ccod(\cT_1)$ be well-defined.  The
elaboration rule, and the corresponding RL term, make this implicit constraint
explicit because it must be enforced at runtime.  Similarly,
the ($\leadsto$if) type $\cT_2 \csubjoin \cT_3$ imposes implicit consistent
subtyping constraints on each branch of the conditional.

The tight connection between $\GTFLcsub$ and RL is confirmed by 
preservation of well-typedness.
\begin{proposition}[Well-formed Translation]
  \label{prop:gtfl-to-rr-wf}
  If $\Gamma |- t \leadsto \rr : \cT$ then $\rr \derivof \Gamma |- t : \cT$.
\end{proposition}
So the source term $t$ elaborates to a runtime term $\rr$ that
represents $t$'s static typing derivation.

We are now equipped to better explain the statement of preservation in
Prop.~\ref{prop:rr-preservation}.  We relate runtime terms to source terms
using the runtime typing judgment $\rr \derivof \Gamma |- t : \cT$ and source
terms to runtime terms using the translation judgment
$\Gamma |- t \leadsto \rr : \cT$.  Preservation clarifies how runtime terms
\emph{``learn''} new type constraints that are not evident in source programs.
\begin{proposition}[Replicant]
  If $\rr \derivof \Gamma |- t : \cT$ and $\Gamma |- t \leadsto \rr' : \cT$
  then $\rr \gprec \rr'$.
\end{proposition}
In this proposition, $\rr \gprec \rr'$ refers to the pointwise extension of
evidence precision from types to terms.
This proposition says that a runtime term $\rr$ embodies at least as many
runtime type constraints as any source term that it represents.  In the above
term, we call $\rr'$ a \emph{replicant} of $\rr$, because it amounts to cloning
the structure of $\rr$ but omitting its \emph{``memory''} of any additional
type constraints acquired while evaluating the program.  The $\rr'$ term starts
with a clean slate, yet to be jaded by the tribulations of runtime type
enforcement.

The key idea is that a source program contains enough local type information to
justify its \emph{plausible} typeability, but the runtime term
\replaced{$\rr$ must account}{must $\rr$
accounts} for type invariants imposed at runtime.  Below, we observe that \deleted{the}
these discovered type constraints must be represended precisely
in order to \replaced{enforce}{enforcing} type-based invariants at runtime in mixed-type programs.

\subsection{$\GTFLcsub$ Shortcomings Revisited}
\label{ssec:shortcomings}

Using the semantics of $\GTFLcsub$ presented above, we can more closely
examine the program from Sec.~\ref{sec:intro} that failed to protect
type invariant boundaries.  Our example program:
\begin{displaymath}
  \<let> \q : \rec{\x:\Int} = \rec{\x = 5, \y = \ttt}
  \<in> (\q :: \? :: \rec{\x:\Int,\y:\Bool}).\y 
\end{displaymath}
elaborates to the following:
\footnote{The typing and elaboration rules for $\<let>\!\!$ can be deduced from
  the rules for lambda and application.}
\begin{displaymath}
  \<let> \q = \ev_1\rec{\x = 5, \y = \ttt}
  \<in> \ev_4(\ev_3\ev_2q).\y 
\end{displaymath}
where
\begin{displaymath}
  \begin{array}{lll}
    \ev_1 &=  \Isub|[ \rec{\x:\Int,\y:\Bool} \csub \rec{\x:\Int} |]
    &= \evpair{\rec{\x:\Int,\y:\Bool}}{\rec{\x:\Int}} \\
    \ev_2 &=  \Isub|[ \rec{\x:\Int} \csub \? |]
    &= \evpair{\rec{\x:\Int}}{\rec{\?}} \\
    \ev_3 &=  \Isub|[ \? \csub  \rec{\x:\Int,\y:\Bool} |]
    &= \evpair{\rec{\x:\Int,\y:\Bool,\?}}{\rec{\x:\Int,\y:\Bool}}\\
    \ev_4 &=  \Isub|[ \rec{\x:\Int,\y:\Bool} \csub {\rec{\y:\Bool}} |]
    &= \evpair{\rec{\x:\Int,\y:\Bool}}{\rec{\y:\Bool}}\\
  \end{array}
\end{displaymath}
Evaluation proceeds as follows:
\begin{align*}
  & \<let> \q = \ev_1\rec{\x = 5, \y = \ttt} \<in> \ev_4(\ev_3\ev_2q).\y &
  -->\;&  \ev_4(\ev_3\ev_2\ev_1\rec{\x = 5, \y = \ttt}).\y \\
  -->\;&  \ev_4(\ev_3(\ev_1\trans{} \ev_2))\rec{\x = 5, \y = \ttt}).\y &
  -->\;&  \ev_4(((\ev_1\trans{} \ev_2) \trans{} \ev_3)\rec{\x = 5, \y = \ttt}).\y \\
  -->\;&  (((\ev_1 \trans{} \ev_2) \trans{} \ev_3) \trans{} \ev_4)\rec{\x = 5, \y = \ttt}.\y &
  -->\;& \ev_5\ttt
\end{align*}
where
\begin{align*}
  (\ev_1 \trans{} \ev_2) &= \evpair{\rec{\x:\Int,\y:\Bool}}{\rec{\?}} &
  \hspace{-4em}(((\ev_1 \trans{} \ev_2) \trans{} \ev_3) \trans{} \ev_4) &=
  \evpair{\rec{\x:\Int,\y:\Bool}}{\rec{\y:\Bool}}\\
  ((\ev_1\trans{} \ev_2) \trans{} \ev_3) &=
  \evpair{\rec{\x:\Int,\y:\Bool}}{\rec{\x:\Int,\y:\Bool}}&
  \ev_5 &= \evpair{\Bool}{\Bool} \\
\end{align*}

This program should signal a \added{run-time} type error, but it does not.
Tracing its evaluation reveals what went wrong.  First, observe that
$\ev_2 \trans{} \ev_3$ is undefined, even though
$((\ev_1\trans{} \ev_2) \trans{} \ev_3)$ \emph{is} defined.  This demonstrates
that composition is not associative.  Now consider
$(\ev_1 \trans{} \ev_2) = \evpair{\rec{\x:\Int,\y:\Bool}}{\rec{\?}}$ more
closely.  Appealing to our notion of well-formed evidence, this object contains
via precision, every valid static subtyping pair $\rec{\x:\Int,\y:\Bool} <: T$,
including $T = \rec{\x:\Int,\y:\Bool}$.

However, if we consider the correctness criteria for Prop.~\ref{prop:ctranscc},
as applied to composing $\ev_1$ with $\ev_2$, then the only static types that
complete the relevant triples $T_1 <: T_2 <: T_3$ are $T_3 = \rec{\x:\Int}$ and
$T_3 = \rec{\;}$.  Thus $\rec{\x:\Int,\y:\Bool}$ is a spurious potential
supertype due solely to the choice of evidence abstraction.
These phenomena suggest that our abstraction for evidence is insufficiently
precise \deleted{to support space-efficiency and type-based reasoning}.  We fix this
problem in Sec.~\ref{sec:brr}.

\section{The Essence of AGT}
\label{sec:agt}

Since our goal is to develop a general solution to the demonstrated issues, our
approach must be couched in terms of the AGT framework itself.  This section
presents a brief introduction to the AGT methodology and its basis in abstract
interpretation~\cite{AbstractInterpretation}.  We reconceive some of the
concepts involved in a manner that is equivalent to the original.  Our new
approach improves on the original in that it clarifies the nature of \emph{two}
distinct but related abstractions: one for the static semantics and one for the
dynamic semantics.  This refinement is critical to our contributions.
Along the way, we explain relevant aspects of AGT that we suppressed in
Sec.~\ref{sec:gtflcsub}.

Given a pre-existing statically typed language and its type safety argument,
AGT starts by extending the syntax of types to form gradual types, in our case
$\cT \in \GType$.  In addition to the syntax for gradual types, we endow
gradual types with a semantics: each gradual type $\cT$ is given an
interpretation as a non-empty set of static types from the static
language, denoted $\Pow^{+}(\Type)$.
Following abstract interpretation terminology, this function is called
\emph{concretization}.
\begin{gather*}
  \gamma : \GType -> \Pow^{+}(\Type) \\
  \gamma(\cT) = \set{T \in \Type | T \sqsubseteq \cT}.
\end{gather*}
Thus, a gradual type denotes every static type that it could represent.  Since
every static type is a (minimal) gradual type, this implies that
$\gamma(T) = \set{T}$.  When starting from scratch, concretization is typically
defined by recursion on the structure of gradual types
\cite{{garciaAl:popl2016}}, and then the precision relation is determined by
the following correctness criterion:
\begin{proposition}[\added{\cite{garciaAl:popl2016}}]
  $\cT_1 \sqsubseteq \cT_2$ if and only if
  $\gamma(\cT_1) \subseteq \gamma(\cT_1)$. 
\end{proposition}
Here we reverse the development, defining concretization in terms of precision,
since precision is introduced in Sec.~\ref{sec:gtflcsub}.  These approaches are
equivalent, though a structurally recursive definition of concretization is
critical to calculating functions and relations on gradual types like $\cproj$
and $\Isub$.

The concretization function in turn induces a converse function, typically
called \emph{abstraction}.
\begin{gather*}
  \alpha : \Pow^{+}(\Type) -> \GType \\
  \alpha(\CC) = \bigsqcap \set{ \cT \in \GType | \CC \subseteq \gamma(\cT) }
\end{gather*}
In this definition, $\bigsqcap$ denotes the greatest lower-bound with respect
to precision $\sqsubseteq$.  Abstraction $\alpha(\CC)$ yields the most precise gradual type
that summarizes a given set of static types.  In order for $\alpha$ to be
well-defined, we must ensure that the right-hand side of the above equation is
well-defined for any collection $\CC$ to be considered: in particular there
must exist an appropriate greatest lower-bound.
In general, a gradual type system does not need to account for every arbitrary
subset of static types $\CC \in \Pow^{+}(\Type)$ as part of concretization and
abstraction.  In fact, in Sec~\ref{ssec:brrai} it will be critical to restrict
the space of sets of types (with no loss of expressiveness) to ensure that
$\alpha$ is well-defined.
\footnote{Unlike \citet{garciaAl:popl2016}, we define $\alpha$ to be a
  \emph{total} function over a family of (non-empty) sets of types, rather than
  a \emph{partial} function over \emph{arbitrary} sets of types $\Pow(\Type)$.
  This approach pushes partiality into the collecting operators (\cf $\cod^{+}$
  below), rather than the Galois connection itself.  We find this approach more
  general, intuitive, and mathematically pleasant.}
If abstraction $\alpha$ is well-defined, then it is uniquely determined by
concretization $\gamma$.  To facilitate the calculation of operators and
relations on gradual types, abstraction is typically presented as a recursive
function definition~\cite{garciaAl:popl2016}.

By design, concretization and abstraction are related by two properties that
together characterize a \emph{Galois connection}~\cite{AbstractInterpretation}.
\begin{proposition}[\added{\cite{garciaAl:popl2016}}]
  \mbox{}
  \begin{enumerate}
  \item Soundness: $\CC \subseteq \gamma(\alpha(\CC))$
  \item Optimality: If $\CC \subseteq \gamma(\cT)$ then
    $\gamma(\alpha(\CC)) \subseteq \gamma(\cT)$.
  \end{enumerate}
\end{proposition}
\noindent These functions satisfy the stronger property that
$\alpha(\gamma(\cT)) = \cT$.  This makes them a \emph{Galois insertion}.

\citet{garciaAl:popl2016} use concretization and abstraction to define all of
the gradual typing artifacts used to define $\GTFLcsub$.  For example, consider
the codomain of a function type, rendered as a partial function $\cod$ on
static types:
\begin{gather*}
  \cod : \Type \rightharpoonup \Type \\
  \begin{aligned}
    \cod(T_1 -> T_2) &= T_2 \\
    \cod(T) &\;\text{undefined otherwise}
  \end{aligned}
\end{gather*}
We can idiomatically \emph{lift} $\cod$ to a partial function $\cod^{+}$ on
non-empty sets of types.
\begin{gather*}
  \cod^{+} : \Pow^{+}(\Type) \rightharpoonup \Pow^{+}(\Type) \\
  \begin{aligned}
    \cod^{+}(\CC) &= \CC' \; \text{ if } \CC' \neq \emptyset \\
    \cod^{+}(\CC) &\;\text{undefined if }  \CC' = \emptyset \\
  \end{aligned} \\
  \text{where}\quad
  \CC' = \set{T' \in \Type | T \in \CC \text{ and } \cod(T) = T'}
\end{gather*}
The $\cod^{+}$ function applies $\cod$ pointwise to a set of types, and
produces the resulting set of codomains, so long as there are some codomains:
if not, then the function is undefined.  Using this lifted function, we can
define the $\ccod$ function as follows.
\begin{gather*}
  \ccod(\cT) = \alpha(\cod^{+}(\gamma(\cT)).
\end{gather*}
Given the definitions of $\ccod$, $\gamma$, and $\alpha$, the correctness
criteria of Prop.~\ref{prop:ccod-correct} follow.  From
there, the recursive characterization of $\ccod$ 
from Fig.~\ref{fig:gtfl-statics} can be calculated by induction
on the structure of $\cT$.
Analogous constructions can be used to first define and then calculate
algorithms for all of the operations on gradual types from
Sec.~\ref{sec:gtflcsub}.  The definitions directly imply their own correctness
criteria, and the corresponding algorithms require tedious but straightforward
calculations.

\rparagraph{Evidence is its own abstraction}
\newcommand{\R}{\mathcal{R}}
The concept of evidence for consistent subtyping arises through an analogous
process.  \citet{garciaAl:popl2016} define evidence by abstracting tuples of
static types, filtering them post-hoc.  Our approach refines theirs in a small
but fundamental way.
In particular, we introduce a \emph{second} Galois connection, but this one is
directly between consistent subtyping, regarded as a set of pairs of gradual
types $\braket{\cT_1,\cT_2} \in {\csub}$, and nonempty subsets of the static
subtyping relation, \ie sets of pairs $\R \subseteq {<:}$.

\added{The key difference in this new approach is that the set of
  evidence objects form \emph{in their entirety} an abstract interpretation
  of static subtyping. In contrast, the evidence objects introduced
  by~\cite{garciaAl:popl2016} were
  arbitrary pairs of gradual types. The concretization of an
  evidence object could therefore contain pairs of types not related by subtyping.
  Integrating subtyping into the conception of
  abstraction, and viewing composition as relational composition
  specifically for fragments of subtyping is critical to formulating
  the possibility, let alone the significance, of forward
  completeness.}

We begin with concretization.
\begin{gather*}
  \gamma^{<:} : {\csub} -> \Pow^{+}(<:) \\
  \gamma^{<:}(\cT_1,\cT_2) = 
  \set{\braket{T_1,T_2} | T_i \in \gamma(\cT_i) \text{ and } T_1 <: T_2}.
\end{gather*}

The corresponding notion of abstraction is defined as follows
\begin{gather*}
  \alpha^{<:} : \Pow^{+}(<:) -> {\csub} \\
  \alpha^{<:}(\R) = \braket{\alpha(\pi_1(\R)),\alpha(\pi_2(\R))}.
\end{gather*}
Abstraction appeals to two point-wise projection functions that operate on sets
of pairs.  These recur in later developments.
\begin{gather*}
  \pi_1,\pi_2 : \Pow(\Type^2) -> \Pow(\Type)\\
  \pi_i(\set{\overline{\braket{T_1,T_2}}}) = \set{\overline{T_i}}
\end{gather*}

\newcommand{\Idp}{\mathit{Id}^{+}}

The functions $\gamma^{<:}$ and $\alpha^{<:}$ form a Galois connection, but
\emph{they do not form a Galois insertion}.  For example:
$\alpha^{<:}(\gamma^{<:}(\Int,\?)) = \braket{\Int,\Int}$.  In fact,
$\alpha^{<:}(\gamma^{<:}(\cT_1,\cT_2))$ is equivalent, by design, to the
initial evidence function $\Isub|[ \cT_1 \csub \cT_2 |]$ from
Fig.~\ref{fig:gtfl-elaboration}.  The $\Isub$ function simply extends its
domain to accept, but be undefined for, inconsistent gradual types
$\cT_1 \not\csub \cT_2$.  Thus, evidence objects $\ev$ are exactly those
pairs of gradual types in the image of $\alpha^{<:} \circ \gamma^{<:}$.

Furthermore, composing evidence objects using the consistent transitivity
operator can be defined in terms of \emph{relational composition} of subsets of
static subtyping $<:$.
\begin{gather*}
  \ev_1 \trans{} \ev_2 =
  \alpha^{<:}(\Idp(\gamma^{<:}(\ev_1) \relcomp \gamma^{<:}(\ev_2)))
  \\
  \text{ where } \R_1 \relcomp \R_2 = 
  \set{\braket{T_1,T_3} | T_1 \mathrel{\R_1} T_2 \text{ and }
    T_2 \mathrel{\R_2} T_3 \text{ for some } T_2} \\
  \text{ and } 
  \Idp(\R) = 
  \begin{cases}
    \R & \R \neq \emptyset \\
    \text{undefined} & \text{otherwise}.
  \end{cases}
\end{gather*}
This definition can be used to calculate a direct recursive characterization of
consistent transitivity (see \replaced{technical report~\cite{techreport}}{Appendix~\ref{sec:ctrans}}), but we can also prove
that it is equivalent to a combination of simpler operations.%
\footnote{\citet{garciaAl:popl2016} erroneously omits the outermost instance of
  $\Isub$ in their definition.}
\begin{proposition}
  \label{prop:ctrans-iml}
  \begin{equation*}
    \evpr{\cT_1,\cT_{21}} \trans{} \evpr{\cT_{22},\cT_3}
    =
    \Isub|[ 
    \pi_1(\Isub|[ \cT_1 \csub (\cT_{21} \sqcap \cT_{22}) |])
    \csub
    \pi_2(\Isub|[ (\cT_{21} \sqcap \cT_{22}) \csub \cT_3 |])
    |]
  \end{equation*}

  \added{where
    $\cT_{21} \sqcap \cT_{22} = \alpha(\Idp(\gamma(\cT_{21}) \cap \gamma(\cT_{22})))$.}
\end{proposition}

A key observation here is that pairs of gradual types are but one arbitrary,
albeit convenient, abstraction for subsets of the static subtyping relation.
This analysis of evidence as abstractions of static subtyping fragments, and
consistent transitivity as abstract relational composition, is critical to our
analysis and improvement of evidence in Sec.~\ref{sec:brr}.

\section{Space efficient  semantics for AGT-based  languages}
The dynamic semantics presented in Fig.~\ref{fig:gtflr-dynamics} follows the
standard AGT methodology.  Evidence is composed \emph{after} an ascribed
expression is reduced, in an order that accumulates evidence on a resulting
value.
Thus RL, and every other AGT-based language, suffers from the same
space issues described by \citet{herman10space}.
This section presents an alternate reduction semantics, $\text{RL}^{+}$, that
composes evidence as soon as possible to prevent the accumulation of pending
evidence objects.  $\text{RL}^{+}$ requires associative evidence composition to
guarantee that it will evaluate to the same results as RL.  Unfortunately, this
property does not hold for $\GTFLcsub$, but Sec.~\ref{sec:brr} addresses this
limitation by providing a better evidence abstraction.
Sec.~\ref{sec:forward-completeness} describes criteria that guarantee
associativity for evidence composition.  Sec.~\ref{sec:brr-space-efficient}
proves that the criteria hold for our novel abstraction, and proves an
upper-bound on space consumption.

\rparagraph{Space-efficient, but observationally equivalent}
To preserve tail recursion, a semantics must compose evidence eagerly, thereby
keeping evidence objects from accumulating in tail position.  However, to keep
the semantics equivalent to the original, the space-efficient one must not fail
early in case a program diverges.  To preserve this behavior, a space-efficient
semantics must remember ``latent failure'' until evaluation yields a value,
while avoiding unbounded stack growth.

\begin{figure}
  \begin{small}
  \flushleft\textbf{Syntax}
  \begin{displaymath}
    \begin{array}{rcll}
      \pls{\botev} & ::= & \pls{\bot} | \pls{\ev} & \text{(total
        evidence)} \\
      \pls{e} \in \pls{\RTerm} & ::= & \pls{n} | \pls{\cast{\botev}{e} + \cast{\botev}{e}} | \pls{x} | \pls{\lambda x.e} |  \pls{\evcast{\botev}{e}\;\evcast{\botev}{e}} | 
      \pls{\evcast{\botev}{e}} | \pls{\rec{\overline{l = e}}} | \pls{\cast{\botev}{e}.l} & \text{(runtime
        terms)} \\
      && | \pls{ \<if> \cast{\botev}{e} \<then> \cast{\ev}{e} \<else> \cast{\ev}{e} }&\\
      \pls{F} \in \pls{\EvFrame} & ::= & \pls{[]\;\cast{\ev}{e}} | \pls{\cast{\ev}{u}\;[]} | \pls{[] + \cast{\ev}{e}} | \pls{\cast{\ev}{u} + []} |  \pls{[].l}& \text{(runtime
        frames)}\\
      & | & \pls{\<if> [] \<then> \cast{\ev}{e} \<else> \cast{\ev}{e}}\\
      \pls{G} \in \pls{\text{ACtxt}} & ::= & \pls{[]} |
      \pls{E[F[ [] ]]} |
      \pls{E[ \rec{\overline{l=\rv},l=[],\overline{l=\rr}} ]} 
      & \text{(Evaluation contexts}\\
      && \multicolumn{2}{r}{ \text{without an innermost ascription)}}\\
      \pls{E} \in \pls{\ECtxt} & ::= & \pls{G} |
      \pls{G[\evcast{\botev}{[]}]}
      &
      \text{(Evaluation contexts)}
    \end{array}
  \end{displaymath}

  \flushleft\textbf{Contextual Reduction}
  \flushleft\boxed{\pls{e} --> \pls{e},\; \pls{e} --> \pls{\error}}

  \begin{mathpar}
    \inference{\pls{e} \leadsto \pls{e'}}
    {\pls{E[e]} --> \pls{E[e']}}
    \quad
    \inference{\pls{e} \leadsto \pls{\error}}{
      \pls{E[e]} --> \pls{\error}}
    \quad
    \inference{}{
      \pls{G[\evcast{\botev_1}{\evcast{\botev_2}{e}}]} -->
      \pls{G[\evcast{\botev_2 \bottrans{} \botev_1}{e}]}}
    \quad
    \inference{}{
      \pls{G[\evcast{\bot}{u}]} --> \pls{\error}}
  \end{mathpar}
  \end{small}
  \caption{Syntax and Contextual Reduction of $\text{RL}^{+}$.  All
    other parts of the semantics are as in RL.}
  \label{fig:serl-syntax}
\end{figure}

The $\text{RL}^{+}$ semantics applies the standard technique of explicitly
representing failure.  It introduces \emph{``total'' evidence}, $\botev$, which
represents either normal evidence $\ev$, or a latent failure $\bot$, which
serves the same purpose as the failure coercion of \cite{herman10space}.  A
latent failure indicates that some evidence composition was undefined, and will
trigger an error \emph{after} its argument finishes reducing.  The syntax for
$\text{RL}^{+}$ terms (Fig.~\ref{fig:serl-syntax}), essentially mirrors RL, but
substitutes total evidence for evidence.  To emphasize this  distinction, we
color the terms in the $\text{RL}^{+}$ differently e.g., $\pls{e}$.

One key difference in the definition of $\text{RL}^{+}$ is its runtime frames:
the empty $[]$ frame, which allows evidence objects to accumulate on the RL
evaluation stack, is omitted.  Frames otherwise remain the same: For example,
$\text{RL}^{+}$ does not need a frame of the form
$\pls{\evcast{\botev}{u}\;[]}$ because when the function position of an
application reduces to $\pls{\evcast{\bot}{u}}$, the original semantics would
trigger an error.

\subsection{The reduction semantics of $\text{RL}^{+}$}
$\text{RL}^{+}$ has the same the notions of reduction as RL.  Thus, it suffices
to change the definition of contextual reduction $-->$.  To compose evidence
objects ahead of time without failing too early, we replace the partial
evidence composition operation $\trans{}$ used in RL with a related total
operation denoted $\bottrans{}$, which yields latent failures $\bot$ when
$\trans{}$ is undefined:
\[\botev_1 \bottrans{} \botev_2 = \begin{cases}
    \ev_1 \trans{} \ev_2 & \text {if } \botev_1 = \ev_1
    \text{ and  }  \botev_2 = \ev_2 \text{ and composition is defined} \\
    \bot & \text{otherwise}
  \end{cases}\]

Since the definition of frames no longer includes holes, ascriptions can no
longer nest in evaluation contexts. In fact, evidence objects that come from
ascriptions, rather than elimination forms, do not form evaluation
contexts: they must be handled explicitly.  Composition no longer waits
for values, but proceeds when nested ascriptions become the
immediate redex.

Finally, to preserve divergent behaviour, errors are triggered by notions of
reduction only upon reaching a value ($\pls{G[\evcast{\bot}{u}]}$).

\begin{figure}
  \begin{small}
\flushleft\boxed{\pla{e} \approx \pls{e}, \pla{\error} \approx \pls{\error}}
\begin{mathpar}
  \inference[Sim-Casts]{\pla{e_1} \neq \pla{\evcast{\ev'}{e'}} &
    \pls{e_2} \neq \pls{\evcast{\botev'}{e'}} & 
    \pla{e_1} \approx \pls{e_2} &
    \pla{\ev_n} \bottrans{} \! \left( \cdots \bottrans{} \pla{\ev_1}\right) =
    \left(\pls{\botev_m} \bottrans{} \! \cdots \right)\bottrans{} \pls{\botev_1}}{
    \pla{\ev_1 \cdots \evcast{\ev_n}{e_1}} \approx \pls{\botev_1
      \cdots \evcast{\botev_m}{e_2}}}

  \inference[Sim-Err]{}{\pla{\error} \approx \pls{\error}}

  \inference[Sim-U]{}{\pla{u} \approx \pls{u}}

  \inference[Sim-App]{\pla{\evcast{\ev_1}{e_1}} \approx
    \pls{\evcast{\botev_3}{e_3}} &
    \pla{\evcast{\ev_2}{e_2} \approx \pls{\evcast{\botev_4}{e_4}}}}{
    \pla{\cast{\ev_1}{e_1}\;\cast{\ev_2}{e_2}} \approx
    \pls{\cast{\botev_3}{e_3}\;\cast{\botev_4}{e_4}}}

  \inference[Sim-Plus]{\pla{\evcast{\ev_1}{e_1}} \approx
    \pls{\evcast{\botev_3}{e_3}} &
    \pla{\evcast{\ev_2}{e_2} \approx \pls{\evcast{\botev_4}{e_4}}}}{
    \pla{\cast{\ev_1}{e_1}+\cast{\ev_2}{e_2}} \approx
    \pls{\cast{\botev_3}{e_3}+\cast{\botev_4}{e_4}}}

  \inference[Sim-Rec]{\pla{e_i} \approx \pls{e'_i}}{
    \pla{\rec{\overline{l_i = e_i}}} \approx \pls{\rec{\overline{l_i = e'_i}}}}

  \inference[Sim-If]{\pla{\evcast{\ev_1}{e_1}} \approx
    \pls{\evcast{\botev_4}{e_4}} &
    \pla{\evcast{\ev_2}{e_2} \approx \pls{\evcast{\botev_5}{e_5}}} &
    \pla{\evcast{\ev_3}{e_3} \approx \pls{\evcast{\botev_6}{e_6}}} 
  }{
    \pla{\<if> \cast{\ev_1}{e_1} \<then> \cast{\ev_2}{e_2} \<else>
      \cast{\ev_3}{e_3}} \approx
    \pls{\<if> \cast{\botev_4}{e_4} \<then> \cast{\botev_5}{e_5}
      \<else> \cast{\botev_6}{e_6}}}

  \inference[Sim-Proj]{\pla{\evcast{\ev_1}{e_1}} \approx
    \pls{\evcast{\botev_2}{e_2}}}{
    \pla{\cast{\ev_1}{e_1}.l} \approx
    \pls{\cast{\botev_2}{e_2}.l}}

\end{mathpar}
\end{small}
\caption{\comment{(NEW FIGURE!)}Bisimulation Relation for
  space efficiency}
\label{fig:bisimulation-relation}
\end{figure}

\added{We prove that RL and $\text{RL}^{+}$ 
  are weakly bisimilar.  Fig.~\ref{fig:bisimulation-relation} introduces a
  bisimulation relation between terms in the both languages to carry
  the proof.  While most rules follow
  inductively, rule [Sim-Casts] accepts terms with differing stacks of evidence compositions, as long as the
  final composed result (in each language's order) is the same.}

\begin{theorem}[Weak Bisimulation between RL and $\text{RL}^{+}$]
  \label{thm:weak-bis}
  If $e_1 \approx \pls{e_2}$%
  \footnote{The \deleted{bisimulation relation $\approx$ and}
    proof of the theorem \replaced{appears}{appear} in the
    \replaced{technical report~\cite{techreport}}{Appendix}. See
    Sec.~\ref{apdx:bisimulation-relation}},
  \textbf{and evidence composition is associative}, then:
  \begin{itemize}
  \item
    If there exists $\pla{e'_1}$ in RL such that $\pla{e_1} --> \pla{e'_1}$, 
    then there exists $\pls{e'_2}$ such that $\pls{e_2 ->^{*} e'_2}$ and
    $\pla{e'_1} \approx \pls{e'_2}$.
  \item If there exists $\pls{e'_2}$ in $\text{RL}^{+}$ such that
    $\pls{e_2 --> e'_2}$, then 
    there exists $\pla{e'_1}$ such that $\pla{e_1} ->^{*} \pla{e'_1}$ and
    $\pla{e'_1} \approx \pls{e'_2}$.
  \end{itemize}
  Therefore, the relation $\cdot \approx \pls{\cdot}$ forms a Weak
  Bisimulation~\cite{sangiorgi_2011} between $-->$ and $\pls{-->}$.
\end{theorem}

\subsection{Claiming Space-Efficiency for $\text{RL}^{+}$}

The $\text{RL}^{+}$ dynamic semantics never accumulate evidence compositions in
tail position, but this alone does not guarantee space efficiency. If evidence
objects themselves can grow without bound, then the changes to the semantics
amount to simply shuffling space usage without eradicating it.  To ensure space
efficiency we must quantify the overhead of evidence throughout evaluation.

Following~\citet{herman10space}, we provide a bound for the space consumption
of a program.  The $\spacep{f}{e}$ function (Fig.~\ref{fig:spacep}) is
parameterized by a function $f$ for the space consumption of an evidence
object.  We use this structural function to compare the space consumption of a
program with ``0-cost'' evidence to $\text{RL}^{+}$ programs.

\begin{figure}
  \begin{small}
  \begin{displaymath}
    \begin{block}
      \boxed{\spacep{f}{e}}\quad \textbf{Space consumption of a program}\\[0.5em]
      \spacep{\cdot}{\cdot}{\cdot} : \deleted{\ensuremath{\left(\left(\Ev{<:}\cup\set{\bot}\right) -> \nat\right) ->}} \left(\left(\Ev{<:}\cup\set{\bot}\right) -> \nat\right) -> \RTerm -> \nat \\[0.5em]
      \spacep{f}{\pls{n}} = \spacep{f}{\pls{x}} = 1 \\[0.5em]
      \spacep{f}{\pls{\cast{\botev}{e_1} + \cast{\botev}{e_2}}} = 1 + f(\pls{\botev_1}) + \spacep{f}{\pls{e_1}} + f(\pls{\botev_2}) + \spacep{f}{\pls{e_2}} \\[0.5em]
      \spacep{f}{\pls{\evcast{\botev}{e}\;\evcast{\botev}{e}}} = 1 + f(\pls{\botev_1}) + \spacep{f}{\pls{e_1}} + f(\pls{\botev_2}) + \spacep{f}{\pls{e_2}} \\[0.5em]
      \spacep{f}{\pls{\lambda x.e}} = 1 + \spacep{f}{e} \\[0.5em]
      \spacep{f}{\pls{\evcast{\botev}{e}}} = 1 + \spacep{f}{\pls{\botev}} + \spacep{f}{\pls{e}}\\[0.5em]
      \spacep{f}{\pls{\cast{\botev}{e}.l}} = 1 + f(\pls{\botev}) + \spacep{f}{\pls{e}}\\[0.5em]
      \spacep{f}{\pls{\rec{\overline{l = e}}}} = 1 + \sum \left(\spacep{f}{\pls{e}}\right)\\[0.5em]
    \end{block}
  \end{displaymath}
  \end{small}
  \caption{Space consumption of a program in $\text{RL}^{+}$}
  \label{fig:spacep}
\end{figure}

\begin{definition}[Space-efficiency of evidence]
  Let $\size$ be the function that computes the space consumption of an
  evidence object.  A reduction semantics is \emph{space efficient} with
  respect to evidence objects if the overhead due to carrying evidence is
  bound throughout evaluation \emph{and} ascriptions never wait on the stack:

  \begin{itemize}
  \item \emph{Bound for overhead}.  There exists a
    bound function $B : \RTerm -> \nat$ such that, for any programs $\pls{e}$ and $\pls{e'}$ such that $\pls{e} -->^{*} \pls{e'}$,
    \(\spacep{\size}{\pls{e'}} \leq \spacep{\added{\ensuremath{\lambda \_ \mathbin{.}}} B(\pls{e})}{\pls{e'}}\).

    This bound is equivalent to requiring that for any evidence $\pls{\ev}$
    appearing in $\pls{e'}$, $\size(\pls{\ev}) \leq B(\pls{e})$.

  \item \emph{Bound for stack}. For any programs $\pls{e}$ and $\pls{e'}$
    such that
    $\pls{e} \replaced{\ensuremath -->}{\ensuremath -->^{*}} \pls{e'}$,  if there exists an evaluation context
    without nested ascriptions
    $\pls{G}$ such that
    $\pls{e} = \pls{G[ \cast{\ev}{\cast{\ev'}{e''}}]}$, then there must
    exist an $\pls{\ev''}$ such that $\pls{e'} = \pls{G[\cast{\ev''}{e''}]}$.
  \end{itemize}
\end{definition}

From this definition, a bound in the style of~\citet{herman10space}
follows as a corollary:

\begin{corollary}[Fixed overhead a-la-\citet{herman10space}]
  If  $\pls{e} -->^{*} \pls{e'}$ is bound for overhead
  \added{with an evidence bound function $B$}, then
  \[\spacep{\size}{\pls{e'}} \leq 3*B(\pls{e})* \spacep{0}{\pls{e'}}\]
\end{corollary}

\begin{proof}
  By induction over the structure of $\pls{e'}$.
\end{proof}

\rparagraph{Evidence requirements that make $\text{RL}^{+}$ space efficient}
The development of space-efficient semantics can be completed with
a new evidence abstraction for the $\text{RL}^{+}$ runtime semantics.
Before introducing the abstraction, we present all the properties we need from
this evidence.
$\text{RL}^{+}$ semantics is space efficient for \emph{any} abstraction whose
evidence composition satisfies the following two properties:

\begin{proposition}[Evidence composition is associative]
  \label{prop:ev-assoc}
  For any evidence $\botev_1$, $\botev_2$ and $\botev_3$,
  \[\left(\botev_1 \bottrans{} \botev_2\right) \bottrans{} \botev_3 = \botev_1 \bottrans{} \left( \botev_2 \bottrans{} \botev_3 \right)\]
\end{proposition}

\begin{proposition}[Evidence composition has a bound]
  \label{prop:ev-bound}
  Let $\ev_1$ and $\ev_2$ be any evidence objects appearing in program
  $\pls{e}$.  For any evidence objects
  $\ev'_1$ and $\ev'_2$ such that
  $\size(\ev'_1) \leq \size(\ev_1)$ and $\size(\ev'_2) \leq \size(\ev_2)$,
  $\size(\ev'_1 \trans{} \ev'_2) \leq B(\pls{e}).$
\end{proposition}

\begin{theorem}[Properties of $\text{RL}^{+}$]
  If evidence composition is associative and has a bound, then the semantics of $\text{RL}^{+}$ is
  space-efficient and is observationally equivalent to the semantics
  of $\text{RL}$.
\end{theorem}

\section{Precise Evidence for $\GTFLcsub$:
  Bounded Records and Rows}
\label{sec:brr}
Previous sections identify the need for a more precise abstraction for runtime
evidence.  This section delivers that abstraction, and finds a natural
characterization of precision that meets the paper's goals.
We devise a notion of evidence that admits an associative consistent
transitivity operation, and preserves the type invariants implied by the
concrete fragments of static subtyping that each evidence object represents.
This notion of evidence has some features in common with record-typing systems
from the literature, though used here for runtime type checking,
e.g.~\citep{remy89records}.
We replace gradual rows in evidence with a more precise gradual type
abstraction, which we call \emph{bounded records and rows} (BRR).  We prove
that evidence objects based on BRR are associative and precise using a general
concept from abstract interpretation called \emph{forward completeness}.

\subsection{Representing optional fields}

To motivate and introduce the structure of BRR, consider once again the failing
example from Sec.~\ref{ssec:shortcomings}.  The key observation in that example
is that in the following evidence composition,
\begin{displaymath}
  (\evpair{\rec{\x:\Int,\y:\Bool}}{\rec{\x:\Int}} \trans{}
  \evpair{\rec{\x:\Int}}{\rec{\?}}) = \evpair{\rec{\x:\Int,\y:\Bool}}{\rec{\?}}
\end{displaymath}
the second resulting gradual type represents too many possible static types, in
particular the spurious cases $\rec{y:\Bool}$ and $\rec{x:\Int,y:\Bool}$.  A
precise analysis of the two composed evidence objects admits only two possible
supertypes: $\rec{\;}$ or $\rec{x:\Int}$.

To faithfully capture this circumstance, BRR first introduces a distinction
between \emph{required} fields and \emph{optional} fields.  For instance, the
BRR type $\brec{\bounded{\x}{R}{\Int}}$ uses the $R$ annotation to indicate
that the $\x$ field is required: this BRR is equivalent to the $\rec{\x:\Int}$
gradual record type (and static record type).  In contrast, the BRR type
$\brec{\bounded{\x}{O}{\Int}}$ uses the $O$ annotation to indicate that the
$\x$ field is optional.  Given the appropriate definition of concretization for
BRR, we have ${\rec{\;} \in \gamma\left(\rec{\bounded{\x}{O}{\Int}}\right)}$
and
${\rec{\bounded{\x}{\;}{\Int}} \in
  \gamma\left(\rec{\bounded{\x}{O}{\Int}}\right)}$,
and no other static type.
This BRR type precisely captures the intended record types in the example
above, and as we might expect, the above evidence and composition can be
replaced by the precise BRR-based version:
\begin{displaymath}
  \left(
    \evpair{\rec{\bounded{\x}{R}{\Int},\bounded{\y}{R}{\Bool}}}{
      \rec{\bounded{\x}{R}{\Int}}}
    \trans{}
    \evpair{\rec{\bounded{\x}{R}{\Int}}}{\rec{\bounded{\x}{O}{\Int}}}
  \right)
  =
  \evpair{\rec{\bounded{\x}{R}{\Int},\bounded{\y}{R}{\Bool}}}{
    \rec{\bounded{\x}{O}{\Int}}}
\end{displaymath}

\added{Because the BRR abstraction is strictly more precise
  than the gradual row abstraction, the criteria set forth in
  Prop.~\ref{prop:wfev} leads to more well-formed evidence objects.
  As a result,
  ${\evpair{\rec{\bounded{\x}{R}{\Int}}}{\rec{\bounded{\x}{O}{\Int}}}}$
  is well formed and, unlike gradual rows, it precisely captures the static subtype
  instances of interest.
  }

\subsection{Representing absent fields}

Unfortunately, adding optional fields is not enough to achieve associativity.
Consider the following trio of BRR-based evidence:
\begin{displaymath}
  \evpair{\brow{\bounded{\x}{O}{\Int}}}{\brow{\bounded{\x}{O}{\Int}}}
  \trans{}
  \evpair{\brow{\bounded{\x}{O}{\Bool}}}{\brow{\bounded{\x}{O}{\Bool}}}
  \trans{}
  \evpair{\brec{\bounded{\x}{R}{\Int}}}{\brec{\;}}
\end{displaymath}
BRR must somehow represent rows to subsume the expressiveness of gradual row
types.  However, evaluating this composition reveals the need for more precise
information.
Composing left-to-right yields
$\evpair{\brow{\bounded{\x}{O}{\Int}}}{\brow{\bounded{\x}{O}{\Int}}}$
because composing the first two yields
$\evpair{\brow{\bounded{\x}{O}{\Int}}}{\brec{\?}}$ which composes with the
third to yield the above.  On the other hand, composing them right-to-left
fails!  The reason is that the requirements of the second and third evidence
objects are incompatible: the second can only have a $\Bool$-typed $\x$ field,
while the third demands an $\Int$-typed one.

Closer inspection of the first two evidence objects reveals that after
composing them, the resulting evidence object should have \emph{no $\x$ field}
in the super type position: the only instances of static subtyping that could
have transitively composed were those lacking an $\x$ field.  Then the
necessary absence of $\x$ is incompatible with requiring $\x$ to have $\Int$
type as in the third evidence.

This observation leads us to represent \emph{necessarily absent} fields.  Then
the result of composing the first two evidence objects above is:
$\evpair{\brow{\bounded{\x}{O}{\Int}}}{\brow{\missing{\x}}}$.  The BRR
${\brow{\missing{\x}}}$ represents the set of record types with any fields
\emph{except} $\x$.  Composing this with
$\evpair{\brec{\bounded{\x}{R}{\Int}}}{\brec{\;}}$ should always fail.
Adding absent and optional fields suffices to ensure that evidence objects
need never lose information about plausible subtype relations,
admitting implausible fields via the unknown row $\?$.

\subsection{Bounded Rows and Bounded Records, Formally}
\label{ssec:brrai}

We now give a full formal definition of Bounded Records and Rows
(BRR)\footnote{\added{While not \emph{Rows} in the sense of[Wand
    Mitchell '91], \emph{Bounded Rows} are derived from Gradual Rows introduced
    by~\cite{garciaAl:popl2016}.  We keep the \emph{Row} designation
    to highlight the connection with the latter paper.}}, as well
BRR-based evidence objects.
The syntax of BRR is follows:
\begin{displaymath}
  \begin{array}{rcll}
    \multicolumn{4}{c}{ %
      M \in \Mapping} \\[2mm]
    M & ::= & \left. \varnothing ~\middle|~ \cT_{R} ~\middle|~ \cT_{O} \right. & \text{
      (mappings for records and rows)}\\
    \cT & ::= & \left. \Bool ~\middle|~ \Int ~\middle|~ \cT -> \cT ~\middle|~ \?
      ~\middle|~ \brec{\mappings{i}{n}{\lx_i : M_i}}
      ~\middle|~ \brow{\mappings{i}{n}{\lx_i : M_i}} \right. & \text{(gradual
      types)}\\
  \end{array}
\end{displaymath}
As its name implies, the only interesting aspects of the BRR abstraction are
its record and row type definitions.  Each field maps to a type (qualified as
required $R$ or optional $O$), or to the absent field indicator
($\varnothing$).  Individual static record types are represented in BRR by
simply marking all fields as required.  BRR is a runtime abstraction, so its
annotations do not affect $\GTFLcsub$ syntax.

Though technically we need missing field annotations only for rows and not for
records, we opt for a uniform representation.  A mapping $\missing{\x}$ in a
record is redundant; similarly, a mapping $\bounded{\x}{O}{\?}$ is redundant in
a row. This redundancy has algorithmic advantages: we can extract a default
mapping for labels in either a row or a record, which simplifies inductive
rules in those cases where the domains of two types differ.

The notation $\mappings{i}{n}{\lx_i : M_i}$ is equivalent to the previous
notation $\overline{\lx_i : M_i}$, but avoids confusion when trying to
distinguish shared subdomains in two record-like
types. \added{Empty record and row types are still allowed
  (i.e. $n$ can be $0$).}

A concretization function (Fig.~\ref{fig:brr-gamma}) determines the meaning of
BRR types.  An intermediate notion of \emph{decompositions} simplifies our
definitions.  Decompositions encode uniformity restrictions on sets of static
types.  These restrictions are made explicit by use of the generator function
$\generator$.

The interesting equations are those for records,
which recursively describe sets of records that meet the constraints described
informally above.  Note that the concretization of non-empty bounded records is
described in terms of the concretization of smaller bounded records.

\begin{figure}
  \centering
  \begin{small}
    \flushleft\textbf{Decompositions}
    \begin{displaymath}
      \begin{array}{rcll}
        && C \in \Pow(\Type), C^{\varnothing} \in \Pow(\Type \cup \set{\varnothing}), d \in \Decomp \\
        d & ::= & \left. [[\Bool]] ~\middle|~ [[\Int]] ~\middle|~ \left\llbracket C -> C \right\rrbracket ~\middle|~ [[\?]]
          ~\middle|~ \left\llbracket \mappings{i}{n}{\lx_i : {C^{\varnothing}}_i} \right\rrbracket
          ~\middle|~ \left\llbracket \?\; \mappings{i}{n}{\lx_i : {C^{\varnothing}}_i} \right\rrbracket
        \right.
        \\
      \end{array}
    \end{displaymath}
    \flushleft$\boxed{ \generator : (\Decomp) -> \Pow(\Type)}$    
    \begin{gather*}
      \begin{array}{rcl}
        \generate{\Bool} &=& \set{\Bool}\\
        \generate{\Int} &=& \set{\Int}\\
        \generate{C_1 -> C_2} &=& \set{\left. T_1 -> T_2 ~\middle|~ T_1 \in C_1 \text{~and~} T_2 \in C_2 \right.}\\
        \generate{\mappings{i}{n}{\lx_i : {C^{\varnothing}}_i} }&=& \set{\left. \rec{\overline{\lx_i : T_i}} ~\middle|~ T_i \in {C^{\varnothing}}_i \right.} 
        \text{
          if
        } \varnothing \not \in {C^{\varnothing}_i} \text{
          for
          every
        } i
        \\
        \generate{\lx_n : {C^{\varnothing}}_n \; \mappings{i}{n - 1}{\lx_i : {C^{\varnothing}}_i} }&=&
        \generate{\mappings{i}{n-1}{\lx_i : {C^{\varnothing}}_i} } \bigcup     \generate{\lx_n : \left({C^{\varnothing}}_n \setminus \set{\varnothing} \right)  \; \mappings{i}{n - 1}{\lx_i : {C^{\varnothing}}_i} }
        \text{
          if
        } \varnothing \in {C^{\varnothing}}_n
        \\
        \generate{\?\;\mappings{i}{n}{\lx_i : {C^{\varnothing}}_i}}&=&
        \left\{\rec{\overline{\lx_i : T_i}\;\overline{\lx_j : T_j}} ~\middle|~ \begin{block}\rec{\overline{\lx_i : T_i}} \in \generate{\mappings{i}{n}{\lx_i : {C^{\varnothing}}_i}} \\ \text{
              and
            } \rec{\overline{\lx_j : T_j}} \in \generate{\mappings{j}{m}{\lx_j : \Pow\left(\Type \cup \set{\varnothing}\right)}}\\
            \added{\ensuremath{
                \text{ where
              } \overline{\lx_i} and \overline{\lx_j} \text{ are
                disjoint.}
          }}
          \end{block}\right\}
        \\ 
      \end{array}
    \end{gather*}
    \begin{displaymath}
      \begin{block}
        \boxed{\gamma : (\GType) -> \Pow^{+}(\Type)}\\[0.5em]
        \begin{array}{rcl}
          \gamma\left(\?\right) &=& \Type \\[0.5em]
          \gamma\left(\Bool\right) &=& \generate{\Bool}\\[0.5em]
          \gamma\left(\Int\right) &=& \generate{\Int}\\[0.5em]
          \gamma\left(\cT_1 -> \cT_2\right) &=& \generate{\gamma(\cT_1) -> \gamma(\cT_2)}\\[0.5em]
          \gamma\left(\left[\mappings{i}{n}{\lx_i : M_i}\;*\right]\right) &=& \generate{*\;\mappings{i}{n}{\lx_i : \gamma^M\left(M_i\right)}}
        \end{array}
      \end{block}
      \qquad
      \begin{block}
        \boxed{\gamma^M : (\Mapping) -> \Pow^{+}(\Type\cup\set{\varnothing})}\\[0.5em]
        \begin{array}{rcl}
          \gamma^M\left(\varnothing\right) &=& \set{\varnothing}\\[0.5em]
          \gamma^M\left(\cT_R\right) &=& \gamma(\cT) \\[0.5em]
          \gamma^M\left(\cT_O\right) &=& \gamma^M\left(\varnothing\right) \cup \gamma\left(\cT_R\right)\\[0.5em]
        \end{array}
      \end{block}
    \end{displaymath}
  \end{small}
  \caption{Decompositions and BRR Concretization Function}
  \label{fig:brr-gamma}
\end{figure}

Fig.~\ref{fig:brr-alpha} defines the corresponding abstraction for bounded
records and rows.  The equations distinguish between relevant sets of static
types.
This abstraction function has a subtlety regarding its domain (which we do
not explicitly name).  The domain cannot be arbitrary sets of types.  To see
why, consider the set $\set{\rec{\lx:\Int} | \lx \in \Label}$: what is the most
precise abstraction of this set when $\Label$ is an infinite set?  The answer
is that there is none, because this set features an infinite number of
``optional'' fields, and since bounded records only have finite fields, there
is no best representative.  The solution is to restrict the domain to
hereditarily admit only collections of records that (1) have non-$\?$ bound for
a finite set of field types; and (2) have either a finite set of potentially
present fields (abstracts to a bounded record) or a finite set of absent fields
(abstracts to a bounded row).  The image of concretization satisfies these
constraints, as do our operations on evidence.  These finitary restrictions are
analogous to the restrictions on the open sets of an infinite product
topology~\cite{munkres}.

\begin{figure}
  \begin{small}
  \begin{displaymath}
    \begin{array}{rcl}
      \alpha\left(\generate{\Bool}\right) &=& \Bool \\[0.5em]
      \alpha\left(\generate{\Int}\right) &=& \Int \\[0.5em]
      \alpha\left(\generate{C_1 -> C_2}\right) &=&
      \alpha\left(C_1\right) ->
      \alpha\left(C_2\right) \\[0.5em]
      \alpha\left(\generate{*\;\mappings{i}{n}{\lx_i : {C^{\varnothing}}_i}}\right) &=& \left[\mappings{i}{n}{\lx_i : \alpha^M\left({C^{\varnothing}}_i\right)} \;*\right] \\[0.5em]
      \alpha\left(\varnothing\right) &\multicolumn{2}{l}{\text{undefined}}\\[0.5em]
      \alpha\left(C\right) &=& \? \text{ otherwise}\\[0.5em]
    \end{array}
    \qquad
    \begin{block}
      \begin{array}{rcl}
        \alpha^M\left(\set{\varnothing}\right) &=& \varnothing \\[0.5em]
        \alpha^M\left(\set{\varnothing}\cup C \right) &=& (\alpha\left(C\right))_O \text{
          if }C\text{ is not empty}\\[0.5em]
        \alpha^M\left(C \right) &=& (\alpha\left(C\right))_R \text{
          if }\varnothing \notin C\\[0.5em]
      \end{array}    \end{block}
  \end{displaymath}
  \end{small}
  \caption{BRR Abstraction Function}
  \label{fig:brr-alpha}
\end{figure}

\begin{theorem}
  $\alpha$ and $\gamma$ for bounded records and rows form a Galois Connection.
\end{theorem}
\begin{proof}
  Consequence of Soundness and Optimality Lemmas (see
  \replaced{technical report~\cite{techreport}}{Appendix~\ref{apdx:brr}}).
\end{proof}

\rparagraph{Well-formed evidence}
Using bounded records and rows, we develop a refined notion of evidence for
$\GTFLcsub$.  Fig.~\ref{fig:brr-wf-ev} defines the inductive structure of
well-formed BRR evidence.  Its structure is somewhat analogous to the
original, but with richer distinctions.

\begin{figure}
  \begin{small}
    \flushleft\boxed{|- \ev \wf}
    \quad\textbf{Well-formed Evidence}\\
    \begin{mathpar}
      \inference{\cT \in \set{\Int,\Bool,\?}}{
        |- \braket{\cT,\cT} \wf
      }
      \and
      \inference{
        |- \braket{\cT_{21},\cT_{11}} \wf &
        |- \braket{\cT_{12},\cT_{22}} \wf
      }{
        |- \braket{\cT_{11} -> \cT_{12},\cT_{21} -> \cT_{22}} \wf
      }
      \and
      \begin{array}{c}
      \added{
        \ensuremath{
          D(\cdot) = \varnothing }}
      \\
      \added{
        \ensuremath{
          D(\?) = \?_O }}
      \end{array}
      \and
      \inference{
        \braket{*_1,*_2} \neq \braket{\cdot,\?} &
        \forall i, |- \braket{M_{i1},M_{i2}} \wf &
        \forall j, |- \braket{M_j,D(*_2)} \wf &
        \forall k, |- \braket{D(*_1), M_k} \wf \\
        \added{\ensuremath{\text{where } \overline{\lx_i},
            \overline{\lx_j}\text{, and } \overline{\lx_k} \text{are disjoint.}}}
      }{
        |- \left\langle \left[\mappings{i}{n}{\lx_i : M_{i1}}
            \mappings{j}{m}{\lx_j : M_j}\;*_1\right],
          \left[\mappings{i}{n}{\lx_i : M_{i2}} \mappings{k}{o}{\lx_k
              : M_k}\;*_2\right]\right\rangle \wf
      }
    \end{mathpar}\\
    \flushleft\boxed{|- \braket{M,M} \wf}
    \quad\textbf{Well-formed Mappings}\\
    \begin{mathpar}
      \inference{}{
        |- \braket{M,\varnothing} \wf
      }
      \and
      \inference{
        |- \braket{\cT_{1},\cT_{2}} \wf 
      }{
        |- \braket{\left(\cT_1\right)_R,\left(\cT_{2}\right)_R} \wf
      }
      \and
      \inference{
        \cT_1 <= \cT_3 &
        |- \braket{\cT_1,\cT_2} \wf
      }{
        |- \braket{\left(\cT_3\right)_{*} ,\left(\cT_2\right)_O} \wf
      } 
    \end{mathpar}
  \end{small}%
  \caption{$\BRRcsub$'s definition of well-formed evidence}
  \label{fig:brr-wf-ev}
\end{figure}

\subsection{Absent labels enable sound optimizations}

We motivate the introduction of bounded records and rows to solve associativity
issues and to guarantee that programs with inconsistent ascriptions always
fail.  Associativity more broadly could also support program optimizations such
as inlining and pre-composing evidence objects.  Consider the following
program:

\begin{displaymath}
  (\lambda x : \rec{\?}.\,
  (\<if> \ttt \<then>
  (\<if> \ttt \<then> x \<else> x :: \rec{ l : \Int,\?}) \<else>
  x :: \rec{l : \Bool,\?}))
\end{displaymath}

This program relies on $\<if>\!\!$ branching to generate uses of consistent
subtype join which will generate evidence objects that use the optional
annotation.  This program will produce an error when given as an argument any
record that has an $l$ mapping.  This is because the ascriptions impose
inconsistent constraints: the function body must have type
$\brow{\bounded{l}{O}{\Bool}}$ in the outermost $\<if>$, and
$\rec{\bounded{l}{O}{\Int}\?}$ in the innermost $\<if>$, thus
$x$ must go through the composition of evidence on both types, reaching an
inconsistency whenever the label is present.

While the sole introduction of optional fields would suffice to run this
program properly, consider now an optimizing compiler that performs constant
propagation in this function.  The body of the function might then be optimized
to \(\lambda x . \evcast{\ev_2}{\evcast{\ev_1}{x}}\), with
$\ev_2 = \evpair{\brow{\bounded{l}{O}{\Bool}}}{\brow{\bounded{l}{O}{\Bool}}}$
and
$\ev_1 = \evpair{\brow{\bounded{l}{O}{\Int}}}{\brow{\bounded{l}{O}{\Int}}}$.  A
more advanced optimizing compiler could try to perform this evidence
composition ahead of time. Unfortunately, unless we introduce absent labels as
in BRR, the only possible composition would be to coalesce the label $l$ into
the row portion to generate the evidence pair
$\evpair{\rec{\vert\?}}{\rec{\vert\?}}$. This ``optimization'' changes the
behaviour of our program, as arguments with a mapping for $l$ will now be
accepted instead of producing an error.  To achieve full associativity and
soundness in the presence of optimizations like the above, we need absent
labels.

\section{Forward Completeness as a key to associativity}
\label{sec:forward-completeness}

The previous section delves into the shortcomings of $\GTFLcsub$'s evidence
abstraction, diagnoses some evident information loss, and devises a new
abstraction that retains the relevant information.  Are these improvements
sufficient?  An example-driven approach can drive us closer to a solution, but
ultimately we need more rigorous and comprehensive confirmation, which we now
provide.  Moreover, we do so by generalizing beyond $\GTFLcsub$, seeking
sufficient criteria that can apply to future applications of AGT,
regardless of the particulars of the type discipline or gradualization.

First, consider associativity.  We must prove that
$\ev_1 \trans{} \left( \ev_2 \trans{} \ev_3 \right) = \left( \ev_1 \trans{} \ev_2 \right) \trans{} \ev_3$
holds for our new notion of evidence.  A direct proof of this property is
possible, but does not scale well.  Each evidence object $\ev$ is a pair of BRR
types, and associativity must consider not only how they compose with one
another, but also with the results of intermediate compositions.  The case
explosion is staggering.

Even after proving associativity, how does one confirm that the evidence
abstraction is precise enough to protect type-based invariants?
Is a second complex formal development necessary?

Fortunately no: there is another way.  In the context of abstract
model-checking, \citet{giacobazzi01incomplete} introduce the concept of
\emph{forward completeness},%
\footnote{Sometimes called gamma-completeness or exactness} which is dual to
the concept of \emph{(backward) completeness}
that arises more naturally for abstract interpretation-based static analysis.
The idea applies to any abstract function, but we present it here in
domain-specific terms.

\begin{definition}[Forward completeness]
  Let $\ev$ denote evidence objects, $\trans{}$ denote consistent transitivity,
  and $\relcomp$ denote relational composition over fragments of static subtyping.
  Then $\trans{}$ is \emph{forward complete} with respect to its evidence
  abstraction if
  \begin{math}
    \gamma^{<:}(\ev_1) \relcomp \gamma^{<:}(\ev_2) =
    \gamma^{<:}(\ev_1 \trans{} \ev_2)
  \end{math}
  for any two evidence objects $\ev_1$ and $\ev_2$.
\end{definition}

Soundness of $\trans{}$ with respect to relational composition $\relcomp$
implies
${\gamma^{<:}(\ev_1) \relcomp \gamma^{<:}(\ev_2) \subseteq
  \gamma^{<:}(\ev_1 \trans{} \ev_2)},$
which ensures that $\trans{}$, operating on abstract evidence objects,
sufficiently overapproximates the behavior of $\relcomp$ on the meanings of
those objects.  Forward completeness implies that the reverse-containment also
holds, which means that $\trans{}$ \emph{exactly} approximates $\relcomp$ for
abstract objects.  This means that $\trans{}$ is a perfect stand-in for
$\relcomp$ \emph{if we need only consider sets in the image of evidence
  concretization}.  In the case of AGT, such is exactly the case: evidence is
initially introduced in terms of source gradual types, which are even less
precise than evidence objects.

\begin{lemma}[Relational Composition is closed wrt concretization]
  \label{lem:rc-closed-wrt-gamma}
  For any $|- \ev_1 \wf$ and $|- \ev_2 \wf$, either
  \(\gamma^{<:}\left(\ev_1\right) \relcomp \gamma^{<:}\left(\ev_2\right) = \emptyset\),
  or there
  exists an
  $|- \ev_3 \wf$ such that
  \begin{math}
    \gamma^{<:}\left(\ev_1\right) \relcomp \gamma^{<:}\left(\ev_2\right) =
    \gamma^{<:}\left(\ev_3\right).
    \end{math}
\end{lemma}
\begin{proof}
  By double induction over 
  $|- \ev \wf$ and $|- \braket{M,M} \wf$ for $\ev_1$ and $\ev_2$.
\end{proof}

\begin{theorem}
  \label{thm:brr-fc}
  In BRR, $\trans{}$ is forward complete.
\end{theorem}
\begin{proof}
  Consequence of Lemma~\ref{lem:rc-closed-wrt-gamma}.
\end{proof}

It is straightforward to prove that forward completeness suffices for
space-efficiency.

\begin{theorem}
  \label{thm:fc-assoc}
  If\/ \replaced{$\trans{}$}{$(\Ev{<:}, \trans{})$} is forward-complete
  \added{with respect to its evidence abstraction,} then
  \begin{math}
    \ev_1 \trans{} \left( \ev_2 \trans{} \ev_3 \right) =
    \left( \ev_1 \trans{} \ev_2 \right) \trans{} \ev_3
  \end{math}
  for any evidence objects $\ev_1$, $\ev_2$ and $\ev_3$.
\end{theorem}
\begin{proof}
  A generalized Coq mechanization of this
  property can be found in~\citet{agt-moving-forward-zenodo}.
\end{proof}

The basic intuition is that relational composition is associative, so a
forward-complete definition must also be.  Of more practical interest, though,
is that the statement of forward completeness involves only two evidence
objects, compared to the three needed to state associativity.  Thus, proving
forward completeness requires far fewer cases than proving associativity
directly.  Additionally, forward completeness ensures that reasoning about
transitivity of subtyping is as precise as operating directly with fragments of
subtyping, as alluded to in Sec.~\ref{sec:agt}.  This precision matters most
when considering precise enforcement of type invariants.

\section{BRR is Space-Efficient}
\label{sec:brr-space-efficient}
As the previous section proved that evidence composition is
associative, we must only find a bound for evidence composition in the
style of~\citet{herman10space} to guarantee that BRR is space
efficient. Since an evidence object is just a pair of types, we can
use the size of a gradual type as a proxy for the size of evidence.
Fig.~\ref{fig:size_height_cT} defines structurally the size and the
height of a gradual type.  The size of a gradual type can be bound in
terms of its height, and thus we can find a bound for the size of an
evidence object in terms of the height of its components.

\begin{figure}
  \begin{small}
  \begin{displaymath}
    \begin{block}
      \boxed{\size(\cT)}\quad 
      \size : \GType -> \nat \\[0.5em]
      \size(\Int) = \size(\Bool) = \size(\?) = 1 \\[0.5em]
      \size(\cT_1 -> \cT_2) = 1 + \size(\cT_1) + \size(\cT_2) \\[0.5em]
      \size\left(\left[\mappings{i}{n}{l_i:\cT_i} \;*\right]\right) = 1 + \sum\limits_{i=0}^n \size(\cT_i)
    \end{block}
    \qquad
    \begin{block}      
      \boxed{\height(\cT)}\quad 
      \height : \GType -> \nat \\[0.5em]
      \height(\Int) = \height(\Bool) = \height(\?) = 1 \\[0.5em]
      \height(\cT_1 -> \cT_2) = 1 + \max \left(\height(\cT_1)\right) \left(\height(\cT_2)\right) \\[0.5em]
      \height\left(\left[\mappings{i}{n}{l_i:\cT_i}\;*\right]\right) = 1 + \max_{i=0}^n \height(\cT_i)
    \end{block}
  \end{displaymath}
  \end{small}
  \caption{Size and height of a gradual type}
  \label{fig:size_height_cT}
\end{figure}

\begin{lemma}[Bound for size of Gradual Types in BRR]
  For any gradual type $\cT$,
  \[\size(\cT) \leq (3 + \ldom(\cT))^{1 + \height(\cT)}\]

  where $\ldom$ corresponds to the size of the label domain for
  records in type $\cT$.
\end{lemma}

While this bound seems quite high, it is comparable to the bound
established by~\citet{herman10space} for the size of coercions: Their structural types only
include functions and references, so they
only needed to consider binary tree branching at any level in their
types.  Once we introduce records, we must instead deal with arbitrarily-wide n-ary tree
branching.

\begin{definition}[Size and height of evidences in BRR]
  \[\size(\ev) = \size(\pi_1(\ev)) + \size(\pi_2(\ev))\]
  \[\height(\ev) = \max \height(\pi_1(\ev)) \height(\pi_2(\ev))\]
  \[\ldom(\ev) = \max \ldom(\pi_1(\ev)) \ldom(\pi_2(\ev))\]
\end{definition}

Since BRR is forward-complete, we can rely on
Prop.~\ref{prop:ctrans-iml} to reach a bound for evidence
composition through a combination of bounds for gradual meet and 
initial evidence:

\begin{lemma}[Height bounds in BRR]
  \label{lem:height-brr}
  \begin{math}\height(\cT_1 \meet \cT_2) \leq \max \left(\height(\cT_1)\right) \left(\height(\cT_2)\right);\end{math}
  \[\height(\Isub|[ \cT_1 \csub \cT_2 |]) \leq \max \left(\height(\cT_1)\right) \left(\height(\cT_2)\right);\]
  \[\height(\ev_1 \trans{} \ev_2) \leq \max \left(\height(\ev_1)\right) \left(\height(\ev_2)\right).\]
\end{lemma}
We can use these lemmas to prove the following bound for evidence
composition:
\begin{theorem}[Bound for the size of evidence composition]
  \label{thm:bound-brr}
  \[\size(\ev_1 \trans{} \ev_2) \leq 2 * \left(3 + \max \left(\ldom(\ev_1)\right) \left(\ldom(\ev_2)\right)\right)^{1 + \max \left(\height(\ev_1)\right) \left(\height(\ev_2)\right)}\]
\end{theorem}

\begin{proof}
  A Coq proof for Lemma~\ref{lem:height-brr} and
  Theorem~\ref{thm:bound-brr}~can be found in \citet{agt-moving-forward-zenodo}.
\end{proof}

Since any evidence composition is bound by the height and label domain
of its components, we can establish a bound for any evidence
composition happening in a program by collecting the greatest height
and domain, respectively:

\begin{theorem}[Bound for evidence compositions in a program]
  Let $\rr$ be a program.  There exists an $\ev_h$ appearing in a
  subterm of $\rr$, such
  that for any evidences $\ev_1$ and $\ev_2$ appearing in any subterms
  of $\rr$,
  \[\size(\ev_1 \trans{} \ev_2) \leq 2 * \left(3 + \ldom(\rr)\right)^{1 +\height(\ev_h)}\]

  where $\ldom(\rr)$ denotes the domain of labels for evidence appearing in
  subterms of $\rr$.
  
\end{theorem}

We can use this bound as a proof of
Prop.~\ref{prop:ev-bound}. Combining this with
Theorems~\ref{thm:brr-fc} and~\ref{thm:fc-assoc} as a proof of
Prop.~\ref{prop:ev-assoc}, thus we have a space efficient semantics.

\section{Discussion}

As part of their analysis of counter-example guided abstraction refinement,
\citet{giacobazzi01incomplete} develop a constructive approach to
systematically refine abstractions to be complete with respect to an operation.
Though we developed BRR manually, future AGT-based designs could employ their
techniques to more systematically produce a forward complete evidence
abstraction starting from the default proposed by \citet{garciaAl:popl2016}.

Space efficiency is but one operational criteria of gradually-typed languages.
In future work we will develop generalized techniques for reducing the density
of evidence objects in programs.  Using AGT off-the-shelf leads to a semantics
that saturates even precisely typed code with evidence.  We believe that the
refined notions of evidence in this work may help with understanding how to
minimize the amount of runtime checks in programs.

As has been observed by others~\citep{toro18secref,toro19parametricity}, the
guarantees provided automatically by AGT are not sufficient to yield a
satisfactory gradually typed language.  Strong type-based reasoning may
necessitate additional tuning and proof obligations.  The two works cited above
demonstrate this in the context of relational type invariants like
parametricity and information-flow control.  Furthermore, the evidence
composition developed by \citet{toro18secref} is associative, but this alone
did not suffice to achieve their noninterference result.  This is the first
work to demonstrate that an AGT-based gradual type system that enforces
\emph{non-relational} type invariants can need additional tuning.  However, in
the non-relational case, forward-completeness suffices.  Nonetheless, these
experiences suggest a desire for additional properties beyond those proposed
by~\citet{siek15criteria}.  It is worth noting that blame, another facet of
gradual typing, does not appear to help here: blame theorems do not indicate
which programs should fail, only which parties may or may not be responsible
when they do.

\added{
  In this paper we have grounded our exposition in the
  concrete example of Bounded Rows and Records (BRR). While
  forward-completeness guarantees associativity of composition, and
  thus suffices for an observationally equivalent
  semantics that is \emph{bound for stack}, it does not guarantee
  space-efficiency on its own, as one must still prove that evidence
  composition is itself bounded. In the particular case of BRR, this
  last proof is supported by the fact that the underlying static
  subtyping is a transitive and syntax-directed relation. Our work
  should transfer to other typing disciplines with these constraints,
  and provides more precise guidelines for future AGT developments
  to systematically provide space-efficiency.  To the best of our
  knowledge, we are the first to prove space
  efficiency in a gradual language with records (and record
  subtyping).
}

\section{Conclusion}

Abstracting gradual typing provides a powerful scaffold, based in the theory of
abstract interpretation, that helps designers construct gradually typed
variants of static type disciplines.  This paper draws on additional concepts
from abstract interpretation to inform and improve those designs.  In doing so,
we provide a deeper analysis of the concept of evidence in AGT, and devise a
general, though not necessarily universally feasible, tool for constructing
space-efficient semantics that precisely enforce the invariants implied by
source gradual types.
We propose that, regardless of the source gradual type abstraction, a
forward-complete runtime abstraction should be the new default when possible.
Space-efficiency, and \deleted{predictable} precise runtime semantics are sufficient
motivation.

\begin{acks}

  \added{
  We acknowledge the support of the Natural Sciences and Engineering
  Research Council of Canada \grantsponsor{NSERC}{(NSERC)}{}.
  Cette recherche a \'et\'e financ\'ee par le Conseil de recherches en
  sciences naturelles et 
  en g\'enie du Canada \grantsponsor{CRSNG}{(CRSNG)}{}.
  }
  \added{
  The authors also thank Phil Wadler, Yuchong Pan, Peter
  Thiemann, and the anonymous reviewers.
}

\end{acks}

\clearpage
\bibliography{references}{}

\clearpage
\appendix

\nobalance

\section*{About these Appendixes}

These appendixes provide a variety of definitions and proofs that could not
appear in the paper for lack of space.  Not all of the proofs undertaken in
this work have been typeset.

\newcommand{\tlet}{\ensuremath{\mathsf{let\;}}}
\newcommand{\tin}{\ensuremath{\mathsf{\;in\;}}}

\section{$\tlet$ Typing and Translation Judgments}

\begin{small}
  \begin{mathpar}
    \inference[($\cT$let)]{
      \Gamma |- t_1 : \cT_1 & 
      \cT_1 \csub \cT \\
      \Gamma, x:\cT |- t_2 : \cT_2
    }{
      \Gamma |- \tlet x : \cT = t_1 \tin t_2 : \cT_2
    }
    \and
    \inference[($\cT$let)]{
      \rr_1 \derivof \Gamma |- t_1 : \cT_1 & 
      \ev_1 \evidenceof \cT_1 \csub \cT \\
      \rr_2 \derivof \Gamma, x:\cT |- t_2 : \cT_2
    }{
      \tlet x = \cast{\ev_1}{\rr_1} \tin \rr_2 \derivof
      \Gamma |- \tlet x : \cT = t_1 \tin t_2 : \cT_2
    }
    \and
    \inference[($\leadsto$let)]{
      \Gamma |- t_1 \leadsto \rr_1 : \cT_1 & 
      \ev_1 = \Isub|[ \cT_1 \csub \cT  |] \\
      \Gamma, x:\cT |- t_2 \leadsto \rr_2 : \cT_2
    }{
      \Gamma |- \tlet x : \cT = t_1 \tin t_2
      \leadsto \tlet x = \cast{\ev_1}{\rr_1} \tin \rr_2 : \cT_2
    }
  \end{mathpar}
\end{small}

\section{$\GTFLcsub$ Operator Definitions}
\label{sec:ctrans}

This section presents the full algorithmic definitions for some $\GTFLcsub$
operators that were omitted from the paper for lack of space.

Fig.~\ref{fig:gtfl-csubjoin} presents consistent subtype join and meet, which
are used to define the type judgment for source and runtime
programs.

\begin{figure}
  \centering
  \begin{small}
    \begin{displaymath}
      \begin{block}
        \boxed{\cT \csubjoin \cT}\quad \textbf{Consistent Subtype Join}\\[0.5em]
        \csubjoin : \GType \times \GType \rightharpoonup \GType \\[0.5em]
        \cT_1 \csubjoin \cT_2 = \cT_2 \csubjoin \cT_1 \\[0.5em]
        \? \csubjoin \? = \? \\[0.5em]
        \Int \csubjoin \Int = \Int \\[0.5em]
        \Int \csubjoin \? = \Int \\[0.5em]
        \Bool \csubjoin \Bool = \Bool \\[0.5em]
        \Bool \csubjoin \? = \Bool \\[0.5em]
        (\cT_{11} -> \cT_{12}) \csubjoin (\cT_{21} -> \cT_{22}) =\\ 
        \hspace{2cm}
        (\cT_{11} \csubmeet \cT_{21}) -> (\cT_{12} \csubjoin \cT_{22}) \\[0.5em]
        (\cT_{11} -> \cT_{12}) \csubjoin \? = 
        (\cT_{11} -> \cT_{12}) \csubjoin (\? -> \?) \\[0.5em]
        \rec{\overline{l_i:\cT_{i1}},*} \csubjoin \? = 
        \rec{\overline{l_i:\cT_{i1}},*} \csubjoin \rec{\?} \\[0.5em]
        \rec{\overline{l_i:\cT_{i1}},\overline{l_j:\cT_j}} \csubjoin
        \rec{\overline{l_i:\cT_{i2}},\overline{l_k:\cT_k}} = \\
        \hspace{4cm}
        \rec{\overline{l_i:\cT_{i1} \csubjoin \cT_{i2}}} \\[0.5em]
        \rec{\overline{l_i:\cT_{i1}}, *} \csubjoin
        \rec{\overline{l_i:\cT_{i2}},\overline{l_k:\cT_k},\?} = 
        \rec{\overline{l_i:\cT_{i1} \csubjoin \cT_{i2}}, *} \\[0.5em]
        \rec{\overline{l_i:\cT_{i1}},\overline{l_j:\cT_j}^{+},*} \csubjoin
        \rec{\overline{l_i:\cT_{i2}},\overline{l_k:\cT_k},\?} = \\
        \hspace{4cm}
        \rec{\overline{l_i:\cT_{i1} \csubjoin \cT_{i2}},\?} \\[0.5em]
        \\ 
        \cT \csubjoin \cT \text{ undefined otherwise} \\[0.5em]
      \end{block}  
      \qquad
      \begin{block}
        \boxed{\cT \csubmeet \cT}\quad \textbf{Consistent Subtype Meet}\\[0.5em]
        \csubmeet : \GType \times \GType \rightharpoonup \GType \\[0.5em]
        \cT_1 \csubmeet \cT_2 = \cT_2 \csubmeet \cT_1 \\[0.5em]
        \? \csubmeet \? = \? \\[0.5em]
        \Int \csubmeet \Int = \Int \\[0.5em]
        \Int \csubmeet \? = \Int \\[0.5em]
        \Bool \csubmeet \Bool = \Bool \\[0.5em]
        \Bool \csubmeet \? = \Bool \\[0.5em]
        (\cT_{11} -> \cT_{12}) \csubmeet (\cT_{21} -> \cT_{22}) =\\ 
        \hspace{2cm}
        (\cT_{11} \csubjoin \cT_{21}) -> (\cT_{12} \csubmeet \cT_{22}) \\[0.5em]
        (\cT_{11} -> \cT_{12}) \csubmeet \? = 
        (\cT_{11} -> \cT_{12}) \csubmeet (\? -> \?) \\[0.5em]
        \rec{\overline{l_i:\cT_{i1}},*} \csubmeet \? = 
        \rec{\overline{l_i:\cT_{i1}},*} \csubmeet \rec{\?} \\[0.5em]
        \rec{\overline{l_i:\cT_{i1}},\overline{l_j:\cT_j}} \csubmeet
        \rec{\overline{l_i:\cT_{i2}},\overline{l_k:\cT_k}} =\\
        \hspace{2cm}
        \rec{\overline{l_i:\cT_{i1} \csubmeet \cT_{i2}},
          \overline{l_j:\cT_j},\overline{l_k:\cT_k}}
        \\[0.5em]
        \rec{\overline{l_i:\cT_{i1}},\overline{l_j:\cT_j},\?} \csubmeet
        \rec{\overline{l_i:\cT_{i2}},\overline{l_k:\cT_k}} = \\
        \hspace{2cm}
        \rec{\overline{l_i:\cT_{i1} \csubmeet \cT_{i2}},
          \overline{l_j:\cT_j},\overline{l_k:\cT_k \csubmeet \?},\?}\\[0.5em]
        \rec{\overline{l_i:\cT_{i1}},\overline{l_j:\cT_j},\?} \csubmeet
        \rec{\overline{l_i:\cT_{i2}},\overline{l_k:\cT_k},\?} = \\
        \hspace{2cm}
        \rec{\overline{l_i:\cT_{i1} \csubmeet \cT_{i2}},
          \overline{l_j:\cT_j \csubmeet \?},
          \overline{l_k:\cT_k \csubmeet \?},\?}\\[0.5em]
        \cT \csubmeet \cT \text{ undefined otherwise} \\[0.5em]
      \end{block}  
    \end{displaymath}
  \end{small}
  \caption{$\GTFLcsub$: Consistent Subtype Extrema}
  \label{fig:gtfl-csubjoin}
\end{figure}

Fig.~\ref{fig:gtfl-initial-evidence} introduces an Inductive
Definition of Initial Evidence, and Fig.~\ref{fig:meet} introduces
an Inductive Definition of gradual meet.  With these two, given
Prop.~\ref{prop:ctrans-iml}, we can produce an inductive
definition of consistent transitivity.

\begin{figure}
  \begin{small}
    \centering
    \flushleft\boxed{\cT \meet \cT}\quad\textbf{Gradual Meet}
    \begin{align*}
      \cT_1 \meet \cT_2 &= \cT_2 \meet \cT_1 \\
      \? \meet \? &= \? \\
      \Int \meet \Int &= \Int \\
      \Bool \meet \Bool &= \Bool\\
      \cT \meet \? &= \cT \\
      (\cT_{11} -> \cT_{12}) \meet (\cT_{21} -> \cT_{22}) &= 
      (\cT_{11} \meet \cT_{21}) -> (\cT_{12} \meet \cT_{22}) \\
      \rec{\overline{l_i:\cT_{i1}},*_1} \meet \rec{\overline{l_i:\cT_{i2}},*_2}
      &= \rec{\overline{l_i:\cT_{i1} \meet \cT_{i1}},{*_1} \meet {*_2}} \qquad
      {*_1} \meet {*_2} =
      \begin{cases}
        ? & \braket{*_1,*_2} = \braket{?,?} \\
        \emptyset & \text{otherwise}
      \end{cases}
      \\
      \rec{\overline{l_i:\cT_{i1}},\overline{l_j:\cT_j}^{+},*_1} \meet
      \rec{\overline{l_i:\cT_{i2}},\?}
      &= 
      \rec{\overline{l_i:\cT_{i1}},\overline{l_j:\cT_j}^{+},*_1} \meet
      \rec{\overline{l_i:\cT_{i2}},\overline{l_j:\?}^{+},\?}
      \\
      \rec{\overline{l_i:\cT_{i1}},\overline{l_j:\cT_j}^{+},\?} \meet
      \rec{\overline{l_i:\cT_{i2}},\overline{l_k:\cT_k}^{+},\?}
      &= 
      \rec{\overline{l_i:\cT_{i1}},\overline{l_j:\cT_j}^{+},
        \overline{l_k:\?}^{+},\?} \meet
      \rec{\overline{l_i:\cT_{i2}},\overline{l_j:\?}^{+},
        \overline{l_k:\cT_k}^{+},\?}
      \\
      \cT_1 \meet \cT_2 &\;\text{undefined otherwise}
    \end{align*}
  \end{small}
  \caption{Gradual Meet}
  \label{fig:meet}
\end{figure}

\begin{figure}
  \begin{small}
    \flushleft\boxed{\Isub|[ \cT \csub \cT |] :
      \GType \times \GType \rightharpoonup \Ev{}}
    \quad\textbf{Initial Evidence}
    \begin{align*}
      \Isub|[ \cT \csub \cT |] &= \evpr{\cT,\cT} & \cT \in \set{\Int,\Bool,\?} \\
      \Isub|[ \cT \csub \? |]  &= \evpr{\cT,\cT} &
      \cT \in \set{\Int,\Bool,\rec{*}} \\
      \Isub|[ \? \csub \cT |] &= \evpr{\cT,\cT} & \cT \in \set{\Int,\Bool} \\
      \Isub|[ \cT_{11} -> \cT_{12} \csub \?|] &= 
      \Isub|[ \cT_{11} -> \cT_{12} \csub  \? -> \?|] \\
      \Isub|[ \? \csub \cT_{21} -> \cT_{22} |] &=
      \Isub|[ \? -> \? \csub \cT_{21} -> \cT_{22} |] \\
      \Isub|[\cT_{11} -> \cT_{12} \csub \cT_{21} -> \cT_{22} |] &= 
      \evpr{\cT'_{11} -> \cT'_{12}, \cT'_{21} -> \cT'_{22}} &
      \begin{block}
        \Isub|[\cT_{21} \csub \cT_{11}|] = \evpr{\cT'_{21},\cT'_{11}} \\
        \Isub|[\cT_{12} \csub \cT_{22}|] = \evpr{\cT'_{12},\cT'_{22}}
      \end{block}\\
      \Isub|[ \? \csub \rec{\overline{l_i:\cT_i},*} |] &= 
      \Isub|[ \rec{\?} \csub \rec{\overline{l_i:\cT_i},*} |] \\
      \Isub|[ \rec{\overline{l_i:\cT_i}^{+},*} \csub \? |] &= 
      \evpr{\rec{\overline{l_i:\cT_i}^{+},*},\rec{\?}} \\
      \Isub|[ \rec{\overline{l_i:\cT_{i1}}} \csub
      \rec{\overline{l_i:\cT_{i2}},*} |] &=
      \evpr{\rec{\overline{l_i:\cT'_{i1}}}, \rec{\overline{l_i:\cT'_{i2}}}}&
      \overline{\Isub|[ \cT_{i1} \csub \cT_{i2} |] = \evpr{\cT'_{i1},\cT'_{i2}}}\\
      \Isub|[ \rec{\overline{l_i:\cT_{i1}},\?} \csub
      \rec{\overline{l_i:\cT_{i2}},*} |] &=
      \evpr{\rec{\overline{l_i:\cT'_{i1}},\?}, \rec{\overline{l_i:\cT'_{i2}},*}} &
      \overline{\Isub|[ \cT_{i1} \csub \cT_{i2} |] = \evpr{\cT'_{i1},\cT'_{i2}}} \\
      \Isub|[ \rec{\overline{l_i:\cT_{i1}},\overline{l_j:\cT_j}^{+},*_1} \csub
      \rec{\overline{l_i:\cT_{i2}},*_2} |] &=
      \evpr{\rec{\overline{l_i:\cT'_{i1}},\overline{l_j:\cT_j}^{+},*_1},
        \rec{\overline{l_i:\cT'_{i2}},*_2}}
      &
      \overline{\Isub|[ \cT_{i1} \csub \cT_{i2} |] = \evpr{\cT'_{i1},\cT'_{i2}}} \\
      \Isub|[ \rec{\overline{l_i:\cT_{i1}},\overline{l_j:\cT_j},\?} \csub
      \rec{\overline{l_i:\cT_{i2}},\overline{l_k:\cT_k}^{+},*} |] &=\\
      &\hspace{-1.5in}
      \Isub|[ \rec{\overline{l_i:\cT_{i1}},\overline{l_j:\cT_j},
        \overline{l_k:\?}^{+},\?} \csub
      \rec{\overline{l_i:\cT_{i2}},\overline{l_k:\cT_k}^{+},*} |] \\
      \Isub|[ S_1 \csub S_2 |] &\text{ undefined otherwise}
    \end{align*}
  \end{small}
  \caption{$\GTFLcsub$: Definition of Initial Evidence}
  \label{fig:gtfl-initial-evidence}
\end{figure}

Figures~\ref{fig:gr_ctrans1}~and~\ref{fig:gr_ctrans2} present a full direct
equational definition of consistent transitivity for the original gradual
rows-based representation.  This definition was calculated from the AGT-based
definition, then used to prove Prop.~\ref{prop:ctrans-iml}.

\begin{figure}[h]
  \centering
  \begin{small}
    \flushleft\boxed{\trans{} : \Ev{} \times \Ev{} \rightharpoonup \Ev{}}
    \quad\textbf{Consistent Transitivity}
    \begin{align*}
      \evpair{\?}{\?} \trans{} \evpair{\?}{\?} &= \evpair{\?}{\?} \\
      \evpair{\cT}{\cT} \trans{} \evpair{\?}{\?} &= \evpair{\cT}{\cT}
      \qquad \text{where}\;\cT \in \set{\Int,\Bool} \\
      \evpair{\?}{\?} \trans{} \evpair{\cT}{\cT} &= \evpair{\cT}{\cT}
      \qquad \text{where}\;\cT \in \set{\Int,\Bool} \\
      \evpair{\cT_{11} -> \cT_{12}}{\cT_{21} -> \cT_{22}} \trans{} 
      \evpair{\?}{\?} 
      &= 
      \evpair{\cT_{11} -> \cT_{12}}{\cT_{21} -> \cT_{22}} \trans{} \
      \evpair{\? -> \?}{\? -> \?} \\
      \evpair{\?}{\?} \trans{} 
      \evpair{\cT_{11} -> \cT_{12}}{\cT_{21} -> \cT_{22}} &=
      \evpair{\? -> \?}{\? -> \?} \trans{} \
      \evpair{\cT_{11} -> \cT_{12}}{\cT_{21} -> \cT_{22}} \\
      \evpair{\?}{\?} \trans{}
      \evpair{\rec{\overline{l_i:\cT_i},*_1}}{\rec{\overline{l_j:\cT_j},*_2}} &=
      \evpair{\rec{\?}}{\rec{\?}} \trans{}
      \evpair{\rec{\overline{l_i:\cT_i},*_1}}{\rec{\overline{l_j:\cT_j},*_2}} \\
      \evpair{\rec{\overline{l_i:\cT_i},*_1}}{\rec{\overline{l_j:\cT_j},*_2}}
      \trans{} \evpair{\?}{\?} &=
      \evpair{\rec{\overline{l_i:\cT_i},*_1}}{\rec{\overline{l_j:\cT_j},*_2}}
      \trans{} \evpair{\rec{\?}}{\rec{\?}}
    \end{align*}
    \begin{gather*}
      \evpair{\cT}{\cT} \trans{} \evpair{\cT}{\cT} = \evpair{\cT}{\cT}
      \qquad \text{where}\;\cT \in \set{\Int,\Bool} \\
      \evpair{\cT_{11} -> \cT_{12}}{\cT_{21} -> \cT_{22}} \trans{} 
      \evpair{\cT_{31} -> \cT_{32}}{\cT_{41} -> \cT_{42}} 
      = \evpair{\cT_{51} -> \cT_{52}}{\cT_{61} -> \cT_{62}} \\
      \qquad\text{where}\;
      \begin{block}
        \evpair{\cT_{41}}{\cT_{31}} \trans{} \evpair{\cT_{21}}{\cT_{11}} 
        = \evpair{\cT_{61}}{\cT_{51}}, \quad
        \evpair{\cT_{12}}{\cT_{22}} \trans{} \evpair{\cT_{32}}{\cT_{42}} 
        = \evpair{\cT_{52}}{\cT_{62}}
      \end{block}
    \end{gather*}
    \begin{gather*}
      \evpair{\rec{\overline{l_i:\cT_{i1}},*_1}}%
      {\rec{\overline{l_i:\cT_{i2}},*_2}}
      \trans{}
      \evpair{\rec{\overline{l_i:\cT_{i3}},*_3}}%
      {\rec{\overline{l_i:\cT_{i4}},*_4}}
      = 
      \evpair{\rec{\overline{l_i:\cT_{i5}},*_5}}%
      {\rec{\overline{l_i:\cT_{i6}},*_6}}
      \\
      \text{where}\;
      \begin{array}[t]{c}
        \overline{\evpair{\cT_{i1}}{\cT_{i2}} \trans{} \evpair{\cT_{i3}}{\cT_{i4}}
          = \evpair{\cT_{i5}}{\cT_{i6}}},
      \end{array}
      \quad
      \begin{array}[c]{c|c|c}
        \braket{*_1,*_2} &\braket{*_3,*_4}& = \braket{*_5,*_6} \\
        \hline
        \braket{\emptyset,\emptyset} & \braket{*_3,*_4} &
        \braket{\emptyset,\emptyset} \\
        \braket{\?,\?} &\braket{\?,\?}&\braket{\?,\?} \\
        \multicolumn{2}{c|}{\text{else}} & \braket{\?,\emptyset} \\
        \hline
      \end{array}
    \end{gather*}
    %
    \begin{multline*}
      %
      \evpair{\rec{\overline{l_i:\cT_{i1}}, \overline{l_j:\cT_{j1}}^{+},*_1}}%
      {\rec{\overline{l_i:\cT_{i2}},\overline{l_j:\cT_{j2}}^{+},*_2}}
      \trans{}
      \evpair{\rec{\overline{l_i:\cT_{i3}},\?}}%
      {\rec{\overline{l_i:\cT_{i4}},*_4}} = \\
      \evpair{\rec{\overline{l_i:\cT_{i1}}, \overline{l_j:\cT_{j1}}^{+},*_1}}%
      {\rec{\overline{l_i:\cT_{i2}},\overline{l_j:\cT_{j2}}^{+},*_2}}
      \trans{}
      \evpair{\rec{\overline{l_i:\cT_{i3}},\overline{l_j:\?}^{+},\?}}%
      {\rec{\overline{l_i:\cT_{i4}},*_4}}
    \end{multline*}
    %
    %
    \begin{multline*}
      %
      \evpair{\rec{\overline{l_i:\cT_{i1}},\?}}%
      {\rec{\overline{l_i:\cT_{i2}},\?}}
      \trans{}
      \evpair{\rec{\overline{l_i:\cT_{i3}}, \overline{l_k:\cT_{k3}}^{+},*_3}}%
      {\rec{\overline{l_i:\cT_{i4}},\overline{l_k:\cT_{k4}}^{+},*_4}}
      =  \\
      \evpair{\rec{\overline{l_i:\cT_{i1}},\overline{l_k:\?}^{+},\?}}%
      {\rec{\overline{l_i:\cT_{i2}},\overline{l_k:\?}^{+},\?}}
      \trans{}
      \evpair{\rec{\overline{l_i:\cT_{i3}}, \overline{l_k:\cT_{k3}}^{+},*_3}}%
      {\rec{\overline{l_i:\cT_{i4}},\overline{l_k:\cT_{k4}}^{+}*_4}}
    \end{multline*}
    %
    %
    \begin{multline*}
      %
      \evpair{\rec{\overline{l_i:\cT_{i1}},\overline{l_j:\cT_{j1}}^{+},\?}}%
      {\rec{\overline{l_i:\cT_{i2}},\overline{l_j:\cT_{j2}}^{+},\?}}
      \trans{}
      \evpair{\rec{\overline{l_i:\cT_{i3}},\overline{l_k:\cT_{k3}}^{+},\?}}%
      {\rec{\overline{l_i:\cT_{i4}},\overline{l_k:\cT_{k4}}^{+},*_4}}
      = \\
      \evpair{\rec{\overline{l_i:\cT_{i1}},\overline{l_j:\cT_{j1}}^{+},
          \overline{l_k:\?}^{+},\?}}%
      {\rec{\overline{l_i:\cT_{i2}},\overline{l_j:\cT_{j2}}^{+},
          \overline{l_k:\?}^{+},\?}}
      \trans{}
      \evpair{\rec{\overline{l_i:\cT_{i3}},\overline{l_k:\cT_{k3}},
          \overline{l_j:\?}^{+},\?}}%
      {\rec{\overline{l_i:\cT_{i4}},\overline{l_k:\cT_{k4}},*_4}}
    \end{multline*}
  \end{small}%
  \caption{Consistent Transitivity: Part 1}
  \label{fig:gr_ctrans1}
\end{figure}

\begin{figure}[h]
  \centering
  \begin{small}
    \flushleft\boxed{\trans{} : \Ev{} \times \Ev{} \rightharpoonup \Ev{}}
    \quad\textbf{Consistent Transitivity (cont'd.)}
    %
    \begin{multline*}
      %
      \evpair{\rec{\overline{l_i:\cT_{i1}},\overline{l_j:\cT_{j1}}^{\oplus_j},%
          \overline{l_k:\cT_{k}}^{\oplus_k},*_1}}%
      {\rec{\overline{l_i:\cT_{i2}},\overline{l_j:\cT_{j2}}^{\oplus_j}}}
      \trans{} 
      \evpair{\rec{\overline{l_i:\cT_{i3}},\overline{l_j:\cT_{j3}}^{\oplus_j}}}%
      {\rec{\overline{l_i:\cT_{i4}},*_4}} = \\
      \hspace{6cm}
      \evpair{\rec{\overline{l_i:\cT_{i5}},\overline{l_j:\cT_{j5}}^{\oplus_j},%
          \overline{l_k:\cT_{k}}^{\oplus_k},*_1}}%
      {\rec{\overline{l_i:\cT_{i6}},*_4}}
    \end{multline*}
    \begin{multline*}
      %
      \evpair{\rec{\overline{l_i:\cT_{i1}},\overline{l_j:\cT_{j1}}^{\oplus_j},%
          \overline{l_q:\cT_{q1}},\overline{l_k:\cT_k}^{\oplus_k},*_1}}%
      {\rec{\overline{l_i:\cT_{i2}},\overline{l_j:\cT_{j2}}^{^{\oplus_j}},%
          \overline{l_q:\cT_{q2}}}}
      \trans{}
      \evpair{\rec{\overline{l_i:\cT_{i3}},\overline{l_j:\cT_{j3}}^{\oplus_j},\?}}%
      {\rec{\overline{l_i:\cT_{i4}},*_4}}
      =  \\
      \evpair{\rec{\overline{l_i:\cT_{i5}},\overline{l_j:\cT_{j5}}^{\oplus_j},%
          \overline{l_q:\cT_{q1}},\overline{l_k:\cT_k}^{\oplus_k},*_1}}%
      {\rec{\overline{l_i:\cT_{i6}},*_4}}
    \end{multline*}
    \begin{multline*}
      %
      \evpair{\rec{\overline{l_i:\cT_{i1}},\overline{l_m:\cT_{m1}},%
          \overline{l_j:\cT_{j1}}^{\oplus_j},\overline{l_n:\cT_{n1}}^{\oplus_n},%
          \overline{l_k:\cT_k}^{\oplus_k},*_1}}%
      {\rec{\overline{l_i:\cT_{i2}},\overline{l_j:\cT_{j2}}^{\oplus_j},\?}}
      \trans{} \\
      \evpair{\rec{
          \overline{l_i:\cT_{i3}},\overline{l_m:\cT_{m3}},\overline{l_p:\cT_{p3}},
          \overline{l_j:\cT_{j3}}^{\oplus_j},\overline{l_n:\cT_{n3}}^{\oplus_n},
          \overline{l_r:\cT_{r3}}^{\oplus_r},*_3}}%
      {\rec{\overline{l_i:\cT_{i4}},\overline{l_m:\cT_{m4}},
          \overline{l_p:\cT_{p4}},*_4}} = \\
      \evpair{\rec{
          \overline{l_i:\cT_{i5}},\overline{l_m:\cT_{m5}},\overline{l_p:\cT_{p5}},
          \overline{l_j:\cT_{j5}}^{\oplus_j},\overline{l_n:\cT_{n5}}^{\oplus_n},
          \overline{l_r:\cT_{r5}}^{\oplus_r},
          \overline{l_k:\cT_k}^{\oplus_k},*_1}}%
      {\rec{\overline{l_i:\cT_{i6}},\overline{l_m:\cT_{m6}},
          \overline{l_p:\cT_{p6}},*_4}}
    \end{multline*}
    \begin{gather*}
      *_1 = \? \;\text{if}\; \set{\overline{l_p},\overline{l_r}^{\oplus_r}}
      \neq \emptyset 
    \end{gather*}

    \begin{multline*}
      %
      \evpair{\rec{\overline{l_i:\cT_{i1}},\overline{l_m:\cT_{m1}},
          \overline{l_j:\cT_{j1}}^{\oplus_j},\overline{l_n:\cT_{n1}}^{\oplus_n},
          \overline{l_q:\cT_{q1}}^{+},\overline{l_k:\cT_k}^{\oplus_k},*_1}}%
      {\rec{\overline{l_i:\cT_{i2}},\overline{l_j:\cT_{j2}}^{\oplus_j},
          \overline{l_q:\cT_{q2}}^{+},\?}}
      \trans{} \\
      \evpair{\rec{
          \overline{l_i:\cT_{i3}},\overline{l_m:\cT_{m3}},\overline{l_p:\cT_{p3}},
          \overline{l_j:\cT_{j3}}^{\oplus_j},\overline{l_n:\cT_{n3}}^{\oplus_n},
          \overline{l_r:\cT_{r3}}^{\oplus_r},\?}}%
      {\rec{\overline{l_i:\cT_{i4}},\overline{l_m:\cT_{m4}},
          \overline{l_p:\cT_{p4}},*_4}}
    \end{multline*}
    \vspace{-1.5em}
    \begin{gather*}
      =
    \end{gather*}
    \vspace{-1.5em}
    \begin{multline*}
      \evpair{\rec{\overline{l_i:\cT_{i1}},\overline{l_m:\cT_{m1}},
          \overline{l_j:\cT_{j1}}^{\oplus_j},\overline{l_n:\cT_{n1}}^{\oplus_n},
          \overline{l_q:\cT_{q1}}^{+},\overline{l_k:\cT_k}^{\oplus_k},*_1}}%
      {\rec{\overline{l_i:\cT_{i2}},\overline{l_j:\cT_{j2}}^{\oplus_j},
          \overline{l_q:\cT_{q2}}^{+},\?}}
      \trans{} \\
      \evpair{\rec{
          \overline{l_i:\cT_{i3}},\overline{l_m:\cT_{m3}},\overline{l_p:\cT_{p3}},
          \overline{l_j:\cT_{j3}}^{\oplus_j},\overline{l_n:\cT_{n3}}^{\oplus_n},
          \overline{l_r:\cT_{r3}}^{\oplus_r},\overline{l_q:\?}^{+},\?}}%
      {\rec{\overline{l_i:\cT_{i4}},\overline{l_m:\cT_{m4}},
          \overline{l_p:\cT_{p4}},*_4}}
    \end{multline*}

    \begin{align*}
      \text{where}\;
      &\overline{\evpair{\cT_{i1}}{\cT_{i2}} \trans{}
        \evpair{\cT_{i3}}{\cT_{i4}} 
        = \evpair{\cT_{i5}}{\cT_{i6}}},
      &\overline{\evpair{\cT_{j1}}{\cT_{j2}} \trans{} \evpair{\cT_{j3}}{\cT_{j3}}
        = \evpair{\cT_{j5}}{\cT_{j6}}} \\
      &\overline{\Isub |[\cT_{m1} \csub \? |] \trans{}
        \evpair{\cT_{m3}}{\cT_{m4}} 
        = \evpair{\cT_{m5}}{\cT_{m6}}},
      &\overline{\Isub |[\cT_{n1} \csub \? |] \trans{}
        \evpair{\cT_{n3}}{\cT_{n3}} 
        = \evpair{\cT_{n5}}{\cT_{n6}}} \\ 
      &\overline{\evpair{\?}{\?} \trans{} \evpair{\cT_{p3}}{\cT_{p4}} 
        = \evpair{\cT_{p5}}{\cT_{p6}}},
      &\overline{\evpair{\?}{\?} \trans{} \evpair{\cT_{r3}}{\cT_{r3}} 
        = \evpair{\cT_{r5}}{\cT_{r6}}} \\ 
      &\braket{\oplus_{j+n+r},\oplus_k} \in
      \set{\braket{\emptyset,+},\braket{+,\emptyset},\braket{+,+}},
      &\oplus_{j+n+r} =
      \begin{cases}
        \oplus_j & \oplus_n,\oplus_r\;\text{do not appear} \\
        \emptyset &
        \braket{\oplus_j,\oplus_n,\oplus_r} =
        \braket{\emptyset,\emptyset,\emptyset} \\
        + & \text{otherwise}
      \end{cases}
    \end{align*}
    \begin{gather*}
      \evpair{\cT_1}{\cT_2} \trans{} \evpair{\cT_3}{\cT_4}\;
      \text{undefined otherwise}
    \end{gather*}
    \begin{mathpar}
    \end{mathpar}
  \end{small}%
  \caption{Consistent Transitivity: Part 2}
  \label{fig:gr_ctrans2}
\end{figure}

\clearpage

\section{Extrinsically- and Intrinsically-typed Terms are Equivalent }
\label{sec:extr-intr-correspondence}

\begin{figure}
  \centering
  \begin{displaymath}
    x^{\cT} \in \VarT{\cT}
  \end{displaymath}
  \begin{small}
    \begin{mathpar}
      \inference[(I$\cT$n)]{}{n \in \TermT{\Int}}
      \and
      \inference[(I$\cT$b)]{}{b \in \TermT{\Bool}}
      \and
      \inference[(I$\cT$x)]{}{x^{\cT} \in \TermT{\cT}}
      \and
      \inference[(I$\cT$+)]{
        t^{\cT_1} \in \TermT{\cT_1} & \ev_1 |- \cT_1 \csub \Int \\
        t^{\cT_2} \in \TermT{\cT_2} & \ev_2 |- \cT_2 \csub \Int \\
      }{
        \evcast{\ev_1}{t^{\cT_1}} + 
        \evcast{\ev_2}{t^{\cT_2}} \in \TermT{\Int}
      }
      \and
      \inference[(I$\cT\lambda$)]{t^{\cT_2} \in \TermT{\cT_2}
      }{
        \lambda x^{\cT_1}.t^{\cT_2} \in \TermT{\cT_1 -> \cT_2}
      }
      \;
      \inference[(I$\cT$rec)]{
        \overline{t^{\cT_i} \in \TermT{\cT_i}}
      }{
        \rec{\overline{l_i=t^{\cT_i}}} \in 
        \TermT{\rec{\overline{l_i:\cT_i}}}
      }
      \and
      \inference[(I$\cT{::}$)]{
        t^{\cT_1} \in \TermT{\cT_1} &
        \ev |- \cT_1 \csub \cT_2 \\
      }{
        \evcast{\ev}{t^{\cT_1}} :: \cT_2 \in \TermT{\cT_2}
      } 
      \and
      \inference[(I$\cT$app)]{
        t^{\cT_1} \in \TermT{\cT_1} &  \ev_1 |- \cT_1 \csub \cT_{11} -> \cT_{12} \\
        t^{\cT_2} \in \TermT{\cT_2} &
        \ev_2 |- \cT_2 \csub \cT_{11}
      }{
        (\evcast{\ev_1}{t^{\cT_1}})%
        \iapp{\cT_{11} -> \cT_{12}}%
        (\evcast{\ev_2}{t^{\cT_2}})
        \in \TermT{\cT_{12}}
      }  
      \and
      \inference[(I$\cT$proj)]{
        \vspace{0.2em}
        t^{\cT} \in \TermT{\cT} &
        \ev |- \cT \csub \rec{l : \cT_1,\?}
      }{
        \evcast{\ev}{t^{\cT}.l^{\cT_1,}} \in \TermT{\cT_1} 
      }
      \and
      \inference[(I$\cT$if)]{
        t^{\cT_1} \in \TermT{\cT_1} & 
        \ev_1 |- \cT_1 \csub \Bool \\
        t^{\cT_2} \in \TermT{\cT_2} & 
        \ev_2 |- \cT_2 \csub  \cT_2 \mathbin{\scalebox{0.8}{$\csubjoin$}} \cT_3\\
        t^{\cT_3} \in \TermT{\cT_3} &
        \ev_3 |- \cT_3 \csub \cT_2 \mathbin{\scalebox{0.8}{$\csubjoin$}} \cT_3
      }{
        \<if>   \evcast{\ev_1}{t^{\cT_1}} 
        \<then> \evcast{\ev_2}{t^{\cT_2}}
        \<else> \evcast{\ev_3}{t^{\cT_3}}
        \in \TermT{\cT_2 \mathbin{\scalebox{0.6}{$\csubjoin$}} \cT_3}
      }  
    \end{mathpar}
  \end{small}%
  \caption{Gradual Intrinsic Terms}
  \label{fig:intrinsic-terms}
\end{figure}

Sec.~\ref{sec:gtflcsub} of the paper presents the dynamic semantics of
$\GTFLcsub$ in terms of a \emph{runtime language} RL, whose type system
classifies runtime terms using source language typing judgments.  This
presentation appears to differ substantially from the \emph{intrinsic terms} of
\citet{garciaAl:popl2016} (Fig.~\ref{fig:intrinsic-terms}).  This section shows
that the differences are merely cosmetic: well-typed RL terms are equivalent to
intrinsic terms.  In particular, every closed intrinsically typed term
corresponds to a closed extrinsically-typed term with the same type, and
reduction of such terms preserves correspondence.

Though formally interchangeable, we believe that the extrinsic presentation of
$\GTFLcsub$ clarifies some aspects of the language's structure and semantics.
First, the type structure of the runtime language precisely formalizes the
sense in which running a gradually-typed programming language corresponds to
running a source typing derivation, which is reflected in its typing judgment.
Second, RL terms are comprised only of constructs that are needed at runtime.
Any static type information that is needed only for metatheory is relegated to
the typing judgment.  In short, RL strictly segregates computationally
irrelevant source type information (found only in the RL type) from essential
runtime type information (found only in the RL term).  This separation can help
the reader disambiguate what is essential only for typing and what is essential
only for runtime execution.  
Third, this separation highlights the way
that runtime execution uncovers more precise type information (reflected in
evidence) without affecting the static type reasoning (reflected in the typing
judgments).  This deduced type information is not present in corresponding
source terms, which appear as part of the runtime typing judgment.
Finally, the extrinsic presentation more closely connects with
older cast-based presentations of gradual typing, where a distinct and
seemingly unrelated internal language is used define evaluation.  Here, RL
plays that role, but its connection to the source program and its gradual type
discipline is precise and firm.

To clarify the correspondence between extrinsic RL terms and intrinsic terms,
we define a correspondence relation between the two kinds of terms, and prove
that reduction preserves and reflects correspondence.  This correspondence can
be used to prove that theorems about intrinsic terms due to
\citet{garciaAl:popl2016} can be transported to well-typed RL terms.

\begin{figure}
  \centering

  \begin{gather*}
    \begin{block}
      \Omega \in \oblset{CgrCtxt} =
      \set{\Omega \in \VarT{*} \rightleftharpoons \Var \times \GType |
        \braket{x^{*}, x, \cT} \in \Omega \implies x^{*} \in \VarT{\cT}
      }
      \\[0.5ex]
      \Omega ::= \overline{x^\cT \sim x : \cT}
    \end{block}
  \end{gather*}
  \flushleft\boxed{t^s \sim \rr \derivof \Omega |- t : \cT}
  \quad\textbf{Correspondence}
  \begin{mathpar}
    \infer{x^{\cT} \sim \rx \derivof \Omega |- x : \cT
    }{
      x^{\cT} \sim x  : \cT \in \Omega
    }
    \and
    \infer{n \sim \rn \derivof \Omega |- n : \Int }{}
    \and
    \infer{b \sim \rb \derivof \Omega |- b : \Bool }{}
    \and
    \infer{
      (\evcast{\ev_1}{t^{\cT_1}} + 
      \evcast{\ev_2}{t^{\cT_2}})
      \sim 
      (\evcast{\ev_1}{\rr_1} + 
      \evcast{\ev_2}{\rr_2})
      \derivof
      \Omega |- t_1\;t_2 : \Int
    }{
      \overline{t^{\cT_i} \sim \rr_i \derivof \Omega |- t_i : \cT_i}
    }
    \and
    \infer{
      (\evcast{\ev_1}{t^{\cT_1}}%
      \iapp{\cT'_1 -> \cT'_2}%
      \evcast{\ev_2}{t^{\cT_2}})
      \sim
      (\cast{\ev_1}{\rr_1}\;\cast{\ev_2}{\rr_2}) \derivof
      \Omega |- t_1\;t_2 : \Gbox{\cT'_2} 
    }{
      \overline{t^{\cT_i} \sim \rr_i \derivof \Omega |- t_i : \cT_i}
    }
    \and
    \infer{
      \lambda x^{\cT_1}.t^{\cT_2} \sim
      \lambda x.\rr \derivof \Omega |-
      \lambda x : \cT_1. t : \Gbox{\cT_2} 
    }{
      t^{\cT_2} \sim
      \rr\derivof \Omega,\,x \sim x^{\cT_1} : \cT_1 |- t : \Gbox{\cT_2} 
    }
    \and
    \infer{
      \rec{\overline{l_i=t^{\cT_i}}}
      \sim
      \rec{\overline{l_i=\rr_i}}
      \derivof
      \Omega |- \rec{\overline{l_i=t_i}} : \rec{\overline{l_i:\cT_i}}
    }{
      \overline{t^{\cT_i} \sim \rr_i \derivof \Omega |- t_i : \cT_i}
    }
    \and
    \infer{
      \evcast{\ev}{t^{\cT_1}} :: \cT_2
      \sim \evcast{\ev}{\rr}
      \derivof 
      \Omega |- t : \cT_2
    }{
      t^{\cT_1}
      \sim \rr
      \derivof 
      \Omega |- t : \cT_1
    }
    \and
    \infer{
      \evcast{\ev}{t^{\cT}.l^{\cT_1}}
      \sim
      \evcast{\ev}{\rr.l} \derivof \Omega |- t.l : \cT_1
    }{
      t^{\cT} \sim \rr \derivof \Omega |- t : \cT
    }
    \and
    \infer{
      \begin{block}
        \<if>   \evcast{\ev_1}{t^{\cT_1}} 
        \<then> \evcast{\ev_2}{t^{\cT_2}}
        \<else> \evcast{\ev_3}{t^{\cT_3}}
        \sim \\
        \<if>   \evcast{\ev_1}{\rr_1} 
        \<then> \evcast{\ev_2}{\rr_2}
        \<else> \evcast{\ev_3}{\rr_3}
      \end{block}
      \derivof \Omega |- 
      \<if>   t_1
      \<then> t_2
      \<else> t_3
      : {\cT_2 \csubjoin \cT_3}
    }{
      \overline{t^{\cT_i} \sim \rr_i \derivof \Omega |- t_i : \cT_i}
    }
  \end{mathpar}
  
  \caption{Intrinsic-Extrinsic Correspondence}
  \label{fig:correspondence}
\end{figure}

Fig.~\ref{fig:correspondence} defines a \emph{correspondence} relationship
between intrinsically-typed terms and their extrinsic counterparts.  The
judgment $t^{\cT} \sim \rr \derivof \Omega |- t : \cT$ says that given the
typed variable correspondences in $\Omega$, the intrinsic term $t^{\cT}$
corresponds to the extrinsic term $\rr$ as viewed through the runtime typing
judgment $\rr \derivof \Gamma |- t : \cT$, where
$\Gamma = \set{x : \cT | x^{\cT} \sim x : \cT \in \Omega }$, a runtime typing
context from the variable correspondence.
We use the notation $A \rightleftharpoons B \times C$ to denote
subsets of $A \times B \times C$ that exhibit a bijective relationship between
$A$'s and $B$'s when $C$ is omitted.
Formally,
\newcommand{\existsmaybe}{\ensuremath{\exists^{\leq 1}}}

\begin{gather*}
  A \rightleftharpoons B \times C =
  \set{S \in A \times B \times C |
    \left(\forall a \in A.\,\existsmaybe s \in S. \pi_1(s) = a \right)
    \land \left(\forall b \in B.\,\existsmaybe s \in S. \pi_2(s) = b\right)
  } \\
  \text{where}\;
  \existsmaybe x.P(x) \equiv
  (\exists x.P(x)) \implies \exists x.P(x) \land \forall y.P(y) \implies x=y.
\end{gather*}
In other words, $\Omega$ uniquely relates
intrinsic variables to extrinsic variables at the given type.
The well-formedness of the intrinsic term $t^{\cT}$ in the judgment
implicitly imposes invariants on the evidence objects referenced in the
rules.

We formally establish the intended properties of the judgment as follows.
\begin{proposition}
  The extrinsic part of a correspondence is well-typed: let $\Omega$ be some
  variable correspondence, and
  $\Gamma = \set{x : \cT | x^{\cT} \sim x : \cT \in \Omega }$ be its induced
  extrinsic typing context.  Then
  ${t^{\cT} \sim \rr \derivof \Omega |- t : \cT}$ implies
  $\rr \derivof \Gamma |- t : \cT$
\end{proposition}
\begin{proof}
  By induction over $t^{\cT} \sim \rr \derivof \Omega |- t : \cT$.
\end{proof}

Furthermore, intrinsically-typed terms correspond uniquely (up to choice of
bound variable names) to well-typed RL terms.  In the opposite direction,
extrinsic typing \emph{derivations} correspond uniquely to intrinsically-typed
terms.

\begin{proposition}
  \mbox{}
  \begin{enumerate}
  \item 
    If ${t^{\cT} \in \TermT{\cT}}$ then there is a unique $\rr$, up to
    $\alpha$-equivalence, such that \linebreak
    ${t^{\cT} \sim \rr \derivof \Omega |- t : \cT}$
    i.e.,
    \begin{equation*}
      \forall \cT \in \GType,t^{\cT} \in\TermT{\cT}. \exists \rr,\Omega,t.\,
      t^{\cT} \sim \rr \derivof \Omega |- t : \cT \land
      \forall \rr'.t^{\cT} \sim \rr' \derivof \Omega |- t : \cT \implies
      \rr \sim_{\alpha} \rr';
    \end{equation*}
  \item
    If $\rr \derivof \Gamma |- t : \cT$, then there is a unique $t^{\cT}$, up
    to $\alpha$-equivalence, such that
    ${t^{\cT} \sim \rr \derivof \Omega |- t : \cT.}$
  \end{enumerate}
\end{proposition}

\begin{proof}
  \mbox{}
  \begin{enumerate}
  \item By induction on the formation of $t^{\cT} \in \Term$.
  \item By induction on the derivation of $\rr \derivof \Gamma |- t : \cT$.
  \end{enumerate}
\end{proof}

Also, corresponding values are observationally equivalent.
\begin{proposition}[Canonical Forms]
  \mbox{}
  \begin{enumerate}
  \item If $n_1 \sim \rn_2 \derivof \emptyset |- n_3 : \cT$ then
    $n_1 = \rn_2 = n_3$;
  \item If $b_1 \sim \rb_2 \derivof \emptyset |- b_3 : \cT$ then
    $b_1 = \rb_2 = b_3$.
  \item If
    $\lambda x^{\cT_1}.t^{\cT'2}_{11} \sim \rr \derivof \emptyset |- t : \cT$ then
    $\rr = \lambda x.\rr_{11}$ and $t = \lambda x : \cT_1.t_{11}$.
  \item If
    $t^{\cT} \sim \lambda x.\rr_{11} \derivof \emptyset |- t : \cT$ then
    $t^{\cT} = \lambda x^{\cT_1}.t^{\cT'2}_{11}$ and $t = \lambda x : \cT_1.t_{11}$.
  \end{enumerate}
\end{proposition}

Finally, closed intrinsic terms and their corresponding well-typed RL terms
have the same reduction behaviour. 
\begin{proposition}[Bisimulation]
  If ${t_1^{\cT} \sim \rr_1 \derivof \emptyset |- t_1 : \cT}$ then:
  \begin{enumerate}
  \item If $t_1^{\cT} --> t^{\cT}_2$ then $\rr_1 --> \rr_2$ and
    ${t_2^{\cT} \sim \rr_2 \derivof \emptyset |- t_2 : \cT}$ for some
    $\rr_2,t_2$
  \item If $\rr_1 --> \rr_2$ then  $t_1^{\cT} --> t^{\cT}_2$ and
    ${t_2^{\cT} \sim \rr_2 \derivof \emptyset |- t_2 : \cT}$ for some
    $t^{\cT}_2,t_2$.
  \end{enumerate}
\end{proposition}
\begin{proof}
  By induction on derivations of
  ${t_1^{\cT} \sim \rr_1 \derivof \emptyset |- t_1 : \cT}$
\end{proof}
This property shows that the static type information is irrelevant to
computational results.  By the structure of the inductive rules that define it,
the correspondence relation is compatible with well-typed term formation.  It
is \emph{not} a congruence because the relation is between distinct languages,
so neither reflexive nor transitive.

\newpage
\section{Associativity implies Space Efficiency}
\label{apdx:assoc-implies-space-efficiency}
If evidence composition is associative, we can construct a space
efficient semantics that simulates the original semantics.  This
section aims to prove associativity suffices for semantic equivalence.  To simplify the proof effort, we
first develop an equivalent Structural Operational Semantics to the
Reduction Semantics presented on the paper.  This makes it easier
for us to appeal to the induction hypothesis. 

\subsection{Non-Space Efficient Structural Operational Semantics}
We first define our Structural Operational Semantics for the
$\GTFLcsub$ runtime language.  We will use colors to distinguish our
semantics: Our Non-Space Efficient Structural Operational Semantics
equivalent to the standard semantics previously presented in the
paper will use the style $\pla{e --> \dots}$.

\flushleft\textbf{Syntax}
\begin{displaymath}
  \begin{array}{rcll}
    \pla{e} \in \pla{\RTerm} & ::= & \pla{n} | \pla{\cast{\ev}{e} + \cast{\ev}{e}} | \pla{x} | \pla{\lambda x.e} |  \pla{e\;e} | 
    \pla{\evcast{\ev}{e}} | \pla{\rec{\overline{l = e}}} | \pla{\cast{\ev}{e}.l} | \pla{ \<if> \cast{\ev}{e} \<then> \cast{\ev}{e} \<else> \cast{\ev}{e} } & \text{(runtime
      terms)} \\
    \pla{F} \in \pla{\EvFrame} & ::= & \pla{[]} | \pla{[]\;\cast{\ev}{e}} | \pla{\cast{\ev}{u}\;[]} | \pla{[] + \cast{\ev}{e}} | \pla{\cast{\ev}{u} + []} | \pla{[].l} \\
    & | & \pla{\<if> [] \<then> \cast{\ev}{e} \<else> \cast{\ev}{e}}
    & \text{(runtime
      frames)}
  \end{array}
\end{displaymath}

\flushleft\textbf{Operational Semantics}    
\flushleft\boxed{\pla{e --> e},\; \pla{e --> \error}}

\begin{mathpar}
  \inference[Ev]{\pla{e} \leadsto \pla{e'}}
  {\pla{e --> e'}}

  \inference[Ex]{\pla{e} \leadsto \pla{\error}}{
    \pla{e --> \error}}

  \inference[Fv]{\pla{{e} --> {e'}}}{
    \pla{F[\evcast{\ev}{e}] --> F[\evcast{\ev}{e'}]}}

  \inference[Fx]{\pla{{e} --> \error}}{
    \pla{F[\evcast{\ev}{e}] --> \error}}

  \inference[Rec-v]{\pla{{e} --> {e'}}}{
    \pla{\rec{\overline{l = \rv}, l = e, \overline{l = \rr}} -->
      \rec{\overline{l = \rv}, l = e', \overline{l = \rr}}}}

  \inference[Rec-x]{\pla{{e} --> \error}}{
    \pla{\rec{\overline{l = \rv}, l = e, \overline{l = \rr}} -->
      \error}}
  
  \inference[Cv]{}{
    \pla{F[\evcast{\ev_1}{\evcast{\ev_2}{u}}] --> F[\evcast{\ev_2
        \trans{} \ev_1}{u}]}}

  \inference[Cx]{ \ev_2 \trans{} \ev_1 \text{ undefined}}{
    \pla{F[\evcast{\ev_1}{\evcast{\ev_2}{u}}] --> \error}}

\end{mathpar}

The notions of reduction ($\leadsto$) are exactly the same as in $\GTFLcsub$.

We can prove that this semantics is equivalent to the one described
in Fig.~\ref{fig:gtflr-dynamics} by simple structural induction
over the appropriate derivation on the assumption (the reduction semantics for $=>$
and the structural operational semantics for $\Leftarrow$).

\begin{lemma}[Growing contexts]
  For any $E$ and $E'$ evaluation contexts in $\GTFLcsub$, then also
  $E [E']$, where $[ ]$ is the standard plugging operation, is also
  an evaluation context.  (We do need to prove this lemma separately
  as we follow a construction order for contexts that is the opposite
  from the one usual in Felleisen style reduction semantices)
\end{lemma}
\begin{proof}
  By structural induction on the definition of $E$.
  \begin{itemize}
  \item Case $E = []$ is trivial, as $E[E']=E'$.
  \item Case $E = E''[F[ \evcast{\ev}{[]} ]]$. By induction, then
    $E''[E']$ is also an evaluation context.  We can use then the
    $(E''[E'])[F [ \evcast{\ev}{[]} ]]$ constructor to build our
    evaluation context for $E [ E' ]$.
  \item Case
    $E' = E''[ \rec{\overline{l = \rv}, l = [], \overline{l = \rr}} ]$.
    Analogous to the previous case.
  \end{itemize}
\end{proof}

\begin{lemma}[Contextual Reduction can grow contexts]
  \label{lem:grow-contexts-gtfl}
  We will need this lemma for proving Theorem \ref{thm:sem-eq-rs-sos-gtfl}. This lemma has two parts:

  \begin{itemize}
  \item Let $E[\rr] --> E[\rr']$. Then for any $E'$ such that
    $E' = E''[E]$ for some $E''$, then also $E'[\rr] --> E'[\rr']$.
  \item Let $E[\rr] --> \error$.  Then for any $E'$ such that
    $E' = E''[E]$ for some $E''$, then also $E'[\rr] --> \error$.
  \end{itemize}

\end{lemma}
\begin{proof}
  By cases on the contextual reduction rule used to derive either
  $E[\rr] --> E[\rr']$ or $E[\rr] --> \error$.  In each case, we can
  use the exact same rule template to derive either
  $E''[E[\rr]] --> E''[E[\rr']]$ or $E''[E[\rr]] --> \error$,
  respectively.
\end{proof}

\begin{theorem}[Semantic equivalence between Reduction Semantics and
  SOS]
  \label{thm:sem-eq-rs-sos-gtfl}
  $\rr --> \rr'$ in $\GTFLcsub$ if and only if
  $\pla{\rr --> \rr'}$ in the SOS semantics as well.  Also,
  $\rr --> \error$ in $\GTFLcsub$ if and only if
  $\pla{\rr --> \error}$ in the SOS semantics as well.
\end{theorem}

\begin{proof}
  We begin by proving the only if direction $(==>)$.
  
  By cases on  the definition of $-->$ we can pick the top-most rule that will be used
  in the SOS. For each rule template in the reduction semantics, we proceed by induction on
  the structure of $E$, adding rules on the derivation as needed,
  with the following protocol: $E = []$ adds no rule as you have a
  complete derivation by that point. $E = E[\rec{\overline{l=\rv}, l = [], \overline{l = \rr}}]$,
  adds an instance of  rule [Rec-v] if the derivation so far steps
  to an expression and rule [Rec-x] if it steps to an error.
  $E = E [ F [ \evcast{\ev}{[]} ]]$ 
  adds an instance of  rule [Fv] if the derivation so far steps
  to an expression and rule [Fx] if it steps to an error.
  
  For the top of the derivation, we chose rules as follows:
  \begin{itemize}
  \item Rule $\inference{\rr \leadsto \rr'}{E[\rr] --> E[\rr']}$.
    We apply rule
    [Ev] and we fill the stepping for evaluation contexts as described.
  \item Rule $\inference{\rr \leadsto \error}{E[\rr] --> \error}$.
    We apply rule
    [Ex] and we fill the stepping for evaluation contexts as described.
  \item Rule
    $\inference{}{E[F[\evcast{\ev_1}{\evcast{\ev_2}{u}}]] --> E[F[(\ev_2 \trans{} \ev1) u]]}$.
    We apply rule
    [Cv] and we fill the stepping for evaluation contexts as described.
  \item Rule
    $\inference{\ev_2 \trans{} \ev_1 \text{ undefined}}{E[F[\evcast{\ev_1}{\evcast{\ev_2}{u}}]] --> \error}$.
    We apply rule
    [Cx] and we fill the stepping for evaluation contexts as described.
  \end{itemize}

  We now proceed to prove the if case ($<==$). We proceed by structural induction on the definition of
  $\pla{-->}$:
  \begin{itemize}
  \item Rules [Ev], [Ex], [Cv] and [Cx] are trivial as we can build
    derivations for them in $-->$ using $E=[]$.
  \item Rule [Fv]. By induction, $e --> e'$.  By inversion on the
    definition of $-->$, then there exists $e_2, e'_2$ and $E$ such that
    $e = E[e_2]$ and $e' = E[e'_2]$. 
    We can then use $E''= ([] [ F [ \evcast{\ev}{[]}]]) [E]$ and Lemma
    \ref{lem:grow-contexts-gtfl} to prove that
    $F[ \evcast{\ev}{e} ] --> F[ \evcast{\ev}{e'} ]$.
  \item Rule [Fx]. Analogous to Rule [Fv].
  \item Rule [Rec-v]. Analogous to Rule [Fv], but with context
    $E'' = ([] [ \rec{\overline{l = \rv}, l = [], \overline{l = \rr}}]) [E]$.
  \item Rule [Rec-x]. Analogous to Rule [Rec-v].
  \end{itemize}
\end{proof}

\subsection{Space-Efficient Structural Operational Semantics}

We can now introduce the space-efficient language.  Since we want
the new language to be observationally equivalent, we must deal with
the fact that evidence composition ($\trans{}$) is a partial
function.  We work around this issue by introducing a notation for
\emph{``total'' evidence}, $\botev$.  A term $\botev$ can be a normal
evidence $\ev$ or $\bot$, the latter acting as a marker for an
output where an evidence composition was undefined.

The terms in the space efficient language are
almost the same as in the structural one, \emph{except} for the fact
that we replace evidence objects with total evidence objects.  We also choose to
color them differently: While a term $\pla{e}$ represents a term in
the SOS-version of the original semantics, a term $\pls{e}$
represents a term in the space-efficient version of the language.

\flushleft\textbf{Syntax}
\begin{displaymath}
  \begin{array}{rcll}
    \pls{\botev} & ::= & \pls{\bot} | \pls{\ev} & \text{(total
      evidence)} \\
    \pls{e} \in \pls{\RTerm} & ::= & \pls{n} | \pls{\cast{\botev}{e} + \cast{\botev}{e}} | \pls{x} | \pls{\lambda x.e} |  \pls{\evcast{\botev}{e}\;\evcast{\botev}{e}} | 
    \pls{\evcast{\botev}{e}} | \pls{\rec{\overline{l = e}}} | \pls{\cast{\botev}{e}.l} & \text{(runtime
      terms)} \\
    \pls{F'} \in \pls{\EvFrame} & ::= & \pls{[]\;\cast{\ev}{e}} | \pls{\cast{\ev}{u}\;[]} | \pls{[] + \cast{\ev}{e}} | \pls{\cast{\ev}{u} + []} |  \pls{[].l} \\
    & | & \pls{\<if> [] \<then> \cast{\ev}{e} \<else> \cast{\ev}{e}}& \text{(runtime frames)}
  \end{array}
\end{displaymath}

\flushleft\textbf{Definitions required to propagate evidence}
We substitute the partial evidence composition operation
$\trans{}$ with a total equivalent, denoted $\bottrans{}$ and defined as
follows:
\[\botev_1 \bottrans{} \botev_2 = \begin{cases}
    \ev_1 \trans{} \ev_2 & \text {if } \botev_1 = \ev_1
    \text{ and  }  \botev_2 = \ev_2 \text{ and composition is defined} \\
    \bot & \text{otherwise}
  \end{cases}\]

With these definitions we are ready to present the definition of our
space-efficient semantics:

\flushleft\textbf{Operational Semantics}

The operational semantics separates the cases that deal with nested
casts (originally appearing as nested uses of the $\pla{[]}$
evaluation frames, such as
$\pla{[] [\cast{\ev}{\dots [] [\cast{\ev}{e'}]}]}$ with
$\pla{e'} \neq \pla{\cast{\ev}{e''}}$.  We omit derivations like those
that would have happened through rules [Fv] and [Fx] and make explicit
that evidence composition must happen first and thus we will always
have at most one pending cast in the context surrounding a redex.

\newcommand{\notleadsto}[0]{\looparrowright}

It is also important to note that since our redexes now hold
\emph{total evidence} instead of plain evidence, the original notions
of reduction are insufficient for this language, as they only deal with
the ``not $\bot$'' cases.  We introduce an extra rule to deal with the
cases where a $\bot$ total evidence appears (Rule [SFx']).

With this extension we can present the full space-efficient
operational semantics:

\flushleft\boxed{\pls{e --> e},\; \pls{e --> \error}}

\begin{mathpar}
  \inference[SEv]{\pls{e} \leadsto \pls{e'}}
  {\pls{e --> e'}}

  \inference[SpEf-Ex]{\pls{e} \leadsto \pls{\error}}{
    \pls{e --> \error}}

  \inference[SRec-v]{\pls{{e} --> {e'}}}{
    \pls{\rec{\overline{l = \rv}, l = e, \overline{l = \rr}} -->
      \rec{\overline{l = \rv}, l = e', \overline{l = \rr}}}}

  \inference[SRec-x]{\pls{{e} --> \error}}{
    \pls{\rec{\overline{l = \rv}, l = e, \overline{l = \rr}} -->
      \error}}

  \inference[SCe]{}{
    \pls{F'[\evcast{\botev_1}{\evcast{\botev_2}{e}}] --> F'[\evcast{\botev_2
        \bottrans{} \botev_1}{e}]}}

  \inference[SCe']{}{
    \pls{\evcast{\botev_1}{\evcast{\botev_2}{e}} --> \evcast{\botev_2
        \bottrans{} \botev_1}{e}}}    
  
  \inference[SFx]{}{\pls{\evcast{\bot}{u} --> \error}}

  \inference[SFx']{}{
    \pls{F'[\evcast{\bot}{u}] --> \error}}
  
  \inference[SFe]{\pls{{e} --> {e'}} &&\pls{e} \neq
    \pls{\evcast{\botev'}{e''}}
  }{
    \pls{F'[\evcast{\botev}{e}] --> F'[\evcast{\botev}{e'}]}}

  \inference[SCx]{ \pls{{e} --> \error} && \pls{e} \neq
    \pls{\evcast{\botev'}{e'}}
  }{
    \pls{F'[\evcast{\botev}{e}] --> \error}}

  \inference[SStack]{\pls{e' --> e''}
    && \pls{e'} \neq \pls{\evcast{\botev'}{e'''}}
  }{\pls{\evcast{\botev_0}{e'}
      --> \evcast{\botev_0}{e''}}}

  \inference[SStack']{\pls{e'
      -->
      \error}
    && \pls{e'} \neq \pls{\evcast{\botev'}{e''}}
  }{\pls{\evcast{\botev_0}{e'}
      --> \error}}
\end{mathpar}

Rules [SEv] and [SpEf-Ex] are exactly as in the original
semantics. Rules [SCe] and [SCe'] perform the cast composition
eagerly, thus avoiding the stack growth present in the original
semantics. Once a $\pls{\evcast{\bot}{u}}$ value is reached through
cast composition, rules [SFx] and [SFx'] trigger an error just as in
the original semantics,  when composing evidence objects in an
innermost to outermost fashion. Rules [SFe] and [SCx] do grow the
stack for redexes that do not represent an ascription, and rules
[SStack] and [SStack'] do grow the stack when an ascription is
wrapping a different redex (not an ascription) that must be reduced
first. It is important to note that, with this approach, casts cannot
arbitrarily nest and produce a space leak.

\begin{lemma}[Outside-in inversion for Evaluation Contexts]
  \label{lem:brr-space-efficient-gc}
  For every evaluation context $G$ capable of producing valid
  syntactic terms once plugged, one of the following is true:
  \begin{enumerate}
  \item $G = []$.
  \item $G = [] [F[ [] ] ]$.
  \item Exists $G'$, such that
    $G = ([] [\rec{\overline{l = \rv}, l = [], \overline{l = \rr}} ]) [ G']$.
  \item Exists $G', \botev$ such that
    $G = ([] [ \evcast{\botev}{[]} ] [\rec{\overline{l = \rv}, l = [], \overline{l = \rr}} ]) [ G']$.
  \item Exists $F, G', \botev$ such that
    $G = E[G']$ with $E = ([] [F[ [] ] ] [ \evcast{\botev}{[]} ]) $ and
    $G' \neq ([] [ \evcast{\botev'}{[]} ]) [G'']$.
  \item Exists $F, G', \botev$ such that
    $G = E[G']$ with $E = ([] [ \evcast{\botev}{[]} ] [F[ [] ] ]) $.
  \end{enumerate}

  Any $G'$s would produce valid syntactic terms once the context is plugged.
\end{lemma}

\begin{proof}
  By structural induction on the definition of G after expansion of
  all E cases:

  \begin{itemize}
  \item $[]$. Trivial.
  \item $G [ F [ [] ] ]$.
    By cases on the induction hypothesis: First trivial, second
    impossible (Frames of the form $ [] [F[ [] ] ] [ F [ [] ] ]$ never
    produce valid syntactic terms), rest direct from induction hypothesis.
  \item
    $ G [ \rec{\overline{l = \rv}, l = [], \overline{l = \rr}} ]$
    By cases on the induction hypothesis: First exists $G'=[]$ and
    thus (3) holds,
    second impossible, rest direct from induction hypothesis.
  \item $ G [ \evcast{\botev}{[]} ] [ F [ [] ] ]$
    By cases on the induction hypothesis: First exists $G' = []$ and
    thus (6) holds,
    second case exists $\botev = \botev$ and thus (5) holds, rest
    from induction hypothesis.
  \item $G [ \evcast{\botev}{[]} ] [ \rec{\overline{l = \rv}, l = [], \overline{l = \rr}} ]$
    By cases on the induction hypothesis: First exists $G' = []$ and
    thus (4) holds, second case exists $\botev = \botev$ and thus (5) holds.
  \end{itemize}
  
\end{proof}

\begin{lemma}[Outside extension for evaluation contexts].
  \label{lem:outside-extension-G}
  For any evaluation context $G \in \pls{\text{ACtxt}}$
  \begin{itemize}
  \item
    $ [] [ \rec{\overline{l = \rv}, l = [], \overline{l = \rr}} ] [G]$
    is also a valid evaluation context in $\pls{\text{ACtxt}}$.
  \item If $G \neq ([] [ \evcast{\botev'}{[]} ])[G']$, then also
    $ [] [F [ \evcast{\botev}{[]} ] ] [G]$ and
    $ [] [ \evcast{\botev}{[]} ] [G]$ are valid evaluation contexts
    in $\pls{\text{ECtxt}}$.
  \end{itemize}

  Note: since $E$ always has a $G$ in the end, this lemma generalizes
  to evaluation contexts $E$ as well.
\end{lemma}
\begin{proof}
  First case is obvious as we can replace the closing top hole $[]$
  in the
  original term derivation by the new term
  $([] [ \rec{\overline{l = \rv}, l = [], \overline{l = \rr}} ])$.

  For the second case, the  left-most case would be expanded as
  $E = G = []$, we instead do
  $E = G[ \evcast{\botev}{[]} ] = [][ \evcast{\botev}{[]} ]$ .
\end{proof}

\begin{theorem}[Semantic equivalence between Reduction Semantics and
  SOS]
  \label{thm:sem-eq-rs-sos-brr}
  $\pls{\rr} --> \pls{\rr'}$ in $\BRRcsub$, if and only if 
  $\pls{\rr --> \rr'}$ in the SOS semantics for BRR.  Similarly,
  $\pls{\rr} --> \pls{\error}$ in $\BRRcsub$, if and only if 
  $\pls{\rr --> \error}$ in the SOS semantics for BRR. 
\end{theorem}

\begin{proof}
  We begin by proving the only if case ($==>$).

  We proceed by cases on the contextual reduction rule templates.
  For the first two cases, 
  we begin our proof by cases on $E$.  We can then treat each of the
  6 proofs by induction on the structure of G.
  Lemma
  \ref{lem:brr-space-efficient-gc} exposes an alternative induction
  principle on the structure of G, which we choose to use.

  We will group the proof by the cases of the induction principle,
  and deal in each case with the contextual reduction rules.

  \begin{enumerate}
  \item $G = []$. We can use rules [SEv], [SpEf-Ex], [SCe'] and
    [SFx] to build a SOS reduction for each respective case.
  \item $G = [] [F [ [] ] ]$. Contextual reduction cannot happen on
    the cases that rely on notions of reductions, as all terms that
    can take a step (either to another $e$ or to $\error$) are
    non-ascription redexes.  We are only left with the cases for
    evidence composition, for which we can build a derivation using
    rules [SCe] and [SFx'], respectively.
  \item Exists $G'$, such that
    $G = ([] [\rec{\overline{l = \rv}, l = [], \overline{l = \rr}}]) [ G']$.
    By induction, we can build a derivation on $G'[e]$ in each case,
    and we can extend that derivation with a rule [SRec-v] or
    [SRec-x], respectively.
  \item Exists $G', \botev$ such that
    $G = ([] [ \evcast{\botev}{[]} ] [\rec{\overline{l = \rv}, l = [], \overline{l = \rr}}]) [ G']$.
    By induction, we can build a derivation on $G'[e]$ in each case,
    and we can extend that derivation with the sequence of rules
    ([SRec-v], [SStack]) or ([SRec-x], [SStack']), respectively.
    
  \item Exists $F, G', \botev$ such that
    $G = E[G']$ with $E = ([] [F[ [] ]] [ \evcast{\botev}{[]}]) $
    and
    $G' \neq ([] [ \evcast{\botev'}{[]} ]) [G'']$.  By induction, we
    can build a derivation on $G'[e]$ in each case, and we can extend
    that rule with rules [SFe] or [SCx], respectively.
  \item Exists $F, G', \botev$ such that
    $G = E[G']$ with $E = ([] [ \evcast{\botev}{[]} ] [F[ [] ]]) $.

    By the structure of evaluation contexts, two ascriptions cannot
    appear next to each other in a $G$.  In the case of the
    contextual reduction rules
    dealing with notions of reduction, we know these cannot reduce on
    ascriptions so there must be a $G''$ such that
    $G' = ([] [ \evcast{\botev}{[]} ] G''$.  We can thus apply either
    the sequence ([SFe], [SStack]) or ([SCx], [SStack']),
    respectively.
    
  \end{enumerate}

  We can now deal with the if direction ($<==$).

  We proceed by induction on the derivation of the structural
  operational semantics ($\pls{-->}$).
  \begin{itemize}
  \item Case [SEv]. We can build a derivation using $E = []$.
  \item Case [SpEf-Ex]. We can build a derivation using $E = []$.
  \item Case [SRec-v]. By induction hypothesis, we get some $E$, or a
    $G$, depending on the contextual reduction in use. We
    can apply lemma \ref{lem:outside-extension-G} to build the
    derivation with the same structure for [SRec-v]'s conclusion.
  \item Case [SRec-x]. Analogous to [SRec-v].
  \item Case [SCe]. We can build a derivation using
    $E = G = [] [F[ [] ] ]$.
  \item Case [SCe']. We can build a derivation using $E = G = []$.
  \item Case [SFx]. Analogous to [SCe'].
  \item Case [SFx']. Analogous to [SCe].
  \item Case [SFe]. Analogous to [SRec-v]. The extra restriction in
    this rule's premise is already handled by lemma
    \ref{lem:outside-extension-G}.
  \item Case [SCx]. Analogous to [SFe].
  \item Case [SStack]. Analogous to [SFe].
  \item Case [SStack']. Analogous to [SFe].
  \end{itemize}
  
\end{proof}

\subsection{Bisimulation Relation}
\label{apdx:bisimulation-relation}
The bisimulation relation must be in charge of ``giving leeway'' on each
point where both semantices diverge, bearing in mind that, for them to
be bisimilar, we do want them to eventually be reconciled. We propose
the following simulation relation.  The rule of real interest is
[Sim-Casts], which encodes the leeway of changing orders in
evaluation of casts,  The rest of the simulation rules are just there to
complete the structural induction. 

\flushleft\boxed{\pla{e} \approx \pls{e}, \pla{\error} \approx \pls{\error}}
\begin{mathpar}
  \inference[Sim-Casts]{\pla{e_1} \neq \pla{\evcast{\ev'}{e'}} &
    \pls{e_2} \neq \pls{\evcast{\botev'}{e'}} & 
    \pla{e_1} \approx \pls{e_2} &
    \pla{\ev_n} \bottrans{} \! \left( \cdots \bottrans{} \pla{\ev_1}\right) =
    \left(\pls{\botev_m} \bottrans{} \! \cdots \right)\bottrans{} \pls{\botev_1}}{
    \pla{\ev_1 \cdots \evcast{\ev_n}{e_1}} \approx \pls{\botev_1
      \cdots \evcast{\botev_m}{e_2}}}

  \inference[Sim-Err]{}{\pla{\error} \approx \pls{\error}}

  \inference[Sim-U]{}{\pla{u} \approx \pls{u}}

  \inference[Sim-App]{\pla{\evcast{\ev_1}{e_1}} \approx
    \pls{\evcast{\botev_3}{e_3}} &
    \pla{\evcast{\ev_2}{e_2} \approx \pls{\evcast{\botev_4}{e_4}}}}{
    \pla{\cast{\ev_1}{e_1}\;\cast{\ev_2}{e_2}} \approx
    \pls{\cast{\botev_3}{e_3}\;\cast{\botev_4}{e_4}}}

  \inference[Sim-Plus]{\pla{\evcast{\ev_1}{e_1}} \approx
    \pls{\evcast{\botev_3}{e_3}} &
    \pla{\evcast{\ev_2}{e_2} \approx \pls{\evcast{\botev_4}{e_4}}}}{
    \pla{\cast{\ev_1}{e_1}+\cast{\ev_2}{e_2}} \approx
    \pls{\cast{\botev_3}{e_3}+\cast{\botev_4}{e_4}}}

  \inference[Sim-If]{\pla{\evcast{\ev_1}{e_1}} \approx
    \pls{\evcast{\botev_4}{e_4}} &
    \pla{\evcast{\ev_2}{e_2} \approx \pls{\evcast{\botev_5}{e_5}}} &
    \pla{\evcast{\ev_3}{e_3} \approx \pls{\evcast{\botev_6}{e_6}}} 
  }{
    \pla{\<if> \cast{\ev_1}{e_1} \<then> \cast{\ev_2}{e_2} \<else>
      \cast{\ev_3}{e_3}} \approx
    \pls{\<if> \cast{\botev_4}{e_4} \<then> \cast{\botev_5}{e_5}
      \<else> \cast{\botev_6}{e_6}}}

  \inference[Sim-Proj]{\pla{\evcast{\ev_1}{e_1}} \approx
    \pls{\evcast{\botev_2}{e_2}}}{
    \pla{\cast{\ev_1}{e_1}.l} \approx
    \pls{\cast{\botev_2}{e_2}.l}}

  \inference[Sim-Rec]{\pla{e_i} \approx \pls{e'_i}}{
    \pla{\rec{\overline{l_i = e_i}}} \approx \pls{\rec{\overline{l_i = e'_i}}}}

\end{mathpar}

\subsection{Proofs}      

Having set up all the required definitions, we can proceed to prove
the properties we desire: 

\begin{lemma}[Substitution preserves simulation]
  If $\pla{e_1} \approx \pls{e_2}$ and
  $\pla{e_3} \approx \pls{e_4}$, then
  ${\pla{\left[ e_3 \middle/ x \right] e_1} \approx \pls{\left[ e_4 \middle/ x \right] e_2}}$.
\end{lemma}

\begin{proof}
  By simple structural induction over the definition of
  $ \pla{\cdot} \approx \pls{\cdot}$. 
\end{proof}

\begin{lemma}[Notions of reduction preserve simulation]
  If $\pla{e_1} \approx \pls{e_2}$, then if
  $\pla{e_1 \leadsto e'_1}$, then there exists an $\pls{e'_2}$ such
  that $\pls{e_2 \leadsto e'_2}$ and $\pla{e'_1} \approx \pls{e'_2}$.
\end{lemma}

\begin{proof}
  By cases on the applicable notions of reduction, and most follow by
  inversion on the simulation relation.
  The only notion of reduction left is $\beta$.  In this case, 
  theorem holds by substitution preserves simulation lemma.
\end{proof}

\begin{lemma}[Weak Simulation]
  Let $\pla{e_1 --> e'_1}$. If
  $\pla{e_1} \approx \pls{e_2}$, then 
  there exists $\pls{e'_2}$ such that $\pls{e_2 ->^{*} e'_2}$ and
  $\pla{e'_1} \approx \pls{e'_2}$.
\end{lemma}

\begin{proof}
  By structural induction on $\pla{-->}$, and proceed by cases on
  the simulation relation.
  
  \begin{itemize}
  \item Rules [Ev] and [Ex] are trivial by ``Notions of reduction
    preserve simulation'' lemma.
  \item Rule [Fv]. We distinguish two cases:
    \begin{itemize}
    \item $F = []$.  
      \textbf{This is the key case of the proof}.
      The only rules of the simulation relation that apply in this case
      is [Sim-Casts], 
      and thus $\pla{e_1} = \pla{\ev_1 \cdots \evcast{\ev_n}{e'}}$.
      We also know that
      $\pla{e'} \neq \pla{\evcast{\ev'}{e''}}$.  We can then proceed by
      case on the form of the derivation tree $\pla{\mathcal{D}}$:
      \begin{itemize}
      \item If
        $\pla{\mathcal{D}} = \inferrule{[\text{Ev}] \\\\[\text{Fv}]...}{}$ or
        $\pla{\mathcal{D}} = \inferrule{[\text{Ex}] \\\\[\text{Fv}]...}{}$,
        we
        can first apply on the related term Rule
        $\pls{[\text{SCe}]}$ as needed to compose the evidence objects, and
        then, respectively, [SEv] or [SpEf-Ex] and the notions of
        reduction preserves simulation lemma.
      \item If
        $\pla{\mathcal{D}} = \inferrule{[\text{Cv}] \\\\[\text{Fv}]...}{}$,
        we
        can first apply on the related term Rule
        $\pls{[\text{SCe}]}$ as needed to compose the evidence objects, and
        then we
        can rebuild the derivation tree using Rule [Sim-Casts].
      \item If
        $\pla{\mathcal{D}} = \inferrule{[\text{Cx}] \\\\[\text{Fv}]...}{}$, 
        we
        can first apply on the related term Rule
        $\pls{[\text{SCe}]}$ as needed to compose the evidence objects, and
        then apply rule [SFx] or [SFx'] to step to an error.  These values are
        related by Rule [Sim-Err].
      \item If
        $\pla{\mathcal{D}} = \inference{\inferrule{\pla{\mathcal{D}'}\\\\ \inference[\text{Fv}]{\pla{e'' --> e'''}}{\pla{F[\cast{\ev'}{e''}
                ] --> F[\cast{\ev'}{e'''}]}}}{}}{[\text{Fv}]...}$,
        we also know that $\pla{e'} = \pla{F[\cast{\ev'}{e''}]}$ with
        $\pla{F} \neq \pla{[]}$.
        As there exists $\pls{e'_2}$ such that $\pla{e'} \approx \pls{e'_2}$.
        we can use our induction hypothesis to infer that there exists
        $\pls{e''_2}$ such that $\pls{e'_2 -->^{*} e''_2}$ and know
        that $\pla{F[\cast{\ev'}{e'''}]} \approx \pls{e''_2}$.
        We can also then step the original related
        term through as many [SCe'] as needed and then we can use rule [SStack].  We have
        to deal now with two cases:
        \begin{itemize}
        \item If $\pls{e''_2} \neq \pls{\evcast{\botev''}{e''}}$, we
          can rebuild the simulation derivation using rule [Sim-Casts].
        \item If $\pls{e''_2} = \pls{\evcast{\botev''}{e''}}$,
          the simulation $\pla{e'} \approx \pls{e'_2}$ must end with
          a rule [Sim-Casts].  We can then take as many more [SCe]
          steps as needed and rebuild the full derivation using rule [Sim-Casts].
        \end{itemize}
      \item
        If
        $\pla{\mathcal{D}} = \inference{\inferrule{\pla{\mathcal{D}'}\\\\ \inference[\text{Fx}]{\pla{e'' --> \error}}{\pla{F[\cast{\ev'}{e''}
                ] --> \error
              }}}{}}{[\text{Fx}]...}$,
        we also know that $\pla{e'} = \pla{F[\cast{\ev'}{e''}]}$ with
        $\pla{F} \neq \pla{[]}$.
        As there exists $\pls{e'_2}$ such that $\pla{e'} \approx \pls{e'_2}$.
        we can use our induction hypothesis to infer that there exists
        $\pls{e''_2}$ such that $\pls{e'_2 -->^{*} e''_2}$ and know
        that $\pla{\error} \approx \pls{e''_2}$.  Thus also
        $\pls{e''_2} = \pls{\error}$.  We can first take as many steps
        as needed through [SCe'] to reduce first the ascriptions
        coming from [Sim-Casts], and then use [SStack'] to get to
        $\pls{error}$.  

      \item If
        $\pla{\mathcal{D}} = \inference{\inferrule{\pla{\mathcal{D}'}\\\\ \inference[\text{Rec-v}]{\pla{e'' --> e'''}}{\pla{\rec{\overline{l = \rv},{l=e''},\overline{l = \rr}}
                --> \rec{\overline{l = \rv},{l=e'''},\overline{l = \rr}}}}}{}}{[\text{Fv}]...}$.
        Since records can only be related through rule [Sim-Rec], then
        we know there exists $\pls{e'_2}$ such that $\pla{e''} \approx \pls{e'_2}$.
        we can use our induction hypothesis to infer that there exists
        $\pls{e''_2}$ such that $\pls{e'_2 -->^{*} e''_2}$ and know
        that $\pla{e'''} \approx \pls{e''_2}$.
        We can again first step as many times as needed through
        [SCe'], then step the  related
        term through [SRec-v] and then we can use rule [SStack].  We
        can rebuild the relation through rule [Sim-Casts] as we have a
        record that we can re-relate through [Sim-Rec]. 
      \item
        If
        $\pla{\mathcal{D}} = \inference{\inferrule{\pla{\mathcal{D}'}\\\\ \inference[\text{Rec-x}]{\pla{e'' --> \error}}{\pla{\rec{\overline{l = \rv},{l=e''},\overline{l = \rr}}
                --> \error}}}{}}{[\text{Fx}]...}$.

        Since records can only be related through rule [Sim-Rec], then
        we know there exists $\pls{e'_2}$ such that $\pla{e''} \approx \pls{e'_2}$.
        we can use our induction hypothesis to infer that there exists
        $\pls{e''_2}$ such that $\pls{e'_2 -->^{*} e''_2}$ and know
        that $\pla{\error} \approx \pls{e''_2}$, thus $\pls{e''_2} = \pls{\error}$.

        We can take as many [SCe'] steps to reduce first the casts
        from [Sim-Casts], and then take [SRec-x] to step to the error.

      \end{itemize}
    \item $ \pla{F} \neq \pla{[]}$.  Thus the simulation relation on
      terms holds via either [Sim-App], [Sim-If], [Sim-Plus], or
      [Sim-Proj].  In each of these cases we would like
      to appeal directly to induction, but we cannot as the stepping
      terms do not take into consideration the evidence that appears
      in the assumptions of these rules.  Instead we appeal inductively
      to the fact that the next rule must be [Sim-Casts], and then
      we can follow an analogous by-cases analysis to the case where
      $F = []$ at a nested level.
    \end{itemize}

  \item Rule [Fx] analogous to rule [Fv].

  \item Rules [Rec-v] and [Rec-x] can appeal directly to the induction
    hypothesis, and use their respective [SRec-v] and [SRec-x]
    counterparts in the space-efficient domain.

  \item Rules [Cv] and [Cx] follow an analogous analysis to rule [Fv],
    but much simpler a there is no need to account for arbitrary
    nesting of casts.
  \end{itemize}
\end{proof}

\begin{lemma}[Reverse Weak Simulation]
  Let $\pls{e_2 --> e'_2}$. If
  $\pla{e_1} \approx \pls{e_2}$, then 
  there exists $\pla{e'_1}$ such that $\pla{e_1 ->^{*} e'_1}$ and
  $\pla{e'_1} \approx \pls{e'_2}$.
\end{lemma}
\begin{proof}
  By structural induction on $\pls{-->}$, and proceeding by cases on
  the simulation relation.
  \begin{itemize}
  \item Rules [SEv] and [SpEf-Ex] are trivial by the ``Notions of
    reduction preserve simulation'' lemma.
  \item Rule [SCe]. Terms are related via [Sim-App] and use
    [Sim-Casts] on the hole, and can be
    related back using [Sim-Casts] after taking no steps on $\pla{->^{*}}$.
  \item Rule [SCe']. Terms are related via [Sim-Casts], and
    we can simply take no steps on $\pla{->^{*}}$. 
  \item Rule [SFx]. Terms can only be related via
    [Sim-Casts]. By associativity of evidence composition, we can
    take a combination of [Cv] steps and then a final [Cx] step to
    at some point reach an undefined evidence composition and
    produce an $\pla{error}$.
    
  \item Rule [SFx']. In the case for this rule, there is a subterm that
    has $\bot$ as evidence, and that term must be related by
    [Sim-Casts] to a subterm of the related term.

    Since all $\pls{F'}$ are also $\pla{F}$s,
    we will take a combination of [Cv] steps and then a [Cx] step
    to produce a $\pla{\error}$.
  \item Rule [SFe]         \textbf{This is the key case of this proof}.
    We would like to appeal directly to the induction hypothesis,
    but we cannot because the term stepping in our assumption is
    not directly related by the simulation relation.  Just like we
    did on the Weak Simulation proof, we can proceed by cases
    assuming that just as [Sim-App]/[Sim-If]/[Sim-Proj]/[Sim-Plus] will be the last rule,
    [Sim-Casts] will be the second-to-last, and we can then
    reason by cases on previous rules that could apply on the
    derivation of $\pls{->}$.

  \item Rule [SCx]. We can apply rules [Fv] and a final rule [Fx]
    at each step of the reduction we get from applying the
    induction hypothesis.  Error terms are then related by [Sim-Err].
  \item Rule [SStack]. We can apply rule [Fv] at each step of the
    reduction we get from applying the induction hypothesis. Since
    the output terms are also
    related by the simulation relation according to the induction
    hypothesis, we can connect the full terms back together
    via [Sim-Casts].
  \item Rule [SStack']. We can apply rule [Fv] at each step of the reduction we get from applying the
    induction hypothesis.  Given the way we have defined all the
    rules in the simulation relation that have an $\pls{\error}$
    term on the right, and given the presences of the [SFx] and [SFx'] rules (which imply that every redex that reduces to
    something with a $\bot$ evidence will eventually produce an
    $\pls{\error}$), all terms will reduce to error, and thus we
    can use rule [Sim-Err] to build the simulation derivation.

  \item Rules [SRec-v] and [SRec-x] can appeal directly to the induction
    hypothesis, and use their respective [Rec-v] and [Rec-x]
    counterparts in the non-space-efficient domain, as we can either
    rebuild the related terms directly via [Sim-Rec] or [Sim-Err], respectively.

  \end{itemize}
\end{proof}

\begin{theorem}[Weak Bisimulation]
  Relation $\pla{\cdot} \approx \pls{\cdot}$ forms a Weak
  Bisimulation~\cite{sangiorgi_2011} between $\pla{-->}$ and $\pls{-->}$.
\end{theorem}
\begin{proof}
  By Weak Simulation and Reverse Weak Simulation lemmas.
\end{proof}

\section{Bounded records and bounded rows}
\label{apdx:brr}

\subsection{Proofs of Galois Connection}

\begin{lemma}[$\alpha$ is Sound]
  If $\collecting{T}$ is not empty, then
  $\collecting{T} \subseteq \gamma\left(\alpha\left(\collecting{T}\right)\right)$.
\end{lemma}
\begin{proof}
  By mutual induction on the structure of sets induced by the definition of the
  abstraction functions $\alpha$ and $\alpha^M$.

  \boxed{\alpha\text{ cases}}
  \begin{itemize}
  \item   \(\alpha\left(\generate{\Int}\right) = \Int\).
    
    It is the case that \(\generate{\Int} \subseteq \set{\Int}\)
  \item   \(\alpha\left(\generate{\Bool}\right) = \Bool\).
    
    It is the case that \(\generate{\Bool} \subseteq \set{\Bool}\)

  \item \(\alpha\left(\generate{C_1 -> C_2}\right) =  \alpha\left(C_1\right) -> \alpha\left(C_2\right)\).

    Let $T_{a} -> T_{b} \in \generate{C_1 -> C_2}$.  By Induction
    Hypothesis, $T_{a} \in \gamma(\alpha(C_1))$ and
    $T_{b} \in \gamma(\alpha(C_2))$. By the definition of
    the concretization function, then
    $T_a -> T_b \in \gamma(\alpha(\generate{C_1 -> C_2}))$.
  \item
    \(\alpha\left(\generate{\mappings{i}{n}{\lx_i : {C^{\varnothing}}_i}}\right) =
    \left[\mappings{i}{n}{\lx_i : \alpha^M({C^{\varnothing}}_i)}\right]\).

    Let $\left[\mappings{j}{m}{\lx_j : T_j}\right] \in \generate{\mappings{i}{n}{\lx_i : {C^{\varnothing}}_i}}$.

    By mutual induction hypothesis on each label,
    $\left[\mappings{j}{m}{\lx_j : T_j}\right] \in \generate{\mappings{i}{n}{\lx_i : \gamma^M(\alpha^M({C^{\varnothing}}_i))}}$.
    The labels missing in $j$ from $i$ must come from a set
    ${C^{\varnothing}}_i$ such that
    $\varnothing \in {C^{\varnothing}}_i$, and for all
    $C^{\varnothing}$,
    $\varnothing \in C^{\varnothing} => \varnothing \in \gamma^M(\alpha^M({C^{\varnothing}}))$.
    
  \item
    \(\alpha\left(\generate{\?\;\mappings{i}{n}{\lx_i : {C^{\varnothing}}_i}}\right) =
    \left[\mappings{i}{n}{\lx_i : \alpha^M({C^{\varnothing}}_i)}\;\?\right]\).

    Let
    $\left[\mappings{j}{m}{\lx_j : T_j}\right] \in \generate{\?\;\mappings{i}{n}{\lx_i : {C^{\varnothing}}_i}}$. Therefore,
    $\mappings{j}{m}{\lx_j : T_j} = \mappings{k}{p}{\lx_p : T_p} \mappings{h}{q}{\lx_q : T_q}$,
    where
    $\left[\mappings{k}{p}{\lx_k : T_k}\right] \in \generate{\mappings{i}{n}{\lx_i : {C^{\varnothing}}_i}}$
    and
    $\left[\mappings{h}{q}{\lx_h : T_h}\right] \in \generate{\mappings{h}{q}{\lx_h : \Pow(\Type \cup \set{\varnothing})}}$,
    where $\dom(\lx_i) \cap \dom(\lx_h) = \emptyset$.  We follow the
    mutual induction hypothesis on each restricted label $\lx_k$ as
    in the previous case
    to know
    $\left[\mappings{k}{p}{\lx_k : T_k}\right] \in \generate{\mappings{i}{n}{\lx_i : \gamma^M(\alpha^M({C^{\varnothing}}_i))}}$,
    and since the abstraction generates a bounded row, we can
    extract from its concretization the appropriate type for the
    rest of the mappings.

  \item
    \(\alpha\left(C\right) = \? \).

    Trivial since $\gamma\left(\alpha(C)\right) = \Type$.
  \end{itemize}

  \boxed{\alpha^M\text{ cases}}
  \begin{itemize}
  \item \(\alpha^M\left(\set{\varnothing}\right) = \varnothing\).
    It is the case that
    $\set{\varnothing} \subseteq \set{\varnothing}$.

  \item \(\alpha^M\left(C\right) = (\alpha(C))_R\),
    \(\varnothing \not\in C\).

    Let $T \in C$.  By mutual induction hypothesis, $T \in \gamma(\alpha(C))$.
    
  \item
    \(\alpha^M\left(\set{\varnothing} \cup C\right) = (\alpha(C)_O)\).

    If $C$ is empty, we fall back to the previous case.  If $C$ is
    not empty, we case by the kind of elements of the set:

    \begin{itemize}
    \item \(\varnothing \in \set{\varnothing} \cup C\).  By
      definition of $\gamma^M(S_O)$, $\varnothing \in
      \gamma(\alpha(C)_O)$.
    \item \(T \in C\).  By mutual induction hypothesis, $T \in \gamma(\alpha(C))$.
    \end{itemize}
    
  \end{itemize}
\end{proof}
\begin{lemma}[$\alpha$ is Optimal]
  If $\collecting{T}$ is not empty and
  $\collecting{T} \subseteq \gamma\left(\cT\right)$, then
  $\alpha\left(\collecting{T}\right) \sqsubseteq \cT$
\end{lemma}
\begin{proof}
  Remember that $\cT_1 \sqsubseteq \cT_2$ iff
  $\gamma(\cT_1) \subseteq \gamma(\cT_2)$. Then we proceed by
  mutual structural induction on the definitions of $\cT$ and $M$.

  \boxed{S\text{ cases}}
  \begin{itemize}

  \item $\cT = \Int$.

    Only nonempty subset of $\gamma(\cT)$ is $\set{\Int}$, and 
    $\alpha(\set{\Int}) = \Int$.
    
  \item $\cT = \Bool$.

    Only nonempty subset of $\gamma(\cT)$ is $\set{\Bool}$, and
    $\alpha(\set{\Bool}) = \Bool$.

  \item $\cT = \cT_1 -> \cT_2$.
    By induction hypothesis.
    
  \item $\cT = \?$

    Trivial since $\gamma(\?) = \Type$.

  \item $\cT = \brec{\mappings{i}{n}{\lx_i : M_i}}$
    Since $\collecting{T} \subseteq \gamma(\cT)$, 
    $\alpha(\collecting{T}) = \brec{\mappings{j}{m}{\lx_j : M_j}}$,
    where $\dom(\lx_j) \subseteq \dom(\lx_i)$.

    Let $T = \rec{\overline{l_a:T_a}} \in \gamma(\alpha(\collecting{T}))$. The
    label domain for $T$ must be a subset of the domain of $\cT$ and
    for every mapping $\lx_a : T_a$, by definition and induction
    hypothesis, $T_a \in \gamma^M(M_j)$ for some $j$ such that
    $\lx_a = \lx_j$.

  \item $\cT = \brec{\mappings{i}{n}{\lx_i : M_i} \;\?}$

    Depending on the subset taken from the concretization,
    abstraction might generate a bounded row or a bounded record.
    In any case, all types in the concretization of that type would
    be a  subset of those in the original row.

    The interesting case is when we get a smaller set that still
    generates a row.  For any record type in this concretization,
    there is a potential set labels for which their type came from
    the row designator $\?$, but those will all have been in the
    original concretization of the row.  For those mappings in the
    type that now come from declared labels in the bounded row, by
    induction hypothesis, their types will be in the concretization.
    We do not need to worry about the set of absent
    labels, as by shrinking the set on the left, the set of absent
    labels can only grow, which will only impact the concretization
    by making it smaller, as we want to.
  \end{itemize}

  \boxed{M\text{ cases}}
  \begin{itemize}
  \item \(M = \set{\varnothing}\).

    Only nonempty subset of $\gamma^M(M)$ is $\set{\varnothing}$,
    and $\alpha^M(\set{\varnothing}) = \varnothing$.
    
  \item \(M = \cT_R\).

    Since $\varnothing$ is not a member of $\gamma^M(\cT_R)$, we
    can appeal to the mutual induction hypothesis directly.

  \item \(M = \cT_O\).

    Let's consider separately the case we take a subset containing
    $\varnothing$ or not.

    \begin{itemize}
    \item If $\varnothing$ not in the set, then we can appeal to
      the mutual induction hypothesis directly.
    \item If $\varnothing$ in the set, we must check for the rest
      of the set contents.  If the set is the singleton
      $\set{\varnothing}$, $\gamma^M(\cT_O)$ contains it.
      If there is any other elements, we remove $\varnothing$ for
      the set and appeal to the mutual induction hypothesis, while
      later adding $\varnothing$ on both sides of the inequality.
    \end{itemize}
    
  \end{itemize}

\end{proof}
\begin{theorem}[Bounded Rows form a Galois Connection]
\end{theorem}
\begin{proof}
  By Sound and Optimality Lemmas.
\end{proof}

\subsection{ BRR is Gamma-Complete }
In this section we prove the property that BRR is gamma-complete.

\begin{lemma}[Consistent transitivity preserves evidence
  well-formedness]
  \label{lem:brr-trans-wf}
  For every two evidence objects $|- \ev_1 \wf$ and
  $|- \ev_2 \wf$,   if there exists
  $\ev_3 = \ev_1 \trans{<:} \ev_2$, then $|- \ev_3 \wf$.
\end{lemma}
\begin{proof}
  By structural induction over the definition of the
  well-formedness judgment. The only interesting cases are
  those with Bounded Records/Rows in both evidence objects, and the
  definition of $\trans{<:}$ in terms of
  $\gamma^{<:}(\ev_1) \relcomp \gamma^{<:}(\ev_2)$ guarantees
  that a well-formed evidence object is produced.
\end{proof}

\begin{theorem}[Gamma-completeness of BRR]
  \label{thm:brr-gamma-completeness}
  
  For every two evidence objects $|- \ev_1 \wf$ and
  $|- \ev_2 \wf$, 
  \begin{displaymath}
    \gamma^{<:}(\ev_1) \relcomp \gamma^{<:}(\ev_2) =
    \gamma^{<:}(\ev_1 \trans{<:} \ev_2)
  \end{displaymath}
\end{theorem}

\begin{proof}

  By definition of evidence composition,
  $\gamma^{<:}(\ev_1 \trans{<:} \ev_2) = \gamma^{<:}(\alpha^{<:}(\gamma^{<:}(\ev_1) \relcomp \gamma^{<:}(\ev_2)))$.

  By Soundness of the Galois Connection,
  $\gamma^{<:}(\ev_1 \trans{<:} \ev_2) \supseteq \gamma^{<:}(\ev_1) \relcomp \gamma^{<:}(\ev_2)$,
  thus we are only left to prove set containment on the other direction.

  Now, by Lemma~\ref{lem:rc-closed-wrt-gamma}, if
  $\gamma^{<:}(\ev_1 \trans{<:} \ev_2) $ is defined, there exists $\ev_3$ such
  that
  $\gamma^{<:}(\ev_1) \relcomp \gamma^{<:}(\ev_2) = \gamma^{<:}(\ev_3)$, 
  and thus by Optimality of the galois connection,
  $ \gamma^{<:}(\ev_1 \trans{<:} \ev_2) \subseteq (\gamma^{<:}(\ev_1) \relcomp \gamma^{<:}(\ev_2))$.
\end{proof}

\subsection{Definition of Initial Evidence}
We present Initial evidence for BRR as if Bounded Records and Rows
were available to the programmer in Fig.~\ref{fig:initial-evidence-brr}
.  This presentation is useful in
writing an inductive definition of evidence composition, and does not impact the definition of initial
evidence with BRRs for $\GTFLcsub$, as we can apply the same
rules for Gradual Rows and Record, under the condition for building initial evidence that Gradual Records become Bounded
Records where every mapping in the Gradual Record is included and annotated as Required, and Gradual
Rows become Bounded Rows where every mapping in the Gradual Row is
included and annotated and Required.

In combination with the definition of gradual meet for BRR
presented in Fig.~\ref{fig:meet-brr}, we can build  an
algorithmic definition of evidence composition.

\begin{figure}
  \centering
  \begin{small}
    \flushleft\boxed{\Isub|[\cT \csub \cT|] : \GType \times \GType \rightharpoonup \Ev{}}
    \begin{displaymath}
      \begin{array}{rll}
        \Isub|[\cT \csub \cT |] &= \evpr{\cT,\cT} & \cT \in \set{\Int,\Bool,\?} \\
        \Isub|[\cT \csub \? |] &= \evpr{\cT,\cT} & \cT \in \set{\Int,\Bool} \\
        \Isub|[\? \csub \cT |] &= \evpr{\cT,\cT} & \cT \in \set{\Int,\Bool} \\
        \Isub|[\cT_{11} -> \cT_{12} \csub \? |] &= \Isub|[\cT_{11} -> \cT_{12} \csub \? -> \?|] \\
        \Isub|[\? \csub\cT_{21} -> \cT_{22} |] &= \Isub|[\? -> \? \csub \cT_{21} -> \cT_{22}|] \\
        \Isub|[\cT_{11} -> \cT_{12} \csub \cT_{21} -> \cT_{22} |] & = \evpr{\cT'_{11} -> \cT'_{12}, \cT'_{21} -> \cT'_{22}} \\[0.5em]
        & \multicolumn{2}{l}{
          \begin{block}
            \quad \text{
              where }
            \begin{block}
              \Isub|[\cT_{21} \csub \cT_{11} |] = \evpr{\cT'_{21},\cT'_{11}} \\
              \Isub|[\cT_{12} \csub \cT_{22} |] = \evpr{\cT'_{12},\cT'_{22}}
            \end{block}
          \end{block}}\\
        \Isub|[\?   \csub \brec{\mappings{i}{n}{\lx_i : M_i}\;*_2}|] &=
        \Isub|[\brec{\?} \csub \brec{\mappings{i}{n}{\lx_i : M_i}\;*_2}|] \\
        \Isub|[\brec{\mappings{i}{n}{\lx_i : M_i}\;*_1} \csub \?|] &=
        \Isub|[\brec{\mappings{i}{n}{\lx_i : M_i}\;*_1} \csub \brec{\?}|] \\
        \Isub|[\brec{\mappings{i}{n}{\lx_i : M_{i1}}\mappings{j}{m}{\lx_j : M_j} \;*_1} \csub \brec{\mappings{i}{n}{\lx_i : M_{i2}}\mappings{k}{p}{\lx_k : M_k} \;*_2}|] &                                                                                                                  \multicolumn{2}{l}{                                                                                            =\begin{block}\left\langle\brec{\mappings{i}{n}{\lx_i : M'_{i1}}\mappings{j}{m}{\lx_j : M'_{j1}}\mappings{k}{p}{\lx_k : M'_{k1}} \;*_1},\right.\\
            \left.\brec{\mappings{i}{n}{\lx_i : M'_{i2}}\mappings{j}{m}{\lx_j : M'_{j2}}\mappings{k}{p}{\lx_k : M'_{k2}} \;*_4}\right\rangle
          \end{block}}\\
        & \multicolumn{2}{l}{
          \begin{block}
            \text{
              where}\\
            \text{for every }i, \Isub|[M_{i1} \csub M_{i2}|] = \evpr{M'_{i1},M'_{i2}}\\
            \text{for every }j, \Isub|[M_{j} \csub D(*_2)|] = \evpr{M'_{j1},M'_{j2}}\\
            \text{for every
            }k, \Isub|[D(*_1) \csub M_{k}|] = \evpr{M'_{k1},M'_{k2}}\\
            *_4 = \begin{cases}
              \? & \text {if } *_1 = *_2 = \? \\
              \cdot & \text{otherwise}
            \end{cases}\\
            D(\cdot) = \varnothing \\
            D(\?) = \?_O
          \end{block}
        }
      \end{array}
    \end{displaymath}

    \flushleft\boxed{\Isub|[M \csub M|] : \Mapping \times \Mapping \rightharpoonup \Mapping * \Mapping}
    \begin{displaymath}
      \begin{array}{rcl}
        \Isub|[M \csub \varnothing |] &=& \evpr{M,\varnothing} \\
        \Isub|[\left(\cT_1\right)_{*1} \csub \left(\cT_2\right)_R |] &=& \evpr{\left(\cT'_1\right)_R,\left(\cT'_2\right)_R}\\[0.5em]
        && \text{ where } \Isub|[\cT_1 \csub \cT_2|] = \evpr{\cT'_1,\cT'_2}\\[0.5em]
        \Isub|[\left(\cT_1\right)_{*1} \csub \left(\cT_2\right)_O |] &=& \begin{cases}
          \evpr{\Gbox{\cT_1}_{*1},\left(\cT'_2\right)_O}
          & \text{ where
          } \Isub|[\cT_1 \csub \cT_2|] = \evpr{\cT'_1,\cT'_2}\\
          \evpr{\Gbox{\cT_1}_{*1},\varnothing} & \text{ if }
          \Isub|[\cT_1 \csub \cT_2|] \text{ is undefined.}
        \end{cases}
      \end{array}
    \end{displaymath}
  \end{small}%
  \caption{Definition of Initial Evidence for Bounded Rows}
  \label{fig:initial-evidence-brr}
\end{figure}

\begin{figure}
  \begin{small}
    \centering
    \flushleft\boxed{\cT \meet \cT}\quad\textbf{Gradual Meet for BRR}
    \begin{align*}
      \cT_1 \meet \cT_2 &= \cT_2 \meet \cT_1 \\
      \? \meet \? &= \? \\
      \Int \meet \Int &= \Int \\
      \Bool \meet \Bool &= \Bool \\
      \cT \meet \? &= \cT \\
      (\cT_{11} -> \cT_{12}) \meet (\cT_{21} -> \cT_{22}) &= 
      (\cT_{11} \meet \cT_{21}) -> (\cT_{12} \meet \cT_{22}) \\
      \brec{\mappings{i}{n}{\lx_i : M_{i1}}\mappings{j}{m}{\lx_j : M_j}\; *_1} \meet
      \brec{\mappings{i}{n}{\lx_i : M_{i2}}\mappings{k}{p}{\lx_k : M_k}\; *_2}
      &=
      \brec{\mappings{i}{n}{\lx_i : M_{i1} \meet M_{i2}} \mappings{j}{m}{\lx_j : M_j \meet D(*_2)} \mappings{k}{p}{\lx_k : D(*_1) \meet M_k} *_3} \\
      &\quad \text{ where } *_3 = \begin{cases}
        ? & \text{ if } *_1 = *_2 = ?\\
        \cdot & \text{otherwise}
      \end{cases}
      \cT_1 \meet \cT_2 &\;\text{undefined otherwise}
    \end{align*}
    \flushleft\boxed{M \meet M}\quad\textbf{Gradual Meet for BRR Mappings}
    \begin{align*}
      M_1 \meet M_2 &= M_2 \meet M_1 \\
      \varnothing \meet \varnothing &= \varnothing \\
      \varnothing \meet \cT_O &= \varnothing \\
      \left(\cT_1\right)_R \meet \left(\cT_2\right)_{*} &= \left(\cT_1 \meet \cT_2\right)_R \\
      \left(\cT_1\right)_O \meet \left(\cT_2\right)_O
      &= \begin{cases}
        \left(\cT_1 \meet \cT_2\right)_O & \text{ if
        } \cT_1 \meet \cT_2 \text{ defined} \\
        \varnothing & \text { if } \cT_1 \meet \cT_2 \text{ undefined}
      \end{cases}
    \end{align*}
  \end{small}
  \caption{Gradual Meet for BRR}
  \label{fig:meet-brr}
\end{figure}

\subsection{Definition of Consistent Subtype Meet and Join for BRR}
The original statically typed language we are basing our systems
upon includes conditional branching via the $\<if>$ construct.  To
assign appropriate types in the context of branching, the static
system includes the Meet $(\submeet)$ and Join $(\subjoin)$ operations
traversing the subtyping lattice over static types defined in Fig.~\ref{fig:static-subtype-meet-join}:

\begin{figure}
  \begin{small}
    \begin{displaymath}
      \begin{block}
        \boxed{T \subjoin T}\quad \textbf{Static Subtype Join}\\[0.5em]
        \subjoin : \Type \times \Type \rightharpoonup \Type \\[0.5em]
        T_1 \subjoin T_2 = T_2 \subjoin T_1 \\[0.5em]
        \Int \subjoin \Int = \Int \\[0.5em]
        \Bool \subjoin \Bool = \Bool \\[0.5em]
        (T_{11} -> T_{12}) \subjoin (T_{21} -> T_{22}) =\\ 
        \hspace{2cm}
        (T_{11} \submeet T_{21}) -> (T_{12} \subjoin T_{22}) \\[0.5em]
        \brec{\mappings{i}{n}{\lx_i:T_{i1}}\mappings{j}{m}{\lx_j:T_j}} \subjoin \\
        \brec{\mappings{i}{n}{\lx_i:T_{i2}}\mappings{k}{p}{\lx_k:T_k}} =
        \brec{\mappings{i}{n}{\lx_i:T_{i1} \subjoin T_{i2}}}
        \\
      \end{block}
      \qquad
      \begin{block}
        \boxed{T \submeet T}\quad \textbf{Static Subtype Meet}\\[0.5em]
        \submeet : \Type \times \Type \rightharpoonup \Type \\[0.5em]
        T_1 \submeet T_2 = T_2 \submeet T_1 \\[0.5em]
        \Int \submeet \Int = \Int \\[0.5em]
        \Bool \submeet \Bool = \Bool \\[0.5em]
        (T_{11} -> T_{12}) \submeet (T_{21} -> T_{22}) =\\ 
        \hspace{2cm}
        (T_{11} \subjoin T_{21}) -> (T_{12} \submeet T_{22}) \\[0.5em]
        \brec{\mappings{i}{n}{\lx_i:T_{i1}}\mappings{j}{m}{\lx_j:T_j}} \submeet \\
        \brec{\mappings{i}{n}{\lx_i:T_{i2}}\mappings{k}{p}{\lx_k:T_k}} = \\
        \hspace{3em}\brec{\mappings{i}{n}{\lx_i:T_{i1} \submeet T_{i2}}\mappings{j}{m}{\lx_j:T_j}\mappings{k}{p}{\lx_k:T_k}}
        \\
      \end{block}
    \end{displaymath}
  \end{small}
  \caption{Static Subtype Meet and Join}
  \label{fig:static-subtype-meet-join}
\end{figure}

We include the gradual versions of these definitions in the
context of BRR in Fig.~\ref{fig:cs-meet-and-join-brr}.

\begin{figure}
  \centering
  \begin{small}
    \begin{displaymath}
      \begin{block}
        \boxed{\cT \csubjoin \cT}\quad \textbf{Consistent Subtype Join with BRR}\\[0.5em]
        \csubjoin : \GType \times \GType \rightharpoonup \GType \\[0.5em]
        \cT_1 \csubjoin \cT_2 = \cT_2 \csubjoin \cT_1 \\[0.5em]
        \? \csubjoin \? = \? \\[0.5em]
        \Int \csubjoin \Int = \Int \\[0.5em]
        \Int \csubjoin \? = \Int \\[0.5em]
        \Bool \csubjoin \Bool = \Bool \\[0.5em]
        \Bool \csubjoin \? = \Bool \\[0.5em]
        (\cT_{11} -> \cT_{12}) \csubjoin (\cT_{21} -> \cT_{22}) =\\ 
        \hspace{2cm}
        (\cT_{11} \csubmeet \cT_{21}) -> (\cT_{12} \csubjoin \cT_{22}) \\[0.5em]
        (\cT_{11} -> \cT_{12}) \csubjoin \? = 
        (\cT_{11} -> \cT_{12}) \csubjoin (\? -> \?) \\[0.5em]
        \begin{footnotesize}
          \brec{\mappings{i}{n}{\lx_i : M_i}\;*} \csubjoin \? =
          \brec{\mappings{i}{n}{\lx_i : M_i}\;*} \csubjoin \brec{\?}
        \end{footnotesize}\\[0.5em]
        \begin{footnotesize}
          \brec{\mappings{i}{n}{\lx_i : M_{i1}}\mappings{j}{m}{\lx_j: M_j}\;*_1} \csubjoin
          \brec{\mappings{i}{n}{\lx_i : M_{i2}}\mappings{k}{p}{\lx_k : M_k}\; *_2} =
        \end{footnotesize}\\
        \begin{footnotesize}
          \brec{\mappings{i}{n}{\lx_i : M_{i1} \csubjoin M_{i2}}\mappings{j}{m}{\lx_j: M_j \csubjoin D(*_2)} \mappings{k}{p}{\lx_k : D(*_1) \csubjoin M_k}\;*_3}\end{footnotesize}\\
        \hspace{2cm}
        \begin{footnotesize}
          \text{ where } *_3 = \begin{cases}
            ? & \text{ if } *_1 = *_2 = \? \\
            \cdot & \text{ otherwise}
          \end{cases}
        \end{footnotesize}
        \\[0.5em]
        \\ 
        \cT \csubjoin \cT \text{ undefined otherwise} \\[0.5em]
      \end{block}  
      \qquad
      \begin{block}
        \boxed{\cT \csubmeet \cT}\quad \textbf{Consistent Subtype Meet with BRR}\\[0.5em]
        \csubmeet : \GType \times \GType \rightharpoonup \GType \\[0.5em]
        \cT_1 \csubmeet \cT_2 = \cT_2 \csubmeet \cT_1 \\[0.5em]
        \? \csubmeet \? = \? \\[0.5em]
        \Int \csubmeet \Int = \Int \\[0.5em]
        \Int \csubmeet \? = \Int \\[0.5em]
        \Bool \csubmeet \Bool = \Bool \\[0.5em]
        \Bool \csubmeet \? = \Bool \\[0.5em]
        (\cT_{11} -> \cT_{12}) \csubmeet (\cT_{21} -> \cT_{22}) =\\ 
        \hspace{2cm}
        (\cT_{11} \csubjoin \cT_{21}) -> (\cT_{12} \csubmeet \cT_{22}) \\[0.5em]
        (\cT_{11} -> \cT_{12}) \csubmeet \? = 
        (\cT_{11} -> \cT_{12}) \csubmeet (\? -> \?) \\[0.5em]
        \begin{footnotesize}
          \brec{\mappings{i}{n}{\lx_i : M_i}\;*} \csubmeet \? =
          \brec{\mappings{i}{n}{\lx_i : M_i}\;*} \csubmeet \brec{\?}
        \end{footnotesize}\\[0.5em]
        \begin{footnotesize}
          \brec{\mappings{i}{n}{\lx_i : M_{i1}}\mappings{j}{m}{\lx_j: M_j}\;*_1} \csubmeet
          \brec{\mappings{i}{n}{\lx_i : M_{i2}}\mappings{k}{p}{\lx_k : M_k}\; *_2} =
        \end{footnotesize}\\
        \begin{footnotesize}
          \brec{\mappings{i}{n}{\lx_i : M_{i1} \csubmeet M_{i2}}\mappings{j}{m}{\lx_j: M_j \csubmeet D(*_2)} \mappings{k}{p}{\lx_k : D(*_1) \csubmeet M_k}\;*_3}
        \end{footnotesize}\\
        \hspace{2cm}
           \begin{footnotesize}
         \text{ where } *_3 = \begin{cases}
           \cdot & \text{ if } *_1 = *_2 = \cdot \\
           \? & \text{ otherwise}
         \end{cases}
         \end{footnotesize}
         \\[0.5em]
         \\ 
         \cT \csubmeet \cT \text{ undefined otherwise} \\[0.5em]
       \end{block}  
     \end{displaymath}
     \begin{displaymath}
       \begin{block}
         \boxed{M \csubjoin M}\quad \textbf{Consistent Subtype Join with BRR}\\[0.5em]
         \csubjoin : \Mapping \times \Mapping \rightharpoonup \Mapping \\[0.5em]
         M_1 \csubjoin M_2 = M_2 \csubjoin M_1 \\[0.5em]
         \varnothing \csubjoin M= \varnothing \\[0.5em]
         \left(\cT_1\right)_R \csubjoin \left(\cT_2\right)_R = \left(\cT_1 \csubjoin \cT_2\right)_R \\[0.5em]
         \left(\cT_1\right)_O \csubjoin \left(\cT_2\right)_{*} = \begin{cases}
           \left(\cT_1 \csubjoin \cT_2\right)_O & \text{ if }
           \cT_1 \csubjoin \cT_2 \text{ defined} \\
           \varnothing & \text{ if }
           \cT_1 \csubjoin \cT_2 \text{ undefined} \\
           \end{cases}
       \end{block}  
       \qquad
       \begin{block}
         \boxed{M \csubmeet M}\quad \textbf{Consistent Subtype Meet with BRR}\\[0.5em]
         \csubmeet : \Mapping \times \Mapping \rightharpoonup \Mapping \\[0.5em]
         M_1 \csubmeet M_2 = M_2 \csubmeet M_1 \\[0.5em]
         \varnothing \csubmeet M = M \\[0.5em]
         \left(\cT_1\right)_R \csubmeet \left(\cT_2\right)_{*} = \left(\cT_1 \csubmeet \cT_2\right)_R \\[0.5em]
         \left(\cT_1\right)_O \csubmeet \left(\cT_2\right)_{O} = \begin{cases}
           \left(\cT_1 \csubmeet \cT_2\right)_O & \text{ if }
           \cT_1 \csubmeet \cT_2 \text{ defined} \\
           \varnothing & \text{ if }
          \cT_1 \csubmeet \cT_2 \text{ undefined} \\
           \end{cases}
       \end{block}  
     \end{displaymath}
   \end{small}
   \caption{Consistent Subtype Extrema with BRR}
   \label{fig:cs-meet-and-join-brr}
 \end{figure}

 \subsection{Space efficiency theorems}
 \begin{corollary}[Fixed overhead a-la-\citet{herman10space}]
  If  $\pls{e} -->^{*} \pls{e'}$ is bound for overhead, then
  \[\spacep{\size}{\pls{e'}} \leq 3*B(\pls{e})* \spacep{0}{\pls{e'}}\]
\end{corollary}

\begin{proof}
  By induction over the structure of $\pls{e'}$.

  At every level of the AST where
  evidence objects occur, in the worst case, the definition of
  $\spacep{f}{\cdot}$ equals 1 plus recursion over two subterms plus applying
  $f$ over two evidences $\ev_1$ and $\ev_2$.  By the bound for overhead,
  $1 + \size(\ev_1) + \size(\ev_2) \leq 3 * (B(\pls{e}) * (1 + 0 + 0))$, and we
  can apply the induction hypothesis over each of the subterms to get the
  desired bound.
\end{proof}

\begin{theorem}[Properties of $\text{RL}^{+}$]
  If evidence composition is associative and has a bound, then the semantics of $\text{RL}^{+}$ is
  space-efficient and is observationally equivalent to the semantics
  of $\text{RL}$.
\end{theorem}

\begin{proof}
  Theorem~\ref{thm:weak-bis} guarantees observational equivalence in
  the presence of associative evidence composition.  For
  space-efficiency, we must separately prove stack and overhead bounds.
  The stack bound holds by the definition of contextual
  reduction in $\text{RL}^{+}$.

  We prove the overhead bound by structural induction over
  the transitive closure of $\text{RL}^{+}$.
  The reflective case holds if Prop.~\ref{prop:ev-bound} holds.
  For a particular step, we proceed inductively over the definition of the
  reduction relation. We only need to deal with the cases that
  modify the evidences in the program:
  \begin{itemize}
  \item Case
    $\pls{G[\cast{\botev_1}\cast{\botev_2}{e}]} --> \pls{G[\cast{\botev_2 \bottrans{} \botev_1}{e}]}$
    holds if
    $\size(\pls{\botev_2 \bottrans{} \botev_1}) \leq 1 + 2 * \size(B(\pls{e}))$,
    which holds when Prop.~\ref{prop:ev-bound} holds.
  \item Cases for addition and conditionals always
    produce smaller terms independently of the definition of
    composition.
  \item Case for projection holds as $\iproj$ always produces a
    smaller evidence and $\rv_j$ is always smaller than the record
    it belongs to.
  \item Case for functions holds if
    $\ev_2 \trans{} \idom(\ev_1) \leq B (\pls{e})$.  Since
    $\idom$ always produces a smaller or equally sized evidence,
    $\size(\ev_2 \trans{} \idom(\ev_1)) \leq B(\pls{e})$ by Prop.~\ref{prop:ev-bound}.
  \end{itemize}
\end{proof}

\clearpage
\section{Forward Completeness implies Associativity}

The following is a generic proof that forward completeness suffices to ensure
that the ideal abstract counterpart to an associative concrete operator is
itself associative.

\newcommand{\abullet}{\mathbin{\widehat{\bullet}}}

Let:
\begin{itemize}
\item $C \in \oblset{Concrete}$ be a set of concrete objects;
\item $A \in \oblset{Abstract}$ be a set of abstract objects;
\item $\gamma : \oblset{Concrete} -> \oblset{Abstract}$ and
  $\alpha : \oblset{Abstract} -> \oblset{Concrete}$ form a Galois connection;
\item
  $\bullet : \oblset{Concrete} \times \oblset{Concrete} -> \oblset{Concrete}$
  be an associative function on concrete objects;
\item
  $\abullet : \oblset{Abstract} \times \oblset{Abstract} ->
  \oblset{Abstract}$
\item $\abullet$ is forward complete:
  $\gamma(A_1) \bullet \gamma(A_2) = \gamma(A_1 \abullet A_2)$.
\end{itemize}

By properties of Galois connections, we have
$\gamma \circ \alpha \circ \gamma = \gamma$.

\begin{proposition}
  $(A_1 \abullet A_2)  \abullet A_3
  = A_1  \abullet (A_2  \abullet A_3).$
\end{proposition}
\begin{proof}
  The following proof depends on $\abullet$ being the \emph{best abstraction}
  of $\bullet$ with respect to its input and output abstractions
  ($\oblset{Abstract}^2 \dashv \oblset{Concrete}^2$ and
  $\oblset{Abstract} \dashv \oblset{Concrete}$ respectively): Let $\abullet$ be
  defined by $A_1 \abullet A_2 = \alpha(\gamma(A_1) \bullet \gamma(A_2))$.
  \begin{align*}
       & (A_1 \abullet A_2)  \abullet A_3 \\
    =\;& \alpha(\gamma(A_1) \bullet \gamma(A_2))  \abullet A_3
       & \text{(defn of $\abullet$)} \\
    =\;& \alpha(\gamma(A_1 \abullet A_2))  \abullet A_3
       & \text{(forward-completeness)}  \\
    =\;& \alpha(\gamma(\alpha(\gamma(A_1 \abullet A_2)))  \bullet \gamma(A_3))
       & \text{(defn of $\abullet$)} \\
    =\;& \alpha(\gamma(A_1 \abullet A_2))  \bullet \gamma(A_3))
       & \text{($\gamma \circ \alpha \circ \gamma = \gamma$)}  \\
    =\;& \alpha((\gamma(A_1) \bullet \gamma(A_2))  \bullet \gamma(A_3))
       & \text{(forward-completeness)}  \\
    =\;& \alpha(\gamma(A_1) \bullet (\gamma(A_2)  \bullet \gamma(A_3)))
       & \text{(associativity of $\bullet$)}  \\
    =\;& \alpha(\gamma(A_1) \bullet \gamma(A_2 \abullet A_3))
       & \text{(forward-completeness)}  \\
    =\;&  A_1  \abullet (A_2  \abullet A_3).
    & \text{(defn of $\abullet$)} 
  \end{align*}
\end{proof}

\begin{proof}
  The following proof depends on $\alpha,\gamma$ forming a Galois insertion,
  $\alpha \circ \gamma = 1$.
  It does not depend on the particulars of $\abullet$'s definition.
  \begin{align*}
       & (A_1 \abullet A_2)  \abullet A_3 \\
       =\;& \alpha(\gamma((A_1 \abullet A_2)  \abullet A_3))
       & \text{(Galois insertion)} \\
    =\;& \alpha(\gamma(A_1 \abullet A_2)  \bullet \gamma(A_3))
       & \text{(forward-completeness)}  \\
    =\;& \alpha((\gamma(A_1) \bullet \gamma(A_2))  \bullet \gamma(A_3))
       & \text{(forward-completeness)}  \\
    =\;& \alpha(\gamma(A_1) \bullet (\gamma(A_2)  \bullet \gamma(A_3)))
       & \text{(associativity of $\bullet$)}  \\
    =\;& \alpha(\gamma(A_1) \bullet \gamma(A_2 \abullet A_3))
       & \text{(forward-completeness)}  \\
    =\;& \alpha(\gamma(A_1  \abullet (A_2  \abullet A_3)))
       & \text{(forward-completeness)}  \\
    =\;&  A_1  \abullet (A_2  \abullet A_3).
       & \text{(Galois insertion)}
  \end{align*}
  
\end{proof}

\end{document}

%% file: definitions.tex
\usepackage{amsmath}     
\usepackage{amsthm}      
\usepackage{semantic}    
\usepackage{stmaryrd}   
\usepackage{mathpartir}  
\usepackage{color}
\usepackage{braket}      
\usepackage{mathtools}   
\usepackage{thmtools}    
\usepackage{balance}     
\usepackage{xspace}
\usepackage{rotating}
\usepackage[normalem]{ulem} 

\DeclareFontFamily{U}{mathx}{\hyphenchar\font45}
\DeclareFontShape{U}{mathx}{m}{n}{
      <5> <6> <7> <8> <9> <10>
      <10.95> <12> <14.4> <17.28> <20.74> <24.88>
      mathx10
      }{} 
\DeclareSymbolFont{mathx}{U}{mathx}{m}{n}
\DeclareFontSubstitution{U}{mathx}{m}{n}
\DeclareMathAccent{\wideparen}{0}{mathx}{"75}

\usepackage{proof}
\usepackage{rotating}
\usepackage{listings}

\usepackage{mdframed}  
\usepackage{mathtools}

\def~{\nobreakspace{}} 

\def\presuper#1#2%
  {\mathop{}%
   \mathopen{\vphantom{#2}}^{#1}%
   \kern-\scriptspace%
   #2}

\ExplSyntaxOn
\DeclareExpandableDocumentCommand{\IfNoValueOrEmptyTF}{mmm}
 {
  \IfNoValueTF{#1}{#2}
   {
    \tl_if_empty:nTF {#1} {#2} {#3}
   }
 }
\ExplSyntaxOff

\makeatletter
\newcommand{\xMapsto}[2][]{\ext@arrow 0599{\Mapstofill@}{#1}{#2}}
\def\Mapstofill@{\arrowfill@{\Mapstochar\Relbar}\Relbar\Rightarrow}
\makeatother

\usepackage{savesym}
\savesymbol{comment}
\usepackage[final]{changes}
\definechangesauthor[name=FB]{felipe}
\definechangesauthor[name=RG]{ron}

\newenvironment{block}[1][t]
  {\begin{array}[#1]{@{}l@{}}}
  {\end{array}}

\definecolor{lightgray}{gray}{0.90}
\newcommand{\Gbox}[1]{\colorbox{lightgray}{$#1$}}

\newcommand{\derivof}{\mathrel{\blacktriangleright}}

\mathlig{[]}{\square}  
\mathlig{|-s}{\vdash^{\!\forall}}
\mathlig{|-}{\vdash}
\mathlig{|}{\mid} 
\mathlig{@>}{\leadsto} 
\mathlig{[[}{\llbracket}
\mathlig{]]}{\rrbracket}

\reservestyle{\oblang}{\mathsf}
\oblang{if[if\;],then[\;then\;],else[\;else\;],let[let\;],in[\;in\;],%
  inn[in\;],inj}

\newcommand{\Z}{\ensuremath{\mathbb{Z}}} 
\newcommand{\Pow}{\ensuremath{\mathcal{P}}} 

\newcommand{\Int}{\mathsf{Int}}
\newcommand{\Bool}{\mathsf{Bool}}

\newcommand{\oblset}[1]{\textsc{#1}}
\newcommand{\Type}{\oblset{Type}}
\newcommand{\Var}{\oblset{Var}}
\newcommand{\Env}{\oblset{Env}}

\newcommand{\Value}{\oblset{Value}}
\newcommand{\Label}{\oblset{Label}}
\newcommand{\Term}{\oblset{Term}}

\newcommand{\GType}{\oblset{GType}}
\newcommand{\RTerm}{\oblset{RTerm}}
\newcommand{\ECtxt}{\oblset{ECtxt}}
\newcommand{\EvFrame}{\oblset{EvFrame}}
\newcommand{\Decomp}{\oblset{Decomp}}
\newcommand{\Mapping}{\oblset{Mapping}}

\newcommand{\fff}{\textsf{false}\xspace}
\newcommand{\ttt}{\textsf{true}\xspace}

\newcommand{\dom}{\mathit{dom}}
\newcommand{\cod}{\mathit{cod}}
\newcommand{\proj}{\mathit{proj}}
\newcommand{\idom}{\mathit{idom}}
\newcommand{\icod}{\mathit{icod}}
\newcommand{\iproj}{\mathit{iproj}}

\newcommand{\height}{\mathit{height}}
\newcommand{\ldom}{\mathit{dom}_l}
\newcommand{\size}{\mathit{size}}
\newcommand{\nat}{\mathbb{N}}
\newcommand{\spacep}[2]{\mathit{space}_{#1}\left({#2}\right)}
\newcommand{\generator}{\mathscr{C}}
\newcommand{\generate}[1]{\generator\left\llbracket #1 \right\rrbracket}

\newcommand{\brec}[1]{\left[#1\right]}
\newcommand{\brow}[1]{\left[#1,\?\right]}
\newcommand{\missing}[1]{{#1} : {\varnothing}}
\newcommand{\mappings}[3]{\sum\limits_{#1 = \added{1}}^{#2} {#3}}
\newcommand{\bounded}[3]{{#1}
  : \mathop{{#3}_{#2}}}

  \newcommand{\pla}[1]{\textcolor{violet}{#1}}
  \newcommand{\pls}[1]{\textcolor{olive}{#1}}
  \newcommand{\botev}[0]{{\varepsilon_{\bot}}}
  \newcommand{\bottrans}[1]{\fatsemi^{#1}_{\bot}} 
  \newcommand{\relcomp}{\mathbin{;}}

\newcommand{\gprec}{\sqsubseteq} 
\newcommand{\sub}{<:}

\newcommand{\csub}{\lesssim}

\newcommand{\meet}{\sqcap}

\newcommand{\finto}{\stackrel{\text{fin}}{\rightharpoonup}}

\newcommand{\?}{\textsf{\upshape ?}} 
\newcommand{\consistent}[1]{\widetilde{#1}}
\newcommand{\collecting}[1]{\wideparen{#1}}
\newcommand{\cT}{S} 
\newcommand{\clT}{\collecting{T}}

\newcommand{\cdom}{\consistent{\dom}}
\newcommand{\ccod}{\consistent{\cod}}
\newcommand{\cproj}{\consistent{\proj}}

\newcommand{\submeet}{\mathbin{\begin{turn}{-90}$\hspace{-0.8em}<:$\end{turn}}}
\newcommand{\subjoin}{\mathbin{\begin{turn}{90}$\hspace{-0.2em}<:$\end{turn}}}
\newcommand{\csubmeet}{\mathbin{\consistent{\wedge}}}
\newcommand{\csubjoin}{\mathbin{\consistent{\vee}}}

\newcommand{\lx}{\ell} 

\DeclareDocumentCommand{\bul}{ O{\lx_1} O{\lx_2} }{[{#1},{#2}]}

\newcommand{\clx}{{\tilde{\lx}}} 
\newcommand{\cS}{{\consistent{S}}} 
\newcommand{\subl}{\preccurlyeq}

\newcommand{\ljoincore}{\begin{turn}{90}$\hspace{-0.2em}\prec$\end{turn}}
\newcommand{\lmeetcore}{\begin{turn}{270}$\hspace{-0.6em}\prec$\end{turn}}
\newcommand{\ljoin}{\mathbin{\ljoincore}}
\newcommand{\lmeet}{\mathbin{\lmeetcore}}

\newcommand{\interior}[1]{\mathcal{I}_{#1}}
\DeclareDocumentCommand{\interiorf}{ m O{\joinfun}}{\mathcal{I}^{#2}_{#1}}

\newcommand{\Ev}[1]{\oblset{Ev}^{#1}}
\newcommand{\Isub}{\interior{}}

\newcommand{\trans}[1]{\fatsemi^{#1}}
\renewcommand{\merge}[1]{\triangle^{#1}}

\newcommand{\rec}[1]{[#1]}
\newcommand{\VarT}[1]{\oblset{Var}_{#1}}

\newcommand{\TermT}[1]{\oblset{Term}_{#1}}

\newcommand{\pr}[1]{\braket{{#1}}}
\newcommand{\cast}[2]{\evcast{\evpr{#1}}{#2}}
\newcommand{\error}{\textup{\textbf{error}}}

\newcommand{\red}{\longmapsto}
\newcommand{\nred}{~-->~}

\newcommand{\iapp}[1]{\mathbin{@^{#1}}}

\newcommand{\ie}{\emph{i.e.}\xspace}
\newcommand{\cf}{\emph{cf.}\xspace}

\newsavebox{\cTsave}
\savebox{\cTsave}{$\cT$}
\newsavebox{\clTsave}
\savebox{\clTsave}{$\clT$}

\usepackage{upgreek}
\newcommand{\ev}{\varepsilon}

\newcommand{\evcast}[2]{#1#2}
\newcommand{\evlangle}{\langle}
\newcommand{\evrangle}{\rangle}
\newcommand{\evpr}[1]{\braket{#1}}
\newcommand{\eve}{\cast{\ev}{\e}}

\newcommand{\STFLsub}{\text{STFL}_{<:}}
\newcommand{\GTFLcsub}{\text{GTFL}_{\csub}}
\newcommand{\BRRcsub}{\text{BRR}_{\csub}}

\newcommand{\new}[1]{{\color{memory} #1}}

\renewcommand{\cast}[2]{#1#2}

\DeclareDocumentCommand{\reffss}{ m m O{\lxc} }{{\mathsf{ref}}^{#2}~#1}
\DeclareDocumentCommand{\reffs}{ m m O{\clxc} }{{\mathsf{ref}}^{#2}~#1}
\DeclareDocumentCommand{\reffsev}{ m m O{\evl} O{\clxc} }{{{\mathsf{ref}}^{{\color{staticcolor} #2}}_{#3}}~#1}

\DeclareDocumentCommand{\assignev}{ O{\evl} O{\clx,\cS_1} m m}{#3 \overset{#1,\stac{#2}}{:=} #4}

\newcommand{\opt}[1]{}

\newcommand{\lxc}[1][]{{\color{clxccolor} \lx_{c#1}}}

\newcommand{\lxr}[1][]{{\color{dynamic} \lx_{r#1}}}

\newcommand{\clxc}[1][]{{\color{clxccolor} \tilde{\lx}_{c#1}}}
\newcommand{\clxcp}[1][]{{\color{clxccolor} \tilde{\lx}'_{c#1}}}

\DeclareDocumentCommand{\iconf}{ m m O{\clxc} }{\pr{#2,{#3}} \triangleright #1}
\newcommand{\cconf}[1][]{{ \phi_{#1}}}

\DeclareDocumentCommand{\pconf}{ m O{} }{{\cconf[#2]} \triangleright #1}

\DeclareDocumentCommand{\iconfnaive}{ m m O{\clxc} }{#2 \triangleright #1}

\DeclareDocumentCommand{\iconfr}{ m m m O{\clxc} }{\iconf{#2}{#1}[{#4}] ~{| #3}}
\DeclareDocumentCommand{\iconfre}{ m m m O{\clxc} }{#2~{| #3}}

\DeclareDocumentCommand{\cprot}{ O{\evl\clx'} m m m O{\clxcp} O{\evrp}}{\mathsf{prot}^{{\color{staticcolor} #5,#2,#3}}_{#6,#1} #4}
\DeclareDocumentCommand{\cprots}{ O{\evl\clx'} m O{\evrp}}{\mathsf{prot}^{{\color{staticcolor}}}_{#3,#1} #2}
\DeclareDocumentCommand{\logr}{ m m m O{\evr} O{\clxc}}{\braket{\iconf{#2}{\cast{#4}{{\color{dynamic} #1}}}[#5], #3}}
\DeclareDocumentCommand{\logrs}{ m m m O{\evr} O{\clxc}}{\braket{\phiexpand{\cast{#4}{{\color{dynamic} #1}}}{#5}, #2, #3}}
\DeclareDocumentCommand{\logri}{ m m m m O{\evr} O{\clxc} O{}}{\braket{\iconf{#2_{#4}}{\cast{#5[#4]}{#1[#4]}}[#6[#4]], #3[#4]}}
\DeclareDocumentCommand{\logris}{ m m m m O{\evr} O{\clxc} O{}}{\braket{\phiexpand{\cast{#5[#4]}{#1[#4]}}{#6[#4]},#2_{#4}, #3[#4]}}

\DeclareDocumentCommand{\logrp}{ m m m O{\evr} O{\clxc}}{\braket{\pconf{#2}, #3}}
\DeclareDocumentCommand{\logrip}{ m m m m O{\evr} O{\clxc} O{}}{\braket{\pconf{#2_{#4}}[#4], #3[#4]}}

\newcommand{\RawValue}{\oblset{RawValue}}
\mathlig{++}{\mathbin{++}}

\newmdenv[topline=false,bottomline=false,leftline=false]{borderright}

\DeclareDocumentCommand{\store}{ O{} }{ {\IfNoValueTF{#1}{\new{\mu}}{\new{\mu_{#1}}}}}
\DeclareDocumentCommand{\storep}{ O{} }{ {\IfNoValueTF{#1}{\new{\mu'}}{\new{\mu'_{#1}}}}}
\DeclareDocumentCommand{\storepp}{ O{} }{ {\IfNoValueTF{#1}{\new{\mu''}}{\new{\mu''_{#1}}}}}
\DeclareDocumentCommand{\storeppp}{ O{} }{ {\IfNoValueTF{#1}{\new{\mu'''}}{\new{\mu'''_{#1}}}}}

\DeclareDocumentCommand{\pstore}{ O{} }{ {\IfNoValueTF{#1}{{\mu}}{{\mu_{#1}}}}}

\newcommand{\lnreds}[1][\lxr]{{\nred}}
\newcommand{\lreds}[1][\lxr]{~{\red}~}
\newcommand{\evr}[1][]{{\color{dynamic} \ev_{r#1}}}
\newcommand{\evrp}[1][]{{\color{dynamic} \ev'_{r#1}}}

\newcommand{\joinfun}{\ljoin}
\newcommand{\joinfunE}{\ljoin_E}
\newcommand{\meetfun}{\lmeet}
\newcommand{\meetfunE}{\lmeet_E}
\newcommand{\cjoinfun}{\widetilde{\ljoin}}

\DeclareDocumentCommand{\setofjoins}{ m O{\joinfun} }{#2\overline{#1}}
\DeclareDocumentCommand{\setofjoinsE}{ m O{\joinfunE} }{#2\overline{#1}}
\DeclareDocumentCommand{\setofmeets}{ m O{\meetfun} }{#2\overline{#1}}
\DeclareDocumentCommand{\setofmeetsE}{ m O{\meetfunE} }{#2\overline{#1}}

\DeclareDocumentCommand{\setofcjoins}{m O{\cjoinfun}}{\setofjoins{#1}[#2]}

\DeclareDocumentCommand{\funkyset}{ O{\lx} O{\star} }{\setj{#1 \IfNoValueOrEmptyTF{#2}{}{, #2}}}
\DeclareDocumentCommand{\funkysetm}{ O{\lx} O{\star} }{\setm{#1 \IfNoValueOrEmptyTF{#2}{}{, #2}}}

\DeclareDocumentCommand{\mergejoin}{ O{\subl} }{{\merge{#1}_{\tiny\ljoin}}}
\DeclareDocumentCommand{\mergemeet}{ O{\subl} }{\merge{#1}_{\tiny\lmeet}}

\DeclareDocumentCommand{\transjoin}{ O{\subl} }{{\trans{#1}_{\tiny\ljoin}}}
\DeclareDocumentCommand{\transmeet}{ O{\subl} }{{\trans{#1}_{\tiny\lmeet}}}

\newcommand{\setj}[1]{\set{#1}^{\tiny\ljoin}}
\newcommand{\setm}[1]{\set{#1}^{\tiny\lmeet}}

\newcommand{\alphalj}[2]{\alpha'_{\lx}}

\newcommand{\gammalj}[2]{\gamma'_{\lx}}

\newcommand{\opl}[1][]{o_{\lx #1}}

\DeclareDocumentCommand{\interiorfp}{ m O{\opl[1]} O{\opl[2]}}{\mathcal{I}^{#2,#3}_{#1}}

\usepackage{calc}
\usepackage{graphicx}

\mathlig{<=}{\gprec}

